\documentclass[letterpaper,12pt]{article}

\makeatletter
\def\monthname{\ifcase\month\or
January\or February\or March\or April\or May\or June\or
July\or August\or September\or October\or November\or December\fi}
\makeatother

\usepackage[round]{natbib}
\makeatletter
\def\@sect#1#2#3#4#5#6[#7]#8{\ifnum #2>\c@secnumdepth
     \let\@svsec\@empty\else
     \refstepcounter{#1}\edef\@svsec{\csname the#1\endcsname. \hskip 0.4em}\fi
     \@tempskipa #5\relax
      \ifdim \@tempskipa>\z@
        \begingroup #6\relax
          \@hangfrom{\hskip #3\relax\@svsec}{\interlinepenalty \@M #8\par}%
        \endgroup
       \csname #1mark\endcsname{#7}\addcontentsline
         {toc}{#1}{\ifnum #2>\c@secnumdepth \else
                      \protect\numberline{\csname the#1\endcsname}\fi
                    #7}\else
        \def\@svsechd{#6\hskip #3\relax  %
                   \@svsec #8\csname #1mark\endcsname
                      {#7}\addcontentsline
                           {toc}{#1}{\ifnum #2>\c@secnumdepth \else
                             \protect\numberline{\csname the#1\endcsname}\fi
                       #7}}\fi
     \@xsect{#5}}
\makeatother

\makeatletter
\renewcommand{\section}{\@startsection{section}{1}{0mm}{-\baselineskip}{0.5\baselineskip}{\center\normalfont\large\bf}}
\renewcommand{\subsection}{\@startsection{subsection}{2}{0mm}{-\baselineskip}{0.25\baselineskip}{\raggedright\normalfont\normalsize\bf}}
\renewcommand{\subsubsection}{\@startsection{subsubsection}{3}{0mm}{-\baselineskip}{0.05\baselineskip}{\raggedright\normalfont\normalsize\itshape}}
\def\@begintheorem#1#2{\trivlist \item[\hskip \labelsep{\bf #1\ #2:}]\it}
\makeatother

\makeatletter
\def\blfootnote{\xdef\@thefnmark{}\@footnotetext}
\makeatother



\renewcommand{\thesection}{\arabic{section}}
\renewcommand{\thesubsection}{\arabic{section}.\arabic{subsection}}

\usepackage[colorlinks=True,citecolor=black,urlcolor=black,pagebackref=true,backref=false,linkcolor=black,hypertexnames=false]{hyperref}
\usepackage{fullpage}
\usepackage{graphicx} %
\usepackage[english]{babel}
\usepackage{array}

\usepackage[normalem]{ulem}
\usepackage{breakcites}
\usepackage[utf8]{inputenc} %
\usepackage[T1]{fontenc}    %
\usepackage{booktabs}       %
\usepackage{amsfonts}       %
\usepackage{nicefrac}       %
\usepackage{microtype}      %
\usepackage{caption}
\usepackage{subcaption}
\usepackage{float}
\usepackage{tikz}
\usetikzlibrary{positioning}
\usetikzlibrary{arrows.meta}
\usepackage{amsmath,amsthm,amsmath}

\usepackage{amssymb}
\usepackage{mathtools}
\usepackage{soul}
\usepackage{bm}
\usepackage{enumitem}
\usepackage{multirow}
\usepackage{wrapfig}
\usepackage{scalerel}
\allowdisplaybreaks
\usepackage{dsfont}
\usepackage[noend]{algpseudocode}
\usepackage{xcolor}
\usepackage{thmtools,thm-restate}

\newtheorem{remark}{Remark}

\newtheorem{definition}{Definition}

\newtheorem{assumption}{Assumption}
\newtheorem{theorem}{Theorem}
\newtheorem{corollary}{Corollary}
\newtheorem{lemma}{Lemma}
\newtheorem{claim}{Claim}
\newtheorem{proposition}{Proposition}

\newtheorem{example}{Example}

\DeclareMathAlphabet{\mathbsf}{OT1}{cmss}{bx}{n}%
\DeclareMathAlphabet{\mathssf}{OT1}{cmss}{m}{sl}%

\newcommand{\mca}{\texttt{MC}}
\newcommand{\cR}{\mathcal{R}} %
\newcommand{\cN}{\mathcal{N}} %
\newcommand{\cC}{\mathcal{C}} %
\newcommand{\cI}{\mathcal{I}} %
\newcommand{\cP}{\mathcal{P}} %
\newcommand{\Probability}{\mathbb{P}}
\newcommand{\Expectation}{\mathbb{E}}

\newcommand{\Variance}{\mathbb{V}\mathrm{ar}}
\newcommand{\Covariance}{\mathbb{C}\mathrm{ov}}
\newcommand{\subGaussian}{\textrm{subGaussian}}
\newcommand{\subExponential}{\textrm{subExponential}}
\newcommand{\subWeibull}{\textrm{subWeibull}}
\newcommand{\subG}[1][\sigma^2]{\textrm{subGaussian}(#1)}
\newcommand{\subE}[1][\sigma^2]{\textrm{subExponential}(#1)}
\newcommand{\subW}[1][\sigma^2]{\textrm{subWeibull}_\rho(#1)}
\newcommand{\Reals}{\mathbb{R}} %
\newcommand{\real}{\Reals} %

\newcommand{\indep}{\hspace{1mm}{\perp \!\!\! \perp}~}
\newcommand{\notindep}{~{\not\!\perp\!\!\!\perp}~}

\newcommand{\variance}[1][DR]{\mathbb{Z}^{\mathrm{#1}}}
\newcommand{\bias}[1][DR]{\mathbb{X}^{\mathrm{#1}}}
\newcommand{\error}[1][DR]{\mrm{Err}^{\mathrm{#1}}}

\newcommand{\termIPW}[1][a]{\mathbb{T}_{i,j}^{(\mathrm{#1}, \mathrm{IPW})}}
\newcommand{\termDR}[1][a]{\mathbb{T}_{i,j}^{(\mathrm{#1}, \mathrm{DR})}}

\newcommand{\termV}[1][a]{\mathbb{T}_{i,j}^{(\mathrm{#1})}}

\newcommand{\lbar}{{\bar\lambda}}
\newcommand{\biasterm}{\mathbb{X}}
\newcommand{\ATETrue}[1][j]{\mathrm{ATE}_{\cdot,#1}}
\newcommand{\ATEDR}[1][j]{\what{\mathrm{ATE}}{}_{\cdot,#1}^{\,\mathrm{DR}}}
\newcommand{\ATEOI}[1][j]{\what{\mathrm{ATE}}{}_{\cdot,#1}^{\,\mathrm{OI}}}
\newcommand{\ATEIPW}[1][j]{\what{\mathrm{ATE}}{}_{\cdot,#1}^{\,\mathrm{IPW}}}

\newcommand{\fsg}{Finite Sample Guarantees for DR}
\newcommand{\fsgoiipw}{Finite Sample Guarantees for OI and IPW}
\newcommand{\fsgoi}{Finite Sample Guarantees for OI}
\newcommand{\fsgipw}{Finite Sample Guarantees for IPW}
\newcommand{\normality}{Asymptotic Normality for DR}

\newcommand{\sigmax}{\wbar{\sigma}}

\newcommand{\tallwide}{\texttt{TW}}
\newcommand{\tall}{(\trm{tall})}
\newcommand{\wide}{(\trm{wide})}
\newcommand{\miss}{}

\newcommand{\cfReg}{\texttt{Cross-}\allowbreak\texttt{Fitted-}\allowbreak\texttt{Regression}}
\newcommand{\ssSVD}{\texttt{Cross-}\allowbreak\texttt{Fitted-}\allowbreak\texttt{SVD}}
\newcommand{\ssMC}{\texttt{Cross}\allowbreak\texttt{-}\allowbreak\texttt{Fitted}\allowbreak\texttt{-}\allowbreak\texttt{MC}}

\newcommand{\mpo}[1][(a)]{{\Theta}^{#1}}
\newcommand{\mpoz}{\mpo[(0)]}
\newcommand{\mpoo}{\mpo[(1)]}
\newcommand{\mpob}{\mpoz,\mpoo}

\newcommand{\empo}[1][(a)]{{\what{\Theta}}^{#1}}
\newcommand{\bmpo}[1][(a)]{{\wbar{\Theta}}^{#1}}
\newcommand{\empoz}{\empo[(0)]}
\newcommand{\empoo}{\empo[(1)]}

\newcommand{\bmpoo}{\bmpo[(1)]}
\newcommand{\empob}{\empoz,\empoo}

\newcommand{\mta}{{P}}
\newcommand{\emta}{{\what{P}}}

\newcommand{\ta}{{A}}
\newcommand{\oo}{{Y}}

\newcommand{\matR}{R}
\newcommand{\matS}{S}
\newcommand{\matU}{U}
\newcommand{\matV}{V}
\newcommand{\matW}{W}
\newcommand{\matT}{T}
\newcommand{\matH}{H}
\newcommand{\matmask}{{F}}

\newcommand{\ooa}[1][]{{{Y}}^{(a),\mathrm{obs}}_{#1}}
\newcommand{\ooz}[1][]{{{Y}}^{(0),\mathrm{obs}}_{#1}}
\newcommand{\oocontrol}{{Y}^{(0),\mathrm{pre}}}
\newcommand{\ooo}[1][]{{{Y}}^{(1),\mathrm{obs}}_{#1}}

\newcommand{\tooz}[1][]{{\wtil{Y}}^{(0),\mathrm{obs}}_{#1}}
\newcommand{\tooo}[1][]{{\wtil{Y}}^{(1),\mathrm{obs}}_{#1}}

\newcommand{\barz}[1][]{{{Y}}^{(0),\mathrm{full}}_{#1}}
\newcommand{\baro}[1][]{{{Y}}^{(1),\mathrm{full}}_{#1}}

\newcommand{\noisey}[1][a]{{E}^{(#1)}}
\newcommand{\noiseyz}{\noisey[0]}
\newcommand{\noiseyo}{\noisey[1]}
\newcommand{\noiseyb}{\noiseyz,\noiseyo}
\newcommand{\noisea}{{W}}

\newcommand{\Rcol}{\mc{E}}

\newcommand{\RP}{\Rcol\bigparenth{\emta}}
\newcommand{\RTheta}{\Rcol\bigparenth{\empo[]}}

\newcommand{\allones}{\boldsymbol{1}}
\newcommand{\ones}[1][N]{\boldsymbol{1}_{#1}}

\renewcommand{\star}{\,?}

\makeatletter
\newcommand{\ogeneric}[2][0.7]{%
  \vphantom{\oplus}\mathpalette\o@generic{{#1}{#2}}%
}
\newcommand{\o@generic}[2]{\o@@generic#1#2}
\newcommand{\o@@generic}[3]{%
  \begingroup
  \sbox\z@{$\m@th#1\oplus$}%
  \dimen@=\dimexpr\ht\z@+\dp\z@\relax
  \savebox\tw@[\totalheight]{$\m@th#1\bigcirc$}%
  \makebox[\wd\z@]{%
    \ooalign{%
      $#1\vcenter{\hbox{\resizebox{\dimen@}{!}{\usebox\tw@}}}$\cr
      \hidewidth
      $#1\vcenter{\hbox{\resizebox{#2\dimen@}{!}{$#1\vphantom{\oplus}{#3}$}}}$%
      \hidewidth
      \cr
    }%
  }%
  \endgroup
}
\makeatother

\newcommand{\odiv}{\mathrel{\ogeneric[0.4]{\boldsymbol{/}}}}

\newcommand{\Uniform}{\texttt{Uniform}} %

\newcommand{\ld}[1][]{\ell_{\delta#1}}

\newcommand{\lone}[1][]{\ell_{1}}

\newcommand{\m}{m}

\newcommand{\propsssvd}{Guarantees for $\ssSVD$}
\newcommand{\propssmc}{Guarantees for $\ssMC$}

\newcommand{\stkout}[1]{\ifmmode\text{\sout{\ensuremath{#1}}}\else\sout{#1}\fi}
\newcommand{\thetamax}{\theta_{\max}}

\newcommand{\muexpected}[1][]{\mu^{[#1]}_{\cdot,T}}
\newcommand{\vareps}{\varepsilon}

\newcommand{\inprob}{\stackrel{p}{\longrightarrow}}

\newcommand{\E}{\mbb E}

\newcommand{\sless}[1]{\stackrel{#1}{\leq}}
\newcommand{\sgreat}[1]{\stackrel{#1}{\geq}}
\newcommand{\sequal}[1]{\stackrel{#1}{=}}

\newcommand{\normalbrackets}[1]{[ #1 ]}
\newcommand{\brackets}[1]{\left[ #1 \right]}
\newcommand{\bigbrackets}[1]{\big[ #1 \big]}
\newcommand{\Bigbrackets}[1]{\Big[ #1 \Big]}
\newcommand{\biggbrackets}[1]{\bigg[ #1 \bigg]}
\newcommand{\normalparenth}[1]{( #1 )}
\newcommand{\parenth}[1]{\left( #1 \right)}
\newcommand{\bigparenth}[1]{\big( #1 \big)}
\newcommand{\Bigparenth}[1]{\Big( #1 \Big)}
\newcommand{\biggparenth}[1]{\bigg( #1 \bigg)}
\newcommand{\normalbraces}[1]{\{ #1  \}}

\newcommand{\sbraces}[1]{\{ #1 \}}
\newcommand{\bigbraces}[1]{\big\{ #1 \big \}}
\newcommand{\Bigbraces}[1]{\Big\{ #1 \Big \}}

\newcommand{\abs}[1]{\left| #1 \right |}
\newcommand{\normalabs}[1]{| #1  |}
\newcommand{\bigabs}[1]{\big| #1 \big|}
\newcommand{\Bigabs}[1]{\Big| #1 \Big|}
\newcommand{\biggabs}[1]{\bigg| #1 \bigg|}

\newcommand{\floors}[1]{\left\lfloor #1 \right \rfloor}

\newcommand{\qtext}[1]{\quad\text{#1}\quad} 
\newcommand{\stext}[1]{\ \text{#1}\ }

\def\defeq{\triangleq} %
\newcommand{\defn}{\defeq}

\def\norm#1{\big\|{#1}\big\|} %
 \def\snorm#1{\|{#1}\|} %
\newcommand{\twonorm}[1]{\norm{#1}_2} %
\newcommand{\stwonorm}[1]{\snorm{#1}_2} %
\newcommand{\sinfnorm}[1]{\snorm{#1}_{\infty}} %

\newcommand{\matnorm}[1]{|\!| #1 | \! |} %
\newcommand{\maxmatnorm}[1]{\matnorm{#1}_{\max}}
\newcommand{\fronorm}[1]{\matnorm{#1}_{\mrm{F}}} %
\newcommand{\twoinfnorm}[1]{\matnorm{#1}_{2, \infty}} 
\newcommand{\stwoinfnorm}[1]{\snorm{#1}_{2, \infty}} 
\newcommand{\onetwonorm}[1]{\matnorm{#1}_{1, 2}}

\def\what#1{\widehat{#1}}

\def\mbb#1{\mathbb{#1}}
\def\mc#1{\mathcal{#1}}
\def\mrm#1{\mathrm{#1}}
\def\trm#1{\textrm{#1}}

\def\til#1{\widetilde{#1}}
\def\wtil#1{\til{#1}}
\def\wbar#1{\overline{#1}}

\def\balign#1\ealign{\begin{align}#1\end{align}}
\def\baligns#1\ealigns{\begin{align*}#1\end{align*}}
\def\balignat#1\ealign{\begin{alignat}#1\end{alignat}}
\def\balignats#1\ealigns{\begin{alignat*}#1\end{alignat*}}
\def\bitemize#1\eitemize{\begin{itemize}#1\end{itemize}}
\def\benumerate#1\eenumerate{\begin{enumerate}#1\end{enumerate}}

\newenvironment{talign*}
 {\let\displaystyle\textstyle\csname align*\endcsname}
 {\endalign}
\newenvironment{talign}
 {\let\displaystyle\textstyle\csname align\endcsname}
 {\endalign}

\def\balignst#1\ealignst{\begin{talign*}#1\end{talign*}}
\def\balignt#1\ealignt{\begin{talign}#1\end{talign}}

\graphicspath{{../figures/},{figs/}}
\usepackage[capitalize,noabbrev]{cleveref} 
\usepackage{autonum}
\newtheorem{myproposition}{Proposition}

\newtheorem{mycorollary}{Corollary}

\newtheorem{myclaim}{Claim}

\newtheorem{mylemma}{Lemma}

\newtheorem{mydefinition}{Definition}

\newtheorem{myremark}{Remark}

\newtheorem{myassumption}{Assumption}

\newtheorem{myexample}{Example}

\crefname{appendix}{Appendix}{Appendices}
\crefname{equation}{Eq.}{Eqs.}
\crefname{lemma}{Lemma}{Lemmas}
\crefname{theorem}{Theorem}{Theorems}
\crefname{Corollary}{Corollary}{Corollaries}
\crefname{Claim}{Claim}{Claims}
\crefname{algorithm}{Algorithm}{Algorithms}
\crefname{example}{Example}{Examples}

\crefname{section}{Section}{Sections}
\crefname{table}{Table}{Tables}
\crefname{remark}{Remark}{Remarks}
\crefname{algorithm}{Algorithm}{Algorithms}
\crefname{definition}{Definition}{Definitions}
\crefname{Proposition}{Proposition}{Propositions}
\crefname{myremark}{Remark}{Remarks}
\crefname{mylemma}{Lemma}{Lemmas}
\crefname{myexample}{Example}{Examples}
\crefname{mydefinition}{Definition}{Definitions}
\crefname{myproposition}{Proposition}{Propositions}
\crefname{mycorollary}{Corollary}{Corollaries}
\crefname{myclaim}{Claim}{Claims}
\crefname{myassumption}{Assumption}{Assumptions}
\crefname{figure}{Figure}{Figures}
\crefname{enumi}{}{}
\crefname{name}{}{}

\begin{document}

\setcounter{page}{1}
\pagestyle{plain}
\vskip 80pt
\centerline{\Large\bf Doubly Robust Inference in Causal Latent Factor Models}
  \begin{center}%
    \vskip 10pt
     \blfootnote{\hspace*{-0.30in} %
    Alberto Abadie, Department of Economics, MIT, abadie@mit.edu.
    Anish Agarwal, Department of Industrial Engineering and Operations Research, Columbia University, aa5194@columbia.edu.
    Raaz Dwivedi, Department of Operations Research and Information Engineering, Cornell Tech, dwivedi@cornell.edu. 
    Abhin Shah, Department of Electrical Engineering and Computer Science, MIT, abhin@mit.edu. We are grateful to Haruki Kono, Guido Imbens, James Robins, Stefan Wager, and seminar participants at Columbia, MIT, the Online Causal Inference Seminar, and Stanford for helpful comments and discussion.
   }
    {\large
     \lineskip .5em%
     \begin{tabular}[t]{ccc}%
        Alberto Abadie&&Anish Agarwal\\[.2ex]
       MIT&&Columbia\\[1ex]
       Raaz Dwivedi&&Abhin Shah\\[.2ex]
       Cornell Tech&&MIT\\
    \end{tabular}
      \par}%
      \vskip 1em%
      {\large \today } \par%
       \vskip 1em%
  \end{center}\par

\bigskip
\begin{center}\normalsize\bf\text{Abstract}
 \end{center}\advance\leftmargini -.8em\begin{quote}\normalsize
\noindent  
This article introduces a new estimator of average treatment effects under unobserved confounding in modern data-rich environments featuring large numbers of units and outcomes. The proposed estimator is doubly robust, combining outcome imputation, inverse probability weighting, and a novel cross-fitting procedure for matrix completion. We derive finite-sample and asymptotic guarantees, and show that the error of the new estimator converges to a mean-zero Gaussian distribution at a parametric rate. Simulation results demonstrate the relevance of the formal properties of the estimators analyzed in this article.
\end{quote}
\bigskip
\def\submission{arxiv}
\def\arxiv{arxiv}
\section{Introduction}
\label{section_introduction}

This article presents a novel framework for the estimation of average treatment effects in modern data-rich environments in the presence of unobserved confounding. We define modern data-rich environments as those featuring many outcome measurements across a wide range of units. Our interest in data-rich environments stems from the emergence of digital platforms (e.g., internet retailers, social media companies, and ride-sharing companies), electronic medical records systems, IoT devices, and other real-time digitized data systems, which gather economic and social behavior data with unprecedented scope and granularity.

Take the example of an internet retailer. The platform collects not only information on purchases of many customers across many products or product categories, but also on glance views, impressions, conversions, engagement metrics, navigation paths, shipping choices, payment methods, returns, reviews, and more. While some variables, such as geo-location and type of device or browser, can be safely treated as pre-determined relative to the platform's treatments (advertisements, discounts, web-page design, etc.), most are outcomes affected by the treatments, latent customer preferences, and unobserved product features. We leverage the availability of many outcome measures in modern data-rich environments to estimate average treatment effects in the presence of unobserved confounding. The core identification concept is that if each element of a high-dimensional outcome vector is influenced by a common low-dimensional vector of unobserved confounders, it becomes possible to remove the influence of the confounders and identify treatment effects.

Two primary approaches to the estimation of treatment effects are outcome-based and assignment-based methods. Consider again the example of an internet-retail platform where customers interact with various product categories. For each consumer-category pair, the platform makes decisions to either offer a discount or not, and records whether the consumer purchased a product in the category. Outcome-based methods operate by imputing the missing potential outcomes for each consumer-product category pair. This process involves predicting whether a consumer, who received a discount, would have made the purchase without the discount (i.e., the potential outcome without discount), and conversely, if a consumer who did not receive the discount would have purchased the product had they received the discount (i.e., the potential outcome with discount). In contrast, assignment-based methods estimate the probabilities of consumers receiving discounts in each product category and adjust for missing potential outcomes by weighting observed outcomes inversely to the probability of missingness.

A substantial body of literature has explored outcome-based methods, particularly in settings where all confounding factors are measured \citep[see, e.g.,][among many others]{cochran1968effectiveness,rosenbaum1983central,angrist1998estimating,abadie2006large}. 
Imputing potential outcomes in the presence of unobserved confounders poses a more complex challenge. In this context, a commonly adopted framework is the synthetic control method and its variants \citep[see, e.g.,][]{abadie2003economic,abadie2010synthetic,cattaneo2021prediction,arkhangelsky2021synthetic}. An alternative but related approach to outcome imputation under unobserved confounding is the latent factor framework \citep{bai2002determining,bai2009panel,xiong2023large}, wherein each element of the large-dimensional outcome vector is influenced by the same low-dimensional vector of unobserved confounders. 
Matrix completion methods  \citep[see, e.g.,][]{chatterjee2015matrix, athey2021matrix, bai2021matrix, dwivedi2022counterfactual, agarwal2023causal} which have found widespread applications in recommendation systems and panel data models, are closely related to latent factor models. 
Similarly, existing assignment-based procedures to estimate average treatment effects rely on the assumption of no unmeasured confounding \citep[see, e.g.,][]{robins2000marginal, hirano2003efficient, wooldridge2007inverse}, common trends restrictions \citep{abadie2005}, or the availability of an instrumental variable \citep{abadie2003semiparametric,tymon2024kappa}. 

In this article, we propose a doubly-robust estimator \citep[see][]{robins1994estimation, bang2005doubly, chernozhukov2018double} of average treatment effects in the presence of unobserved confounding. This estimator leverages information on both the outcome process and the treatment assignment mechanism under a latent factor framework. It combines outcome imputation and inverse probability weighting with a new cross-fitting approach for matrix completion. We show that the proposed doubly-robust estimator has better finite-sample guarantees than alternative outcome-based and assignment-based estimators. Furthermore, the doubly-robust estimator is approximately Gaussian, asymptotically unbiased, and converges at a parametric rate, under provably valid error rates for matrix completion, irrespective of other properties of the matrix completion algorithm used for estimation.

To our knowledge, this is the first article that leverages latent structures in both the assignment and the outcome processes to obtain a doubly-robust estimator of average treatment effects in the presence of unobserved confounding. \cite{arkhangelsky2022doubly} study doubly-robust identification with longitudinal data under the assumption that conditioning of a function of the treatment assignments over time (e.g., the fraction of times an individual is exposed to treatment) is enough to remove confounding. \cite{athey2021matrix}, \cite{bai2021matrix}, \cite{dwivedi2022counterfactual}, \cite{agarwal2023causal}, and \cite{xiong2023large} propose estimators that apply matrix completion techniques to impute potential outcomes. Although these studies utilize low-rank restrictions in the outcome process, they do not investigate the possibility of similar latent structures in the treatment assignment process. Our article addresses this question, and demonstrate substantial benefits from incorporating knowledge about the structure of the assignment mechanism. 
\medskip

\noindent {\bf Terminology and notation.} 
For any real number $b\in \Reals$, $\floors{b}$ is the greatest integer less than or equal to $b$. For any positive integer $b$, $[b]$ denotes the set of integers from $1$ to $b$, i.e., $[b] \defn \normalbraces{1,\cdots, b}$. We use $c$ to denote any generic universal constant, whose value may change between instances. For any $c > 0$, $\m(c)=\max\normalbraces{c, \sqrt{c}}$ and $\ell_c=\log(2/c)$. For any two deterministic sequences $a_n$ and $b_n$ where $b_n$ is positive, $a_n = O(b_n)$ means that there exist a finite $c > 0$ and a finite $n_0 > 0$ such that $\normalabs{a_n} \leq c\,b_n$ for all $n \geq n_0$. Similarly, $a_n = o(b_n)$ means that for every $c > 0$, there exists a finite $n_0 > 0$ such that $\normalabs{a_n} < c\,b_n$ for all $n \geq n_0$. Further, $a_n = \Omega(b_n)$ means that there exist a finite $c > 0$ and a finite $n_0 > 0$ such that $\normalabs{a_n} \geq c\,b_n$ for all $n \geq n_0$. For a sequence of random variables, $x_n = O_p(1)$ means that the sequence $\normalabs{x_n}$ is stochastically bounded, i.e., for every $\varepsilon > 0$, there exists a finite $\delta >0$ and a finite  $n_0 > 0$ such that $\Probability\bigparenth{\normalabs{ x_n } > \delta } < \varepsilon$ for all $n \geq n_0$. Similarly, $x_n = o_p(1)$ means that the sequence $\normalabs{x_n}$ converges to zero in probability, i.e., for every $\varepsilon > 0$ and $\delta >0$, there exists a finite  $n_0 > 0$ such that $\Probability\bigparenth{\normalabs{ x_n } > \delta } < \varepsilon$ for all $n \geq n_0$. For sequences of random variables $x_n$ and $b_n$, $x_n = O_p(b_n)$ means $x_n = \wbar{x}_n b_n$ where the sequence $\wbar{x}_n = O_p(1)$. Likewise, $x_n = o_p(b_n)$ means $x_n = \wbar{x}_n b_n$ where the sequence $\wbar{x}_n = o_p(1)$.

A mean-zero random variable $x$ is $\subGaussian$ if there exists some $b >0$ such that $\Expectation[\exp(sx)] \leq \exp(b^2s^2/2)$ for all $s\in\Reals$. Then, the $\subGaussian$ norm of $x$ is given by $\snorm{x}_{\psi_2} = \inf\normalbraces{t > 0: \Expectation[\exp(x^2/t^2)] \leq 2}$. A mean-zero random variable $x$ is $\subExponential$ if there exist some $b_1, b_2 >0$ such that $\Expectation[\exp(sx)] \leq \exp(b_1^2s^2/2)$ for all $-1/b_2 < s < 1/b_2$. Then, the $\subExponential$ norm of $x$ is given by $\snorm{x}_{\psi_1} = \inf\normalbraces{t > 0: \Expectation[\exp(\normalabs{x}/t)] \leq 2}$. $\Uniform(a,b)$ denotes the uniform distribution over the interval $[a,b]$ for $a,b \in \Reals$ such that $a<b$. $\cN(\mu,\sigma^2)$ denotes the Gaussian distribution with mean $\mu$ and variance $\sigma^2$. 

For a vector $u \in \Reals^n$, we  denote its $t^{th}$ coordinate by $u_t$ and its $2$-norm $\stwonorm{u}$. For a matrix $U \in \Reals^{n_1 \times n_2}$, we denote the element in $i^{th}$ row and $j^{th}$ column by $u_{i,j}$, the $i^{th}$ row by $U_{i, \cdot}$, the $j^{th}$ column by $U_{\cdot, j}$, the largest eigenvalue by $\lambda_{\max}(U)$, and the smallest by $\lambda_{\min}(U)$. Given a set of indices $\cR \subseteq [n_1]$ and $\cC \subseteq [n_2]$,  $U_{\cI} \in  \real^{\normalabs{\cR} \times \normalabs{\cC}}$ is a sub-matrix of $U$ corresponding to the entries in $\cI \defn \cR \times \cC$, and $U_{-\cI}= \{u_{i,j} : (i,j)\in\{[n_1]\times [n_2]\}\setminus \cI\}$. Further, we denote the Frobenius norm by $\fronorm{U} \defn  \bigparenth{\sum_{i \in [n_1], j \in [n_2]} u_{i,j}^2}^{1/2}$, the $(1,2)$ operator norm by $\onetwonorm{U} \defn \max_{j \in [n_2]} \bigparenth{\sum_{i \in [n_1]} u_{i,j}^2}^{1/2}$, the $(2,\infty)$ operator norm by $\stwoinfnorm{U}\defn \max_{i \in [n_1]} \bigparenth{\sum_{j \in [n_2]} u_{i,j}^2}^{1/2}$, and the maximum norm by $\maxmatnorm{U} \defn \max_{i \in [n_1], j \in [n_2]}  \normalabs{u_{i,j}}$. Given two matrices $U, V \in \Reals^{n_1 \times n_2}$, the operators $\odot$ and $\odiv$ denote element-wise multiplication and division, respectively, i.e., $t_{i,j} = u_{i,j} \cdot v_{i,j}$ when $T = U \odot V$, and $t_{i,j} = u_{i,j} / v_{i,j}$ when $T = U \odiv V$. When $V$ is a binary matrix, i.e., $V \in \normalbraces{0,1}^{n_1 \times n_2}$, the operator $\otimes$ is defined such that $t_{i,j} = u_{i,j}$ if $v_{i,j} = 1$ and  $t_{i,j} = \star$ if $v_{i,j} = 0$ for $T = U \otimes V$. Given two matrices $U \in \Reals^{n_1 \times n_2}$ and $V \in \Reals^{n_1 \times n_3}$, the operator $*$ denotes the (transposed column-wise) Khatri-Rao product of $U$ and $V$, i.e., $T = U * V \in \Reals^{n_1 \times n_2 n_3}$ such that $t_{i,j} = u_{i,j-n_2 \bar{j}} \cdot v_{i,1+\bar{j}}$ where $\bar{j} = \floors{(j-1)/n_2}$. For random objects $U$ and $V$, $U \indep V$ means that $U$ is independent of $V$. 

\section{Setup}\label{sec_prob_formulation}

Consider a setting with $N$ units and $M$ measurements per unit. For each unit-measurement pair $i \in [N] \stext{and} j \in [M]$, we observe a treatment assignment $a_{i,j} \in \normalbraces{0,1}$ and the value of the outcome $y_{i,j} \in \Reals$. Although our results can be easily generalized to multi-ary treatments, for the ease of exposition, we focus on binary treatments.

We operate within the Neyman-Rubin potential outcomes framework and denote the potential outcome for unit $i \in [N]$ and measurement $j \in [M]$ under treatment $a \in \normalbraces{0,1}$ by $y_{i,j}^{(a)} \in \Reals$. 
A no-spillover assumption is implicit in the notation, i.e., the potential outcome $y_{i,j}^{(a)}$ does not depend on the treatment assignment for any other unit-measurement pair. 
In the context of online retail data, the assumption of no spillovers across measurements is justified if the cross-elasticity of demand across product categories, $j$, is low.
Our framework allows for the possibility that the same treatment affects multiple outcomes (e.g., $a_{i,j}=a_{i,j'}$ with probability one, for some $j$ and $j'$ in $[M]$). 
Realized outcomes, $y_{i,j}$, depend on potential outcomes and treatment assignments,
\begin{align}
\label{eq_consistency}
    y_{i,j} & = y_{i,j}^{(0)} (1-a_{i,j}) + y_{i,j}^{(1)} a_{i,j}, 
\end{align}
for all $i \in [N] \stext{and} j \in [M]$. \cref{extension_delayed} 
and \cref{sec:app_dynamic} 
extend the framework proposed in this article to a panel data setting with lagged treatment effects. 

\subsection{Sources of stochastic variation}  

In the setup of this article, each unit $i \in [N]$ is characterized by a set of unknown parameters, $\{(\theta_{i,j}^{(0)},\theta_{i,j}^{(1)}, p_{i,j})\in \mathbb R^2\times [0, 1] \}_{j\in [M]}$, which we treat as fixed. Potential outcomes and treatment assignments are generated as follows: for all $i \in [N], j \in [M]$, and $a \in \normalbraces{0,1}$, 
\begin{gather}\label{eq_outcome_model}
y_{i,j}^{(a)} = \theta_{i,j}^{(a)} + \varepsilon_{i,j}^{(a)}\\
\shortintertext{and}
\label{eq_treatment_model}
    a_{i,j} = p_{i,j} + \eta_{i,j},
\end{gather}
where $\varepsilon_{i,j}^{(a)}$ and $\eta_{i,j}$ are mean-zero random variables, and 
\begin{align}\label{eq_eta_defn}
\eta_{i,j} = \begin{cases}
    - p_{i,j} \qtext{with probability} 1 - p_{i,j}\\
    1 - p_{i,j} \qtext{with probability} p_{i,j}.
\end{cases}
\end{align} 
It follows that $\theta_{i,j}^{(a)}$ is the mean of the potential outcome $y_{i,j}^{(a)}$, and $p_{i,j}$ is the unknown assignment probability or latent propensity score.
The matrices $\mpoz \defn \normalbraces{\theta_{i,j}^{(0)}}_{i \in [N], j \in [M]}$, $\mpoo \defn \normalbraces{\theta_{i,j}^{(1)}}_{i \in [N], j \in [M]}$, and $\mta \defn \normalbraces{p_{i,j}}_{i \in [N], j \in [M]}$
collect mean potential outcomes and assignment probabilities. Then, the matrices $\noiseyz\defn \normalbraces{\varepsilon_{i,j}^{(0)}}_{i \in [N], j \in [M]}, \noiseyo\defn \normalbraces{\varepsilon_{i,j}^{(1)}}_{i \in [N], j \in [M]}$, and $ \noisea\defn \normalbraces{\eta_{i,j}}_{i \in [N], j \in [M]}$ capture all sources of randomness in potential outcomes and treatment assignments.

Our setup allows $\mpoz, \mpoo$ to be arbitrarily associated with $\mta$, inducing unobserved confounding. The assumptions in \cref{sec_main_results} imply that $\mpoz, \mpoo$, and $\mta$ include all confounding factors, and require $(\varepsilon^{(0)}_{i, j}, \varepsilon^{(1)}_{i, j}) \indep \eta_{i, j}$ for every $i \in [N]$ and $j \in [M]$. 

\subsection{Target causal estimand} 

For any given measurement $j \in [M]$, we aim to estimate the effect of the treatment averaged over all units,
\begin{gather}
\ATETrue \defn \mu^{(1)}_{\cdot,j} - \mu^{(0)}_{\cdot,j} \label{eq_ate_parameter_combined} \shortintertext{where} 
\mu^{(a)}_{\cdot,j} \defn \frac{1}{N}\sum_{i \in [N]} \theta_{i,j}^{(a)}.
\end{gather}
$\ATETrue$ akin to the conditional average treatment effect of \cite{abadie2006large}, but based on the latent means, $\theta_{i,j}^{(a)}$, in \cref{eq_outcome_model} rather than on conditional means that depend on observed covariates only. 
It is straightforward to adapt the methods in this article to the estimation of alternative parameters, like the average treatment effect across measurements for each unit $i$, or the estimation of treatment effects over a subset of the units, $S \subset [N]$.

\section{Estimation}\label{sec_estimators}

In this section, we propose a procedure that uses the treatment assignment matrix $\ta$ and the observed outcomes matrix $\oo$ to estimate $\ATETrue$, where
\begin{align}
    \oo \defn \normalbraces{y_{i,j}}_{i \in [N], j \in [M]} \qtext{and} \ta \defn \normalbraces{a_{i,j}}_{i \in [N], j \in [M]}.
\end{align}
The estimator proposed in this section leverages matrix completion as a key subroutine. We start the section with a brief overview of matrix completion methods.
\subsection{Matrix completion: A primer}\label{subsec_mc_primer}
Consider a matrix of parameters $\matT \in \Reals^{N \times M}$. While $\matT$ is unobserved, we observe the matrix
$\matS \in \{\Reals \cup \{\star\}\}^{N \times M}$ where $\star$ denotes a missing value.  The relationship between $\matS$ and $\matT$ is given by 
\begin{align}\label{eq_mc_main}
  \matS = (\matT + \matH) \otimes \matmask.
\end{align}
Here, $\matH \in \Reals^{N \times M}$ is a noise matrix, and $\matmask \in \{0, 1\}^{N \times M}$ is a masking matrix with ones for the recorded entries of $\matS$ and zeros for the missing entries.

A matrix completion algorithm, denoted by $\mca$, takes the $\matS$ as its input, and returns an estimate of $\matT$, which we denote by $\what {\matT}$ or $\mca(\matS)$.
In other words, $\mca$ produces an estimate of a matrix from noisy observations of a subset of all the elements of the matrix.

The matrix completion literature is rich with algorithms $\mca$ that provide error guarantees, namely bounds on $\snorm{\mca(\matS)-\matT}$ for a suitably chosen norm/metric $\snorm{\cdot}$, under a variety of assumptions on the triplet $(\matT, \matH, \matmask)$. 
Typical assumptions are $(i)$ $\matT$ is low-rank, $(ii)$ the entries of $\matH$ are independent, mean-zero and sub-Gaussian random variables, and $(iii)$ the entries of $\matmask$ are independent Bernoulli random variables. 
Though matrix completion is commonly associated with the imputation of missing values, a typically underappreciated aspect is that it also denoises the observed matrix.
Even when each entry of $\matS$ is observed, $\mca(\matS)$ subtracts the effects of $\matH$ from $\matS$, i.e., it performs matrix denoising. \cite{nguyen2019low} provide a survey of various matrix completion algorithms.

\begin{figure}[t]
    \centering
    \hspace*{-0.1cm}\begin{tabular}{c@{\hskip 0.6cm}c@{\hskip 0.6cm}c@{\hskip 0.6cm}c}
    \includegraphics[trim={3cm 7cm 2.5cm 2.75cm}, width=0.22\linewidth,clip]{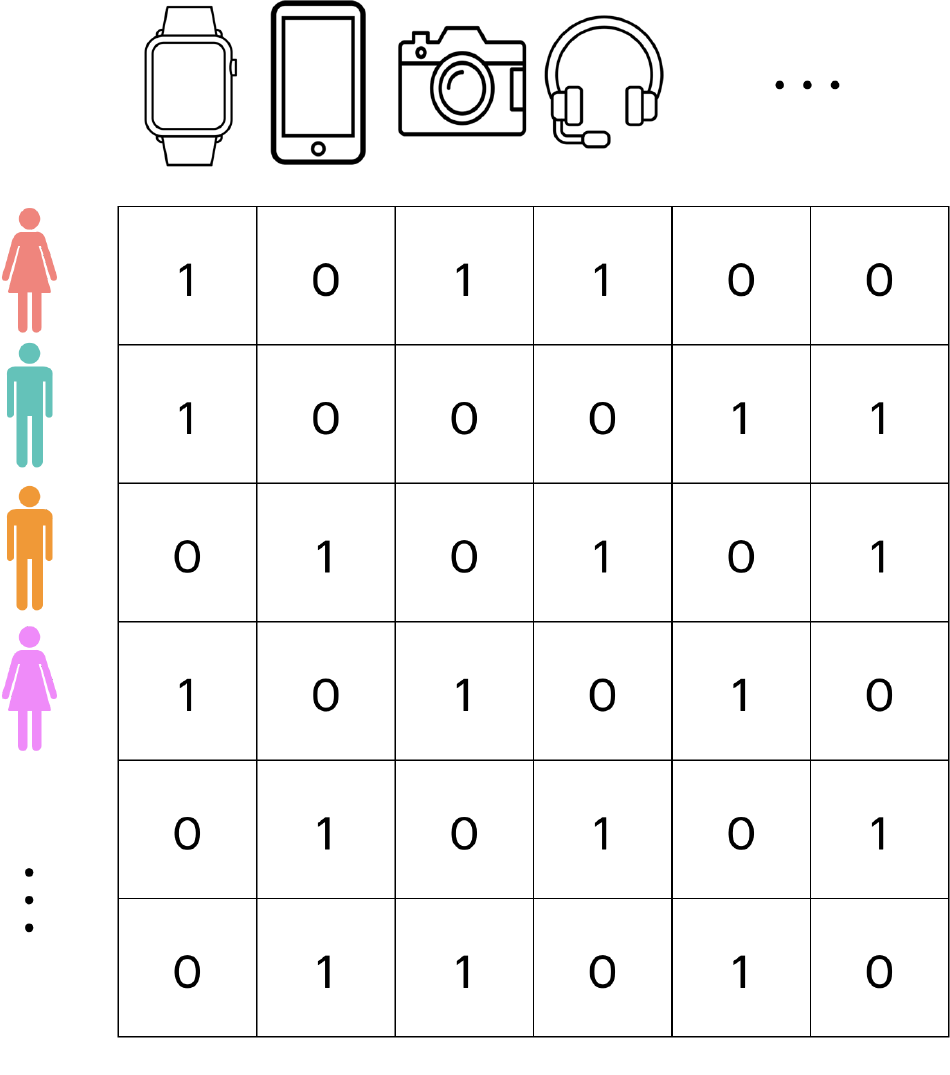}&
    \includegraphics[trim={3cm 7cm 2.5cm 2.75cm}, width=0.22\linewidth,clip]{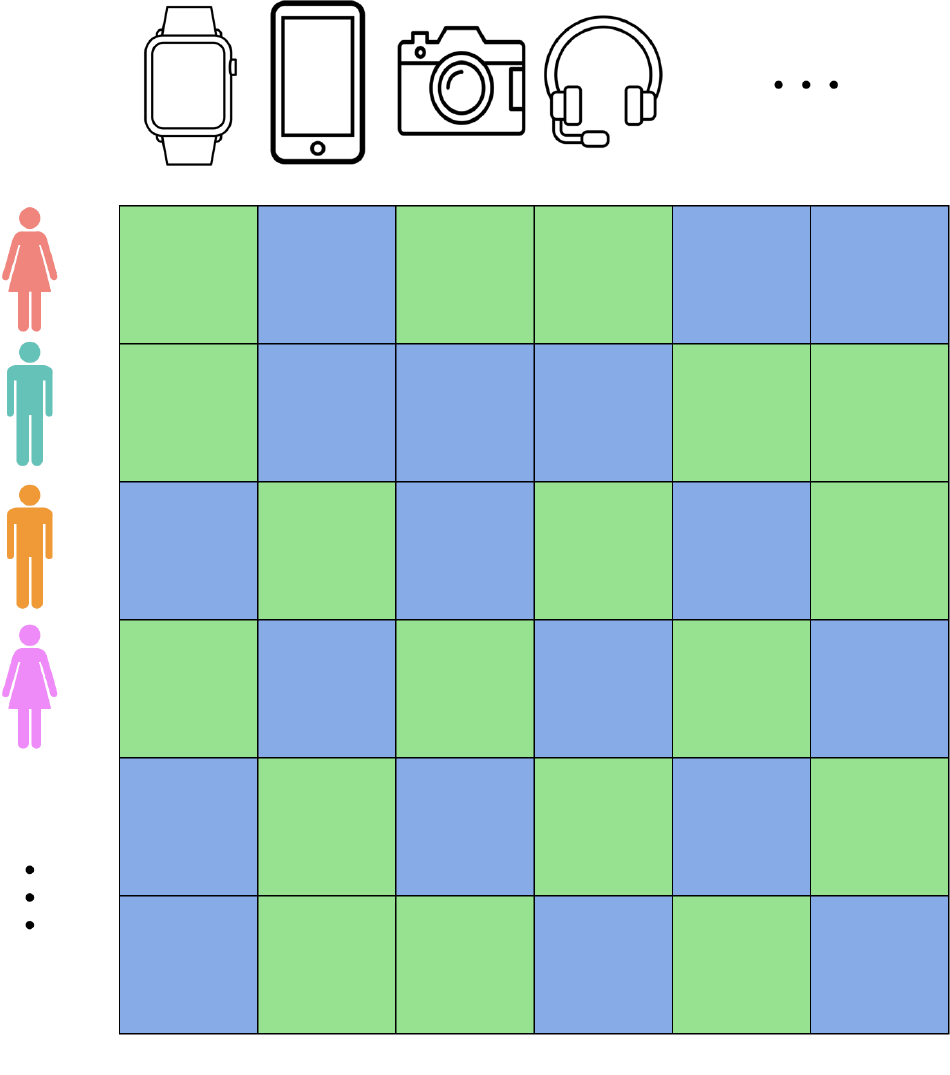}&
    \includegraphics[trim={3cm 7cm 2.5cm 2.75cm}, width=0.22\linewidth,clip]{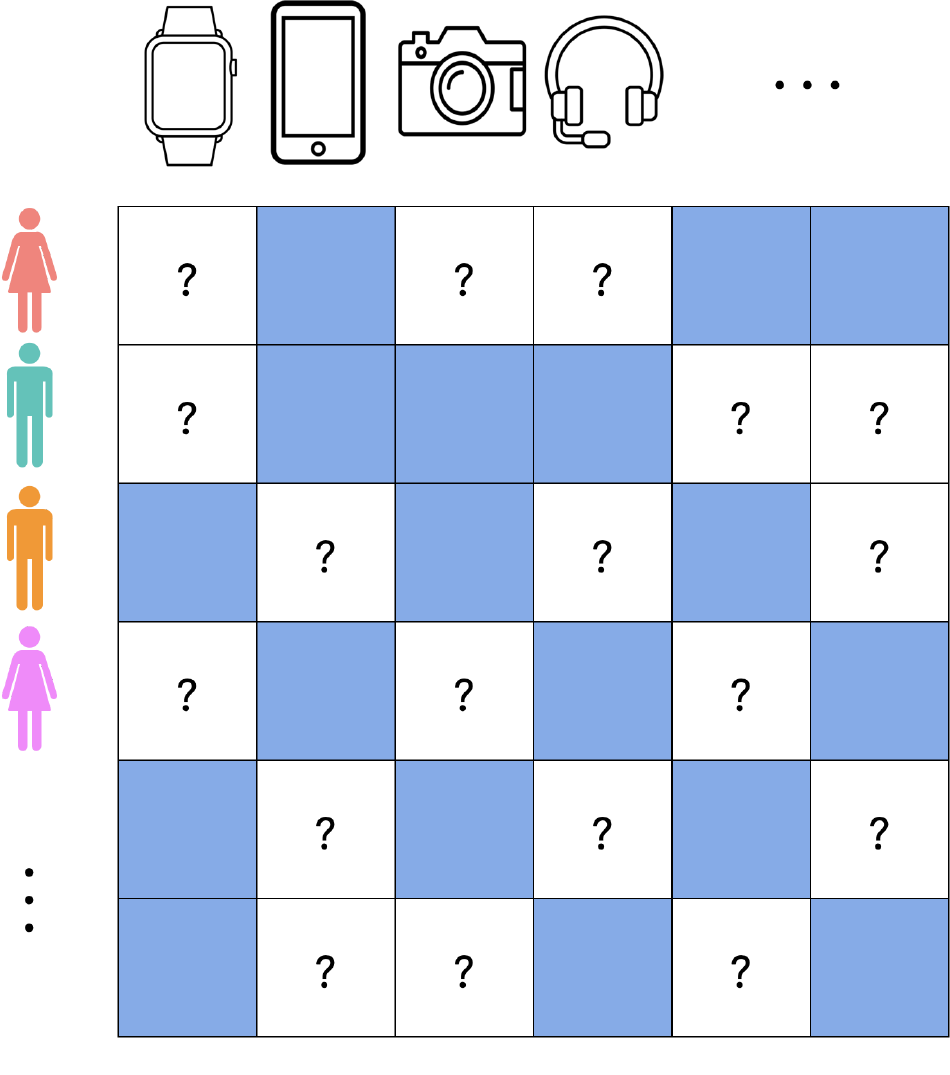}&
    \includegraphics[trim={3cm 7cm 2.5cm 2.75cm}, width=0.22\linewidth,clip]{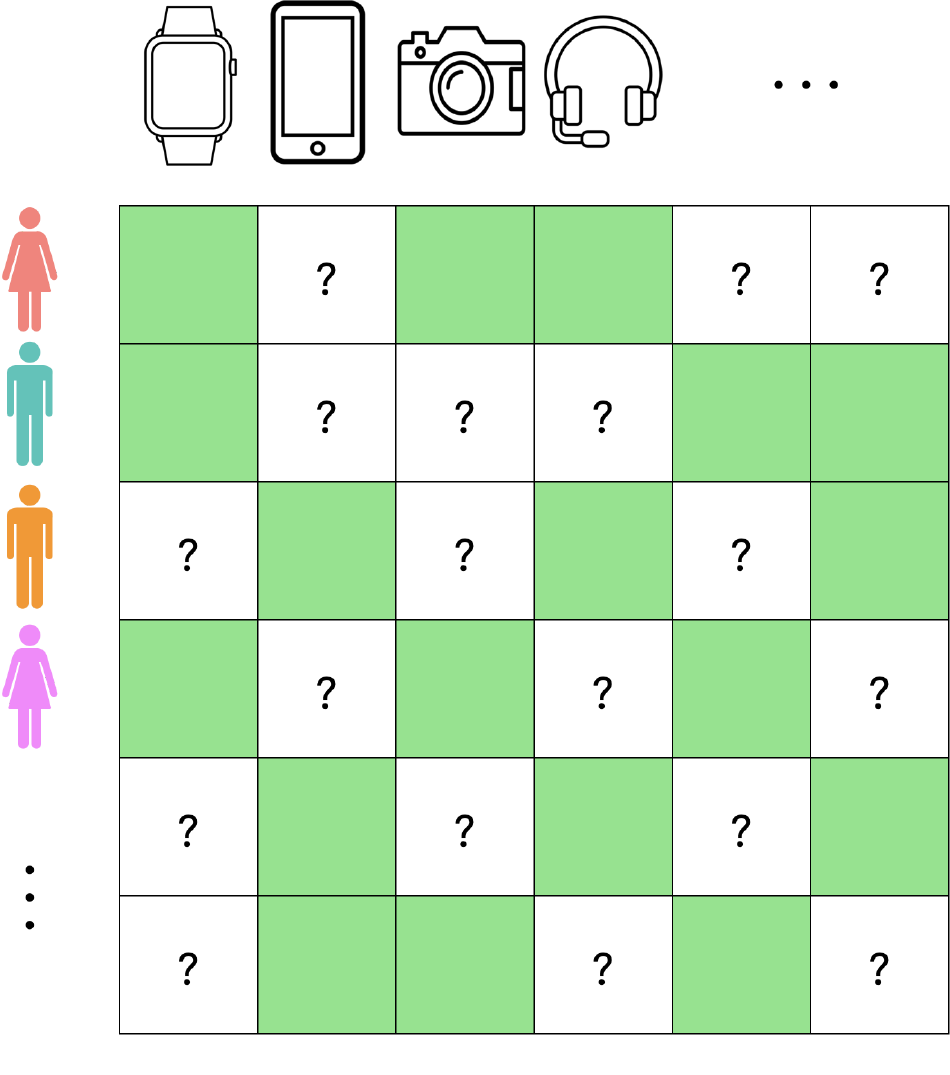}
    \\  
    $(a)$ $\ta$ & $(b)$ $\oo$ & $(c)$ $\ooz$ & $(d)$ $\ooo$
    \end{tabular}
\caption{
Schematic of the treatment assignment matrix $\ta$, the observed outcomes matrix $\oo$ (where green and blue fills indicate observations under $a = 1$ and $a = 0$, respectively), and the observed component of the potential outcomes matrices, i.e., $\ooz$ and $\ooo$ (where $\star$ indicates a missing value).
All matrices are $N \times M$ where $N$ is the number of customers and $M$ is the number of products. }

\label{figure_caricature}
\end{figure}

\subsection{Key building blocks} 

We now define and express matrices that are related to the quantities of interest $\mpoz, \mpoo$, and $\mta$ in a form similar to \cref{eq_mc_main}.
See Figure \ref{figure_caricature} for a visual representation of these matrices.
\begin{itemize}
\item 
\textbf{Outcomes}:  Let $\ooz = \oo \otimes (\allones - \ta) \in \{\Reals\cup \{\star\}\}^{N \times M}$ be a matrix with $(i, j)$-th entry equal to $y_{i, j}$ if $a_{i, j} = 0$, and equal to $\star$ otherwise. Here, $\allones$ is the $N \times M$ matrix with all entries equal to one. 
Analogously, let $\ooo  = \oo \otimes \ta \in \{\Reals\cup \{\star\}\}^{N \times M}$ be a matrix with $(i, j)$-th entry equal to $y_{i, j}$ if $a_{i, j} = 1$, and equal to $\star$ otherwise. 
In other words, $\ooz$ and $\ooo$ capture the observed components of $\normalbraces{y_{i,j}^{(0)}}_{i \in [N], j \in [M]}$ and $\normalbraces{y_{i,j}^{(1)}}_{i \in [N], j \in [M]}$, respectively, with missing entries denoted by $\star$.
Then, we can write
\begin{align}
\ooz = (\mpoz + \noiseyz)\otimes (\allones - \ta)
\qtext{and} \ooo = (\mpoo + \noiseyo)\otimes \ta.
\label{eq:y_matrix}
\end{align}
\item \textbf{Treatments}: From \cref{eq_treatment_model}, we can write 
\begin{align}
\ta = (\mta + \noisea).
\label{eq:a_matrix}
\end{align}
\end{itemize}
Building on the earlier discussion, the application of matrix completion yields the following estimates:
\begin{align}\label{eq:denoised_estimated}
\empoz = \mca(\ooz), \quad \empoo = \mca(\ooo), \qtext{and} \emta = \mca(\ta),
\end{align}
where the algorithm $\mca$ may vary for $\empoz$, $\empoo$, and $\emta$. Because
all entries of $\ta$ are observed, $\mca(\ta)$ denoises $\ta$ but does not need to impute missing entries. 
From \cref{eq:y_matrix} and \cref{eq:denoised_estimated}, it follows that $\empoz$ and $\empoo$ depend on $\ta$ and $\oo$, whereas $\emta$ depends only on $\ta$.

In this section, we deliberately leave the matrix completion algorithm $\mca$ as a ``black-box''. In \cref{sec_main_results}, we establish finite-sample and asymptotic guarantees for our proposed estimator, contingent on specific properties for $\mca$. 
In \cref{sec_algorithm}, we propose a novel end-to-end matrix completion algorithm that satifies these properties. 

Given matrix completion estimates of $(\empoz, \empoo, \emta)$, we formulate two preliminary estimators for $\ATETrue$: $(i)$ an outcome imputation estimator, which uses $\empoz$ and $\empoo$ only, and $(ii)$ an inverse probability weighting estimator, which uses $\emta$ only. Then, we combine these to obtain a doubly-robust estimator of $\ATETrue$. \medskip

\noindent {\bf Outcome imputation (OI) estimator.}
Let $\what{\theta}_{i,j}^{(a)}$ denote the $(i,j)$-th entry of $\empo$ for $i \in [N], j \in [M]$, and $a \in \{0, 1\}$.
The OI estimator for $\ATETrue$ is defined as follows:
\begin{gather}
\ATEOI  \defn \what{\mu}^{(1, \mathrm{OI})}_{\cdot,j} -  \what{\mu}^{(0, \mathrm{OI})}_{\cdot,j},\label{eq_counterfactual_mean_oi}
\\  \shortintertext{where} 
\what{\mu}^{(a, \mathrm{OI})}_{\cdot,j}  \defn \frac{1}{N} \sum_{i \in [N]}  \what{\theta}_{i,j}^{(a)}   
\qtext{for} a \in  \normalbraces{0,1}.
\end{gather}
That is, the OI estimator is obtained by taking the difference of the average value of the $j$-th column of the estimates $\empoz$ and $\empoo$. 
The quality of the OI estimator depends on how well $\empoz$ and $\empoo$ approximate the mean potential outcome matrices $\mpoz$ and $\mpoo$, respectively.\medskip

\noindent {\bf Inverse probability weighting (IPW) estimator.}
Let $\what{p}_{i,j}$ denote the $(i,j)$-th entry of $\emta$ for $i \in [N]$ and $j \in [M]$. 
The IPW estimate for $\ATETrue$ is  defined as follows:
\begin{gather}
\ATEIPW  \defn  \what{\mu}^{(1, \mathrm{IPW})}_{\cdot,j} -  \what{\mu}^{(0, \mathrm{IPW})}_{\cdot,j}, \label{eq_counterfactual_mean_ipw} 
\shortintertext{where} 
\what{\mu}^{(0, \mathrm{IPW})}_{\cdot,j}  \defn \frac{1}{N}\sum_{i \in [N]}  \frac{y_{i,j} \bigparenth{1 - a_{i,j}}}{1 - \what{p}_{i,j}}  
\quad\mbox{ and }\quad 
\what{\mu}^{(1, \mathrm{IPW})}_{\cdot,j}  \defn \frac{1}{N}\sum_{i \in [N]} \frac{ y_{i,j} a_{i,j}}{\what{p}_{i,j}}. 
\end{gather}
That is, the IPW estimator is obtained by taking the difference of the average value of the $j$-th column of 
the matrices $\ooz$ and $\ooo$, replacing unobserved entries with zeros, and weighting each outcome by the inverse of the estimated assignment probability to account for confounding.
The quality of the IPW estimate depends on how well $\emta$ approximates the probability matrix $\mta$.
\medskip

The matrix completion-based OI and IPW estimators in \cref{eq_counterfactual_mean_oi}  and \cref{eq_counterfactual_mean_ipw} have the same form as the classical OI and IPW estimators, which are derived for settings where all confounders are observed \citep[e.g.,][]{imbens_rubin_2015}. In contrast to the classical setting, our framework is one with unmeasured confounding. 

\subsection{Doubly-robust (DR) estimator}

The DR estimator of $\ATETrue$ combines the estimates $\empoz, \empoo$, and $\emta$ from \cref{eq:denoised_estimated}.
It is defined as follows:

\begin{gather}
\ATEDR  \defn  \what{\mu}^{(1, \mathrm{DR})}_{\cdot,j} -  \what{\mu}^{(0, \mathrm{DR})}_{\cdot,j}, \label{eq_counterfactual_mean_dr}
\\  \shortintertext{where} 
\what{\mu}^{(0, \mathrm{DR})}_{\cdot,j}   \defn \frac{1}{N}\sum_{i \in [N]}  \what{\theta}_{i,j}^{(0, \mathrm{DR})}
\qtext{with}
\what{\theta}_{i,j}^{(0, \mathrm{DR})}
\defeq 
\what{\theta}_{i,j}^{(0)} + \bigparenth{y_{i,j} - \what{\theta}_{i,j}^{(0)}} \frac{1 - a_{i,j}}{1 - \what{p}_{i,j}}, \label{eq_counterfactual_mean_dr_0}
\\ \shortintertext{ and } \what{\mu}^{(1, \mathrm{DR})}_{\cdot,j}  \defn \frac{1}{N}\sum_{i \in [N]}  \what{\theta}_{i,j}^{(1, \mathrm{DR})} 
\qtext{with}
\what{\theta}_{i,j}^{(1, \mathrm{DR})}
\defeq 
 \what{\theta}_{i,j}^{(1)} + \bigparenth{y_{i,j} - \what{\theta}_{i,j}^{(1)}} \frac{a_{i,j}}{\what{p}_{i,j}}. \label{eq_counterfactual_mean_dr_1}
\end{gather}
In \cref{sec_main_results}, we prove that $\ATEDR$  consistently estimates $\ATETrue$ as long as either $(\empoz, \empoo)$ is consistent for $(\mpoz, \mpoo)$ or $\emta$ is consistent for $\mta$, i.e., it is doubly-robust. 
Furthermore, we show that the DR estimator provides superior finite sample guarantees than the OI and IPW estimators, and that it satisfies a central limit theorem at a parametric rate under weak conditions on the convergence rate of the matrix completion routine. 
Using simulated data, \cref{figure_comparison} demonstrates the improved performance of DR, relative to OI and IPW. 
Despite substantial biases observed in both OI and IPW estimates, the error of the DR estimate closely follows a mean-zero Gaussian distribution. 
We provide a detailed description of the simulation setup in \cref{sec_simulations}.   

\begin{figure}[t]
    \centering
    \hspace*{-0.5cm}\begin{tabular}{c}
    \includegraphics[width=0.55\linewidth,clip]{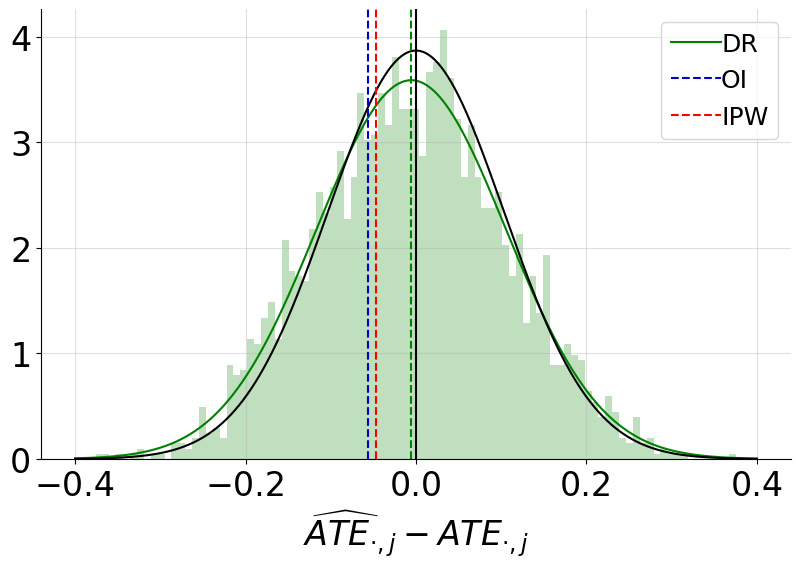}
    \end{tabular}
\caption{Simulation evidence of the convergence of the error of the doubly-robust (DR) estimator to a mean-zero Gaussian distribution. The histogram represents $\ATEDR - \ATETrue$, the green curve represents the (best) fitted Gaussian distribution, and the black curve represents the Gaussian approximation from \cref{thm_normality} in \cref{sec_main_results}. Histogram counts are normalized so that the area under the histogram integrates to one. Unlike DR, the outcome imputation (OI) and inverse probability weighting (IPW) estimators have non-trivial biases, as evidenced by the means of the distributions in dashed green, blue, and red, respectively. \cref{sec_simulations} reports complete simulation results. 
}
\label{figure_comparison}
\end{figure}

\section{Main Results}\label{sec_main_results}

This section presents the formal results of the article. 
\cref{sub:asum} details assumptions, \cref{subsec_non_asymp_guarantee} discusses finite-sample guarantees, and \cref{subsec_asymp_guarantee} presents a central limit theorem for $\ATEDR$.

\subsection{Assumptions}
\label{sub:asum}

\noindent {\bf Requirements on data generating process.}
We make two assumptions on how the data is generated.
First, we impose a positivity condition on the assignment probabilities.
\begin{assumption}[Positivity on true assignment probabilities]
\label{assumption_pos}
The unknown assignment probability matrix $\mta$ is such that
\begin{align}
 \lambda \leq p_{i,j} \leq 1 - \lambda, 
 \label{equation:bounds_ps}
\end{align}
for all $i \in [N]$ and $j \in [M]$, where $0 < \lambda \leq 1/2$.
\end{assumption}
\cref{assumption_pos} requires that the propensity score for each unit-outcome pair is bounded away from $0$ and $1$, implying that any unit-item pair can be assigned either of the two treatments. 
An analogous assumption is pervasive in causal inference models with no-unmeasured confounding. For simplicity of exposition and to avoid notational clutter, \cref{assumption_pos} requires \cref{equation:bounds_ps} for all outcomes, $j\in [M]$.
In practical applications, however, $\ATETrue$ may be estimated for a select group of those outcomes. In that case, the positivity assumption applies only for the selected subset of outcomes for which $\ATETrue$ is estimated.

Next, we formalize the requirements on the noise variables. 

\begin{assumption}[Zero-mean, independent, and $\subGaussian$ noise]
\label{assumption_noise}
Fix any $j \in [M]$. Then,
\begin{enumerate}[itemsep=1mm, topsep=2mm, label=(\alph*)]
    \item\label{item_ass_2aa} $\sbraces{\normalparenth{\varepsilon_{i,j}^{(0)}, \varepsilon_{i,j}^{(1)},\eta_{i,j}} : i\in[N]}$ are mean zero and independent (across $i$); 
    \item\label{item_ass_2bb} for every $i \in [N]$ and $j\in [M]$, $(\varepsilon_{i,j}^{(0)}, \varepsilon_{i,j}^{(1)}) \indep \eta_{i,j}$; moreover, the distribution of $(\varepsilon_{i,j}^{(0)}, \varepsilon_{i,j}^{(1)})$ depends on $(\mpoz, \mpoo,\mta)$ only through $(\theta_{i,j}^{(0)}, \theta_{i,j}^{(1)})$, and the distribution of $\eta_{i,j}$ depends on $(\mpoz, \mpoo,\mta)$ only through $p_{i,j}$; and
    \item\label{item_ass_2cc} $\varepsilon_{i,j}^{(a)}$ has $\subGaussian$ norm bounded by a constant $\sigmax$ for every $i \in [N]$ and $a \in \normalbraces{0,1}$.
\end{enumerate}
\end{assumption}
\cref{assumption_noise}\cref{item_ass_2aa} defines $(\mpoz, \mpoo, \mta)$ as matrices collecting the means of the potential outcomes and treatment assignments in \cref{eq_outcome_model,eq_treatment_model}. Further, for every measurement, it imposes independence across units in the noise variables. \cref{assumption_noise}\cref{item_ass_2bb} imposes independence between the noise in the potential outcomes and noise in treatment assignment, and implies that for each particular unit $i$ and measurement $j$, confounding emerges only from the interplay between $(\theta_{i,j}^{(0)}, \theta_{i,j}^{(1)})$ and $p_{i,j}$. 
Finally, \cref{assumption_noise}\cref{item_ass_2cc} is mild and useful to derive finite-sample guarantees. 
For the central limit theorem in \cref{subsec_asymp_guarantee}, subGaussianity could be disposed of by restricting the moments of $\varepsilon_{i,j}^{(a)}$. 
\cref{assumption_noise} does not restrict the dependence between $\varepsilon_{i,j}^{(0)}$ and $\varepsilon_{i,j}^{(1)}$.
Neither \cref{assumption_noise} restricts the dependence of $\eta_{i,j}$ across outcomes. In particular, \cref{assumption_noise} allows for the existence of pairs of outcomes $(j,j')$ such that $\Expectation[\eta_{i,j}^2]=\Expectation[\eta_{i,j'}^2]=\Expectation[\eta_{i,j}\eta_{i,j'}]$, in which case $a_{i,j}=a_{i,j'}$ with probability one.\medskip 

\noindent {\bf Requirements on matrix completion estimators.}
First, we assume the estimate $\emta$ is consistent with \cref{assumption_pos}. 
\begin{assumption}[Positivity on estimated assignment probabilities]
\label{assumption_pos_estimated}
The estimated probability matrix $\emta$ is such that
\begin{align}
 \lbar \leq \what{p}_{i,j} \leq 1 - \lbar, 
\end{align}
for all $i \in [N]$ and $j \in [M]$, where $0<\lbar \leq \lambda$.
\end{assumption}

\cref{assumption_pos_estimated} holds when the entries of $\emta$ are truncated to the range $[\lbar, 1-\lbar]$, provided $\lbar$ is not greater than $\lambda$.
Second, our theoretical analysis requires independence between certain elements of the estimates $(\emta, \empoz, \empoo)$  from \cref{eq:denoised_estimated}, and the noise matrices $(\noisea, \noiseyb)$. 
We formally state this independence condition as an assumption below. 

\begin{assumption}[Independence between estimates and noise]
\label{assumption_estimates}
Fix any $j \in [M]$. There exists a non-empty partition $(\cR_0, \cR_1)$ of the units $[N]$ such that
\begin{gather}
\bigbraces{\bigparenth{\what{p}_{i, j}, \what{\theta}_{i,j}^{(a)}}}_{i\in\cR_s} \indep \bigbraces{\eta_{i, j}}_{i\in\cR_s}
\label{eq_independence_requirement_estimates_dr}
\shortintertext{and}
\bigbraces{ \what{p}_{i, j} }_{i\in\cR_s} \indep \bigbraces{\bigparenth{\eta_{i, j}, \varepsilon^{(a)}_{i,j}}}_{i\in\cR_s},
\label{eq_independence_requirement_estimates_dr_p}
\end{gather}
for every $a \in \normalbraces{0,1}$ and $s \in \normalbraces{0,1}$.
\end{assumption}
\cref{eq_independence_requirement_estimates_dr} requires that within each of the two partitions of the units, estimated mean potential outcomes and estimated assignment probabilities are jointly independent of the error in assignment probabilities, for every measurement. Similarly, \cref{eq_independence_requirement_estimates_dr_p} requires that within each of the two partitions of the units, estimated assignment probabilities are independent jointly of the noise in assignment probabilities and potential outcomes, for every measurement. Conditions like \cref{eq_independence_requirement_estimates_dr} and \cref{eq_independence_requirement_estimates_dr_p} are familiar in the doubly-robust estimation literature. \citet{chernozhukov2018double} employ a cross-fitting device to enforce
an assumption similar to \cref{assumption_estimates} in a context with no unmeasured confounders. 
\cref{sec_algorithm} provides a novel cross-fitting procedure for matrix estimation under which \cref{assumption_estimates} holds for any $\mca$ algorithm (under additional assumptions on the noise variables).\medskip

\noindent {\bf Matrix completion error rates.} The formal guarantees in this section depend on the normalized $(1,2)$-norms of the errors in estimating the unknown parameters $(\mpoz, \mpoo, \mta)$.
We use the following notation for these errors:
\begin{gather}
 \RP \defn \frac{\onetwonorm{\emta - \mta}}{\sqrt{N}} \qtext{and} \RTheta \defn \!\!\sum_{a \in \normalbraces{0,1}} \!\!  \Rcol\bigparenth{\empo}, \label{eq_notation} \\ \shortintertext{where} \Rcol\bigparenth{\empo}  = \frac{\onetwonorm{\empo - \mpo}}{\sqrt{N}}. 
 \label{eq:etheta_p}
\end{gather}
A variety of matrix completion algorithms deliver $\RP=O_p(\min\normalbraces{N,M}^{-\alpha})$ and $\RTheta=O_p(\min\normalbraces{N,M}^{-\beta})$, where $0<\alpha, \beta \leq 1/2$. The conditions in this section track dependence on $N$ only. We say that the normalized errors $\RP$ and $\RTheta$ achieve the parametric rate when they have the same rate as $O_p(N^{-1/2})$.
\cref{sec_algorithm} explicitly characterizes how the rates of convergence $\RP$ and $\RTheta$ depend on $N$ and $M$ for a particular matrix completion algorithm based on \cite{bai2021matrix}.

\subsection{Non-asymptotic guarantees}\label{subsec_non_asymp_guarantee}

The first main result of this section provides a non-asymptotic error bound for  $\ATEDR - \ATETrue$ in terms of the errors $\RP$ and $\RTheta$ defined in \cref{eq_notation}. 
\begin{theorem}[\textbf{\fsg}]\label{thm_fsg}
Suppose \cref{assumption_pos,assumption_noise,assumption_estimates,assumption_pos_estimated} hold. Fix $\delta \in (0,1)$ and $j \in [M]$.
Then, with probability at least $1-\delta$, we have
\begin{align}
\bigabs{\ATEDR - \ATETrue}  \leq \error_{N, \delta}, 
\label{fsg_dr_ate_t}
\end{align}
where 
\begin{align}
\error_{N, \delta} & \!\defn\!  
\frac{2}{\lbar}  \biggbrackets{\RTheta  \RP + \Bigparenth{\frac{\sqrt{c \ld[/12]}}{\sqrt{\lone}} \RTheta + 2\sigmax \sqrt{c \ld[/12]} + \frac{2\sigmax \m(c \ld[/12])}{\sqrt{\lone}} }   \frac{1}{\sqrt{N}}},
\label{eq_combined_error_dr}
\end{align}
for $\m(c)$ and $\ell_c$ as defined in \cref{section_introduction}. 
\end{theorem}
The proof of \cref{thm_fsg} is given in \cref{sec_proof_thm_fsg}. \cref{fsg_dr_ate_t,eq_combined_error_dr} bound the absolute error of the DR estimator by  {the rate of $\RTheta (\RP  + N^{-0.5}) + N^{-0.5}$.} When $\RP$ is lower bounded at the parametric rate of $N^{-0.5}$,  $\error_{N, \delta}$ has the same rate as $\RP \RTheta + N^{-0.5}$.\medskip

\noindent {\bf Doubly-robust behavior of \texorpdfstring{$\ATEDR$}{}.}
The error rate of  $\RP \RTheta + N^{-0.5}$ immediately reveals that the DR estimate is doubly-robust with respect to the error in estimating the mean potential outcomes $(\mpob)$ and the assignment probabilities $\mta$. 
First, the error  $\error_{N, \delta}$ decays at a parametric rate of $O_p(N^{-0.5})$ as long as the product of error rates, $\RP \RTheta$, decays as $O_p(N^{-0.5})$. As a result, $\ATEDR$ can exhibit a parametric error rate even when neither the mean potential outcomes nor the assignment probabilities are estimated at a parametric rate.
Second, $\error_{N, \delta}$ decays to zero as long as either of $\RP$ or $\RTheta$ decays to zero, provided both errors are $O_p(1)$. 

We next compare the performance of DR estimator with the OI and IPW  estimators from \cref{eq_counterfactual_mean_oi,eq_counterfactual_mean_ipw}, respectively.
Towards this goal, we characterize the $\ATETrue$ estimation error of $\ATEOI$ in terms of $\RTheta$ and of $\ATEIPW$ in terms of $\RP$. 

\begin{proposition}[\textbf{\fsgoiipw}]
\label{thm_fsg_ipw_oi}
Fix any $j \in [M]$. For OI, we have
\begin{align}
\bigabs{\ATEOI - \ATETrue} & \leq  \error[OI]_N \defn \RTheta. \label{fsg_oi_ate_t}\\
\intertext{For IPW, suppose \cref{assumption_pos,assumption_noise,assumption_estimates,assumption_pos_estimated} hold. Define $\thetamax \defn \sum_{a \in \normalbraces{0,1}}\maxmatnorm{\mpo}$, and fix any $\delta \in (0,1)$. Then, with probability at least $1-\delta$, we have}
\bigabs{\ATEIPW - \ATETrue}
& \leq  \error[IPW]_{N, \delta}, \label{fsg_ipw_ate_t}
\end{align}
where 
\begin{align}
\error[IPW]_{N, \delta} \defn  \frac{2}{\lbar}  \biggbrackets{\thetamax\,  \RP + \Bigparenth{\frac{\sqrt{c \ld[/12]}}{\sqrt{\lone}} \thetamax + 2\sigmax \sqrt{c \ld[/12]} + \frac{2\sigmax \m(c \ld[/12])}{\sqrt{\lone}} }  \frac{1}{\sqrt{N}}},
\end{align}
for $\m(c)$ and $\ell_c$ as defined in \cref{section_introduction}. 
\end{proposition} 
The proofs of \cref{fsg_oi_ate_t} and \cref{fsg_ipw_ate_t} are given in  \cref{sec_proof_thm_fsg_oi,sec_proof_thm_fsg_ipw}, respectively. \cref{thm_fsg_ipw_oi} implies that 
\ifx\submission\arxiv
in an asymptotic sequence
\fi
with bounded $\thetamax$, OI and IPW attain the parametric rate $O_p(N^{-0.5})$ provided $\RTheta$ and $\RP$ are $O_p(N^{-0.5})$, respectively.
The next corollary, proven in \cref{proofs_coro}, compares these error rates with those obtained for the DR estimator in \cref{thm_fsg}.
\begin{corollary}[\textbf{Gains of DR over OI and IPW}]
\label{coro_compare}
Suppose
\cref{assumption_pos,assumption_noise,assumption_estimates,assumption_pos_estimated} hold. Fix any $j \in [M]$.
\ifx\submission\arxiv
Consider an asymptotic sequence such that 
\fi
\ifx\submission\ec
Suppose
\fi
$\thetamax$ is bounded. If $\RP = O_p(N^{-\alpha})$ and $\RTheta = O_p(N^{-\beta})$ for $0 \leq \alpha\leq 0.5$ and $0\leq \beta \leq 0.5$, then
\begin{gather}
     \bigabs{\ATEOI - \ATETrue} = O_p(N^{-\beta}),\qquad \bigabs{\ATEIPW - \ATETrue} = O_p(N^{-\alpha}),\\
    \shortintertext{and}
    \bigabs{\ATEDR - \ATETrue} = O_p(N^{- \min\normalbraces{\alpha+\beta, 0.5}}).
    \end{gather}
\end{corollary}
\cref{coro_compare} shows that the DR estimate's error decay rate is consistently superior to that of the OI and IPW estimates across a variety of regimes for $\alpha, \beta$. 
Specifically, the error $\error_{N, \delta}$ scales strictly faster than both $\error[OI]_N$ and $\error[IPW]_{N, \delta}$ if the estimation errors of $\empoz$, $\empoo$, and $\emta$ converge slower than at the parametric rate $O_p(N^{-1/2})$. 
When the estimation errors of $\empoz$, $\empoo$, and $\emta$ all decay at a parametric rate, OI, IPW, and DR estimation errors decay also at a parametric rate.

\subsection{Asymptotic guarantees}\label{subsec_asymp_guarantee}

The next result, proven in  \cref{proofs_coro} as a corollary of \cref{thm_fsg}, provides conditions on $\RP$ and $\RTheta$ for consistency of $\ATEDR$.
\begin{corollary}[\textbf{Consistency for DR}]\label{coro_consistency}
    Suppose \cref{assumption_pos,assumption_noise,assumption_estimates,assumption_pos_estimated} hold. As $N\rightarrow \infty$, if either {\it (i)} $\RP = o_p(1), \ \RTheta = O_p(1)$, or {\it (ii)} $\RTheta = o_p(1), \ \RP = O_p(1)$, it holds that
\begin{align}
\ATEDR -  \ATETrue \inprob 0, \label{1con_dr_ate_t}
\end{align}
for all $j\in[M]$. 
\end{corollary}
\cref{coro_consistency} states that $\ATEDR$ is a consistent estimator for $\ATETrue$ as long as either the mean potential outcomes or the assignment probabilities are estimated consistently.

The next theorem, proven in \cref{sec_proof_thm_normality}, establishes a Gaussian approximation for $\ATEDR$ under mild conditions on error rates $\RP$ and $\RTheta$.
\begin{theorem}[\textbf{\normality}]\label{thm_normality}
Suppose \cref{assumption_pos,assumption_noise,assumption_pos_estimated,assumption_estimates} and the following conditions hold,
\begin{enumerate}[label=(C\arabic*), leftmargin=2.5em] 
\item\label{item_error_for_normality} 
$\RP = o_p(1)$ and $\RTheta = o_p(1)$. 
\item\label{item_product_of_error_for_normality} $\RP \RTheta = o_p\bigparenth{N^{-1/2}}$.
\item\label{item_sigma_bounded_below} For every $i \in [N]$ and $j \in [M]$, let $\sigma_{i,j}^{(0)}$ and $\sigma_{i,j}^{(1)}$ be the standard deviations of $\varepsilon_{i,j}^{(0)}$ and $\varepsilon_{i,j}^{(1)}$, respectively. The sequence
\begin{equation}
\wbar{\sigma}_j^2 \defn 
\frac{1}{N} \sum_{i \in [N]} \frac{(\sigma_{i,j}^{(1)})^2}{p_{i,j}} + \frac{1}{N} \sum_{i \in [N]} \frac{(\sigma_{i,j}^{(0)})^2}{1-p_{i,j}},
\label{eq_normality_variance_thm}
\end{equation}
is bounded away from zero as $N$ increases.
\end{enumerate}
Then, for all $j \in [M]$, 
\begin{align}
\sqrt{N} \bigparenth{\ATEDR - \ATETrue}  / \wbar{\sigma}_j  \stackrel{d}{\longrightarrow} \mathcal{N}\bigparenth{0, 1}, 
\label{eq_clt}
\end{align}
as $N\to \infty$.
\end{theorem}
\cref{thm_normality} describes two simple requirements on the estimated matrices $\emta$ and $(\empoz, \empoo)$, under which $\ATEDR$ exhibits an asymptotic Gaussian distribution centered at $\ATETrue$. 
Condition \cref{item_error_for_normality} requires that the estimation errors of $\emta$ and  $(\empoz, \empoo)$ converge to zero in probability. Condition \cref{item_product_of_error_for_normality} requires that the product of the errors decays sufficiently fast, at a rate $o_p\bigparenth{N^{-1/2}}$, ensuring that the bias of the normalized estimator in \cref{eq_clt} converges to zero. 
Condition \cref{item_product_of_error_for_normality} is similar to conditions in the literature on doubly-robust estimation of average treatment effects under observed confounding \citep[e.g., Assumption 5.1 in][]{chernozhukov2018double}.
Specifically, in that context, \citet{chernozhukov2018double} assume that the product of propensity estimation error and outcome regression error decays faster than $N^{-1/2}$.\medskip

\noindent {\bf Black-box asymptotic normality.} We emphasize that \cref{thm_normality} applies to any matrix completion algorithm $\mca$, provided conditions \cref{item_error_for_normality} and \cref{item_product_of_error_for_normality} hold. This level of generality is useful because the product of $\RP$ and $\RTheta$ is $o_p\bigparenth{N^{-1/2}}$ for a wide range of $\mca$ algorithms, under mild assumptions on $(\mpob, \mta)$. In contrast, achieving such black-box asymptotic normality for OI or IPW estimates is challenging. Their biases are tied to the individual error rates, $\RTheta$ and $\RP$, which are typically lower-bouded at the parametric rate of $N^{-0.5}$.

The next result, proven in \cref{subsec_proof_variance_estimate}, provides a consistent estimator for the asymptotic variance $\wbar{\sigma}_j^2$ from \cref{thm_normality}.
\begin{proposition}[\textbf{Consistent variance estimation}]\label{prop_variance_estimate}
Suppose \cref{assumption_pos,assumption_noise,assumption_pos_estimated} and condition \cref{item_error_for_normality} in \cref{thm_normality} holds. Suppose 
the partition $(\cR_0, \cR_1)$ of the units $[N]$ from \cref{assumption_estimates} is such that
\begin{align}
    \normalbraces{\bigparenth{\what{p}_{i, j}, \what{\theta}^{(a)}_{i,j}}}_{i\in\cR_s} \indep \normalbraces{\normalparenth{\eta_{i, j}, \varepsilon^{(a)}_{i,j}}}_{i\in\cR_s}, \label{eq_independence_requirement_estimates_new}
\end{align}
for every $j \in [M]$, $a \in \normalbraces{0,1}$ and $s \in \normalbraces{0,1}$. Then, for all $j \in [M]$, $\what{\sigma}_j^2 - \wbar{\sigma}_j^2 \inprob 0$, where
    \begin{align}
    \what{\sigma}_j^2 \defn {\frac{1}{N} \sum_{i \in [N]} \frac{ \bigparenth{y_{i,j} - \what{\theta}^{(1)}_{i,j}}^2 a_{i,j}}{\bigparenth{\what{p}_{i,j}}^2} + \frac{1}{N} \sum_{i \in [N]}  \frac{ \bigparenth{y_{i,j} - \what{\theta}^{(0)}_{i,j}}^2 \normalparenth{1-a_{i,j}}} {\bigparenth{1-\what{p}_{i,j}}^2}}. \label{eq_estimated_variance}
\end{align}
\end{proposition}

\subsection{Application to panel data with lagged treatment effects}
\label{extension_delayed} \cref{subsec_non_asymp_guarantee,subsec_asymp_guarantee} 
considered a model where the outcome $y_{i,j}$ for unit $i \in [N]$ and measurement $j \in [M]$ depends on treatment assignment only for unit $i$ and measurement $j$, i.e., $a_{i,j}$. \cref{sec:app_dynamic} discusses how to extend the results of this section to a setting of panel data with lagged treatment effects. In a panel data setting, the $M$ measurements correspond to $T$ time periods, and $t$ denotes the time index.
Then, \cref{sec:app_dynamic} considers an auto-regressive setting, where the potential outcomes at time $t$ depends on the treatment assignment at time $t$ and the realized outcome at time $t-1$, i.e., for all $i \in [N], t \in [T]$, and $a \in \normalbraces{0,1}$, 
\begin{equation}
y_{i,t}^{(a|y_{i,t-1})} = \alpha^{(a)} y_{i,t-1} + \theta_{i, t}^{(a)} 
+ \vareps_{i,t}^{(a)},
\end{equation}
and observed outcomes satisfy 
\begin{align}
\label{eq:dynamic_obs_orig}
    y_{i,t} & = y_{i,t}^{(0|y_{i,t-1})} (1-a_{i,t}) + y_{i,t}^{(1|y_{i,t-1})}  a_{i,t}.
\end{align}

The presence of lagged treatment effects in this model makes it crucial to define causal estimands for entire sequences of treatments. \cref{sec:app_dynamic} describes how the proposed doubly-robust estimation can be extended to treatment sequences and derives a generalization of \cref{thm_fsg}.

\section{Matrix Completion with Cross-Fitting}
\label{sec_algorithm}

In this section, we introduce a novel algorithm designed to construct estimates $(\empoz, \empoo, \emta)$ that adhere to \cref{assumption_estimates} and satisfy conditions \cref{item_error_for_normality,item_product_of_error_for_normality} in \cref{thm_normality}. 
We first explain why traditional matrix completion algorithms fail to deliver the properties required by \cref{assumption_estimates}.
We then present $\ssMC$, a meta-algorithm that takes any matrix completion algorithm and uses it to construct $(\empoz, \empoo, \emta)$ that satisfy \cref{assumption_estimates}, and the stronger independence condition in \cref{prop_variance_estimate}. 
Finally, we describe $\ssSVD$, an end-to-end algorithm obtained by combining $\ssMC$ with the singular value decomposition (\texttt{SVD})-based algorithm of \citet{bai2021matrix}, and establish that it also satisfies conditions \cref{item_error_for_normality,item_product_of_error_for_normality} in \cref{thm_normality}.\medskip

\noindent {\bf Traditional matrix completion.} Estimates $(\empoz, \empoo, \emta)$ obtained from existing matrix completion algorithms need not satisfy \cref{assumption_estimates}. In particular, using the entire assignment matrix $\ta$ to estimate each element of $\mta$ typically results in a violation of $\bigbraces{\what{p}_{i, j}}_{i\in\cR_s} \indep \bigbraces{\eta_{i, j}}_{i\in\cR_s}$ in \cref{assumption_estimates}, as each entry of $\what{P}$ is allowed to depend on the entire noise matrix $\noisea$. For example, in spectral methods \citep[e.g.,][]{nguyen2019low}, $\emta$ is a function of the \texttt{SVD} of the entire matrix $\ta$, and
\begin{align}\label{eq_dependence}
    \what{p}_{i,j} \notindep a_{i',j'}, 
\end{align}
for all $(i, j), (i', j') \in [N] \times [M]$ in general, which implies $\bigbraces{\what{p}_{i, j}}_{i\in\cR_s} \notindep \bigbraces{\eta_{i, j}}_{i\in\cR_s}$, for every $\cR_s \subset [N]$. Similarly, in matching methods such as nearest neighbors \citep{li2019nearest}, $\emta$ is a function of the matches/neighbors estimated from the entire matrix $\ta$. 
Dependence structures such as $\what{p}_{i,j} \notindep a_{i,j}$ for any $i, j \in [N] \times [M]$---which is weaker than \cref{eq_dependence}---are enough to violate the $\bigbraces{\what{p}_{i, j}}_{i\in\cR_s} \indep \bigbraces{\eta_{i, j}}_{i\in\cR_s}$ requirement in \cref{assumption_estimates}.
Likewise, the requirement $\bigbraces{ \what{\theta}_{i,j}^{(a)}}_{i\in\cR_s} \indep \bigbraces{\eta_{i, j}}_{i\in\cR_s}$ in \cref{assumption_estimates} can be violated, because $\empoz$ and $\empoo$ depend respectively on $\ooz$ and $\ooo$, which themselves depend on the entire matrix $\ta$.

\subsection{\texorpdfstring{\ssMC}{}: A meta-cross-fitting algorithm for matrix completion}
\label{sub:kbb}

We now introduce $\ssMC$, a cross-fitting procedure that modifies any $\mca$ algorithm to produce $(\empoz, \empoo, \emta)$ that satisfy \cref{assumption_estimates}. We employ the following assumption on the noise variables.
\begin{assumption}[Block independence between noise]
\label{assumption_block_noise}
Let $(\cR_0, \cR_1)$ denote the partition of the units $[N]$ from \cref{assumption_estimates}. There exists partitions $(\cC_0, \cC_1)$ of the measurements $[M]$, such that for each block $\cI \in \mc P \defn \normalbraces{\cR_s \times \cC_k: s,k \in \normalbraces{0,1}}$,
\begin{align}
\noisea_{\cI} & \indep  \noisea_{-\cI},  \noisey_{-\cI} 
\label{eq_noise_independence_1}
\shortintertext{and}
\noisea_{-\cI} & \indep \noisea_{\cI}, \noisey_{\cI}.
\label{eq_noise_independence_2}
\end{align}
for every $a \in \normalbraces{0,1}$.
\end{assumption}
For a given block $\cI$, \cref{eq_noise_independence_1} requires the noise in the treatment assignments corresponding to $\cI$ to be independent jointly of the noise in the treatment assignments and the potential outcomes corresponding to the remaining three blocks. Likewise, \cref{eq_noise_independence_2} requires the noise in the treatment assignments corresponding to the remaining three blocks to be independent jointly of the noise in the treatment assignments and the potential outcomes corresponding to $\cI$. \cref{assumption_block_noise} leaves unrestricted the dependence of the noise variables across outcomes that belong to the same block.

For notational simplicity, \cref{assumption_block_noise} imposes independence conditions across blocks of outcomes in a partition of $[M]$ into  two blocks only. It is important to note, however, that the results in this section  hold under more general dependence patterns. In particular, at the cost of additional notational complexity, it is straightforward to extend the result in this section to partitions of outcomes $(\cC_0, \cC_1, \ldots, \cC_m)$ such that for each $k \in \{0, 1, \ldots, m\}$, $s\in\{0,1\}$ and $a\in\{0,1\}$, there exists $k' \in \{0, 1, \ldots, m\} \setminus \{k\}$ with 
$\{\eta_{i,j}\}_{(i,j) \in \mathcal{R}_s \times \cC_k} \indep \{\eta_{i,j}, \varepsilon_{i,j}^{(a)}\}_{(i,j) \in \mc{R}_{1-s} \times \cC_{k'}}$ and $\{\eta_{i,j}\}_{(i,j) \in  \mc{R}_{1-s} \times \cC_{k'}} \indep \{\eta_{i,j}, \varepsilon_{i,j}^{(a)}\}_{(i,j) \in \mathcal{R}_s \times \cC_k}$.
This  allows for rather general patterns of dependence across outcomes while preserving independence across specific sets of outcomes (e.g., certain product categories in the retail example of \cref{section_introduction}).

Recall the setup from \cref{subsec_mc_primer}: Given an observation matrix $\matS \in\{\Reals \cup \{\star\}\}^{N \times M}$, a matrix completion algorithm $\mca$  produces an estimate $\what{\matT} = \mca(\matS) \in \Reals^{N \times M}$ of a matrix of interest $\matT$, where $\matS$ and $\matT$ are related via \cref{eq_mc_main}. 
With this background, we now describe the $\ssMC$ meta-algorithm. 
\begin{enumerate}
    \item The inputs are $(i)$ a matrix completion algorithm $\mca$, $(ii)$ an observation matrix $\matS \in \{\Reals \cup \{\star\}\}^{N \times M}$, and $(iii)$ a block partition $\cP$ of the set $[N] \times [M]$ into four blocks as in  \cref{assumption_block_noise}.
    \item\label{step2_ssmc} 
    For each block $\cI \in \cP$, construct $\what{\matT}_{\mc I}$ by applying $\mca$ on $\matS  \otimes \allones^{-\mc I}$ where $\allones^{-\mc I} \in \Reals^{N \times M}$ denotes a masking matrix with $(i,j)$-th entry equal to $0$ if $(i, j) \in \mc I$ and $1$ otherwise, and the operator $\otimes$ is as defined in \cref{section_introduction}.
    In other words,
         \begin{align}\label{eq_mc_sample_split}
         \what{\matT}_{\mc I} & = \wbar{\matT}_{\mc I} \qtext{where} \wbar{\matT} = \mca(\matS  \otimes \allones^{-\mc I}). 
         \end{align}
    \item Return $\what{\matT} \in \Reals^{N\times M}$ obtained by collecting together $\normalbraces{\what{\matT}_{\cI}}_{\cI \in \cP}$, with each entry in its original position.
\end{enumerate}
We represent this meta-algorithm succinctly as below:
\begin{align}
    \what{\matT} = \ssMC(\mca, \matS, \cP). \label{eq_meta_ssmc}
\end{align}
In summary, $\ssMC$ produces an estimate $\what{\matT}$ such that for each block $\mc I  \in \cP$, the sub-matrix $\what{\matT}_{\mc I}$ is constructed only using the entries of $\matS$ corresponding to the remaining three blocks of $\cP$. 
\cref{figure_sample_split}$(a)$ provides a schematic of the block partition $\mc P$ for $\cR_0 = [\floors{N/2}]$ and $\cC_0 = [\floors{M/2}]$. See \cref{figure_sample_split}$(b)$ for a visualization of $\matS  \otimes \allones^{-\mc I}$. 
The following result, proven in \cref{proof_prop_ssmc}, establishes $(\empob, \emta)$ generated by $\ssMC$ satisfy \cref{assumption_estimates}. 

\begin{proposition}[\textbf{\propssmc}]
\label{prop_ssmc}
    Suppose \cref{assumption_noise,assumption_block_noise} hold. Let $\mca$ be any matrix completion algorithm and $\cP$ be the block partition of the set $[N] \times [M]$ into four blocks from \cref{assumption_block_noise}. 
    Let
    \begin{align}
        \empoz & = \ssMC(\mca, \ooz, \cP), \label{ssmc_1}\\
        \empoo & = \ssMC(\mca, \ooo, \cP), \label{ssmc_2}\\
        \emta & = \ssMC(\mca, \ta, \cP), \label{ssmc_3}
    \end{align}
    where $\ooz$ and $\ooo$ are defined in \cref{eq:y_matrix}. Then, \cref{assumption_estimates} holds for all $j \in [M]$. Further, suppose
    \begin{align}
        \noisea_{\cI}, \noisey_{\cI} \indep \noisea_{-\cI}, \noisey_{-\cI}, \label{eq_noise_independence_3}
    \end{align}
    for every block $\cI \in \mc P$ and $a \in \normalbraces{0,1}$.
    Then, \cref{eq_independence_requirement_estimates_new} holds too.
\end{proposition}
A host of $\mca$ algorithms are designed to de-noise and impute missing entries of matrices under random patterns of missingness; the most common missingness pattern studied is where each entry has the same probability of being missing, independent of everything else.
In contrast, $\ssMC$ generates patterns where all entries in one block are deterministically missing, as in \cref{figure_sample_split}$(b)$. 
A recent strand of research on the interplay between matrix completion methods and causal inference models---specifically, within the synthetic controls framework---has contributed matrix completion algorithms that allow for block missingness \citep[see, e.g.,][]{athey2021matrix, agarwal2021robustness, bai2021matrix, SI, arkhangelsky2021synthetic, agarwal2023causal, dwivedi2022counterfactual,dwivedi2022doubly}. 
However, it is a challenge to apply known theoretical guarantees for these methods to the setting in this article because of: (i) the use of cross-fitting---which creates blocks where all observations are missing---and (ii) outside of the completely-missing blocks, there can still be missing observations with heterogeneous probabilities of missingness.
In the next section, we show how to modify an $\mca$ algorithm designed for block missingness patterns so that it can be applied to our setting with cross-fitting and heterogeneous probabilities of missingness outside the folds. 
For concreteness, we work with the Tall-Wide matrix completion algorithm of \citet{bai2021matrix}.

\begin{figure}[t]
    \centering
    \begin{tabular}{cc}
    \includegraphics[trim={3.0cm 9.75cm 3.0cm 2.75cm}, width=0.35\linewidth,clip]{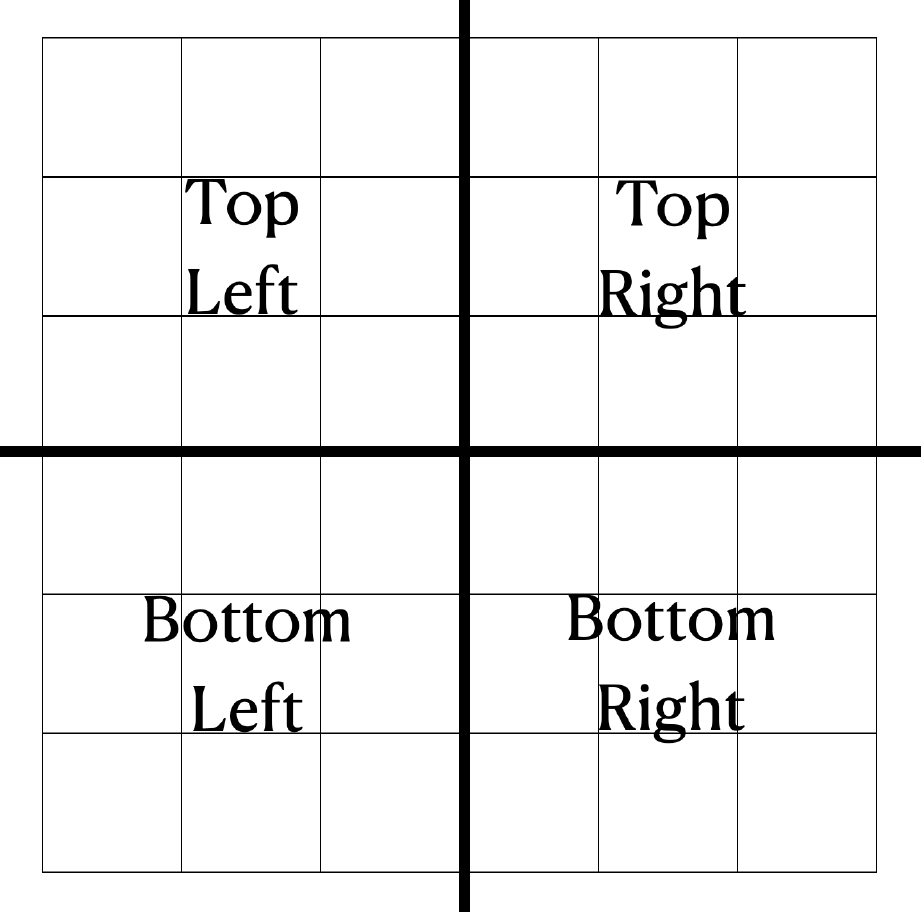} &
    \includegraphics[trim={3.0cm 9.75cm 3.0cm 2.75cm}, width=0.35\linewidth,clip]{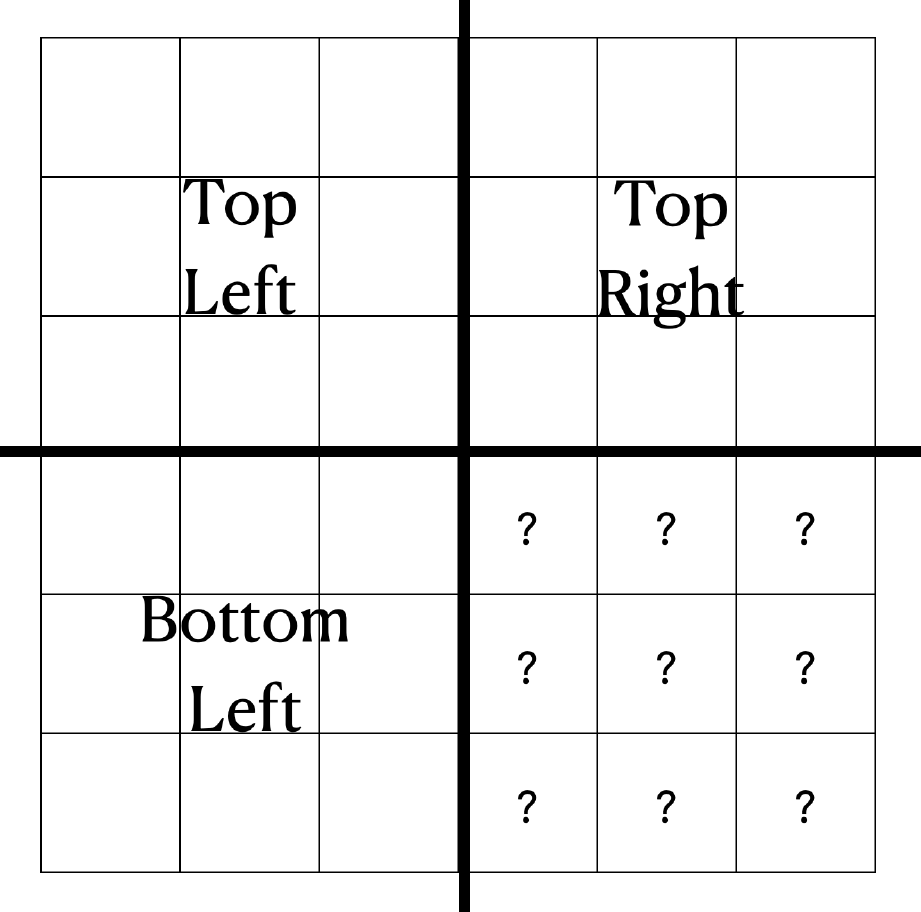} \\
    $(a)$ $\matS$ & $(b)$ $\matS  \otimes \allones^{-\text{Bottom Right}}$
    \end{tabular}
\caption{Panel $(a)$: A matrix $\matS$ partitioned into four blocks  when $\cR_0 = [N/2]$ and $\cC_0 = [M/2]$ in \cref{assumption_block_noise}, i.e., $\mc P = \normalbraces{\text{Top Left, Top Right, Bottom Left, Bottom Right}}$. Panel $(b)$: The matrix $\matS  \otimes \allones^{-\text{Bottom Right}}$ obtained from the matrix $\matS$  by masking the entries corresponding to the $\text{Bottom Right}$ block with $\star$.}
\label{figure_sample_split}
\end{figure}

\subsection{The \texorpdfstring{$\ssSVD$}{} algorithm} \label{sub:ssvd}

$\ssSVD$ is an end-to-end $\mca$ algorithm obtained by instantiating the $\ssMC$ meta-algorithm with the Tall-Wide algorithm of \citet{bai2021matrix}, which we denote as $\tallwide$.
For completeness, we detail the $\tallwide$ algorithm in \cref{sec:TW_algorithm}, and then use it to describe $\ssSVD$ in \cref{sec:ssvd_algorithm}.
\subsubsection{The \texorpdfstring{$\tallwide$}{} algorithm of \texorpdfstring{\citet{bai2021matrix}}{}.}\label{sec:TW_algorithm}
\citet{bai2021matrix} propose $\tallwide$ to impute missing values in matrices with a set of rows and a set of columns without missing entries. 
More concretely, for any matrix $\matS \in \normalbraces{\Reals \cup \{\star\}}^{N \times M}$, let $\cR_{\mathrm{obs}}  \subseteq [N]$ and $\cC_{\mathrm{obs}}  \subseteq [M]$ denote the set of all rows and all columns, respectively, with all entries observed. 
Then, all missing entries of $\matS$ belong to the block $\cI = \cR_{\mathrm{miss}} \times \cC_{\mathrm{miss}}$, where $\cR_{\mathrm{miss}} \defn [N] \setminus \cR_{\mathrm{obs}}$ and $\cC_{\mathrm{miss}} \defn [M]  \setminus \cC_{\mathrm{obs}}$.

Given a rank hyper-parameter $r \in [\min\normalbraces{\normalabs{\cR_{\mathrm{obs}}},\normalabs{\cC_{\mathrm{obs}}}}]$, $\tallwide_r$ produces an estimate of $\matT$ as follows: 

\begin{enumerate}
\item Run \texttt{SVD}
separately on $\matS^{\tall} \defn \matS_{[N] \times \cC_{\mathrm{obs}}}$ and $\matS^{\wide} \defn \matS_{\cR_{\mathrm{obs}} \times [M]}$, i.e.,
\begin{align}
\texttt{SVD}(\matS^{\tall})  & = (\matU^{\tall} \in \Reals^{N \times \wbar{r}_N}, \Sigma^{\tall} \in \Reals^{\wbar{r}_N \times \wbar{r}_N}, \matV^{\tall} \in \Reals^{\normalabs{\cC_{\mathrm{obs}}}  \times \wbar{r}_N})\\ \shortintertext{and}  
\texttt{SVD}(\matS^{\wide})  & = (\matU^{\wide} \in \Reals^{\normalabs{\cR_{\mathrm{obs}}} \times \wbar{r}_M}, \Sigma^{\wide} \in \Reals^{\wbar{r}_M \times \wbar{r}_M}, \matV^{\wide} \in \Reals^{M  \times \wbar{r}_M})\label{eq_svd_tw}
\end{align}
where $\wbar{r}_N \defn \min\normalbraces{ N, \normalabs{\cC_{\mathrm{obs}}}}$ and $\wbar{r}_M \defn \min\normalbraces{ \normalabs{\cR_{\mathrm{obs}}}, M }$. The
 columns of $\matU^{\tall}$ and $\matU^{\wide}$ are the left singular vectors of $\matS^{\tall}$ and $\matS^{\wide}$, respectively, and the columns of $\matV^{\tall}$ and $\matV^{\wide}$ are the right singular vectors of $\matS^{\tall}$ and $\matS^{\wide}$, respectively. 
 The diagonal entries of $\Sigma^{\tall}$ and $\Sigma^{\wide}$ are the singular values of $\matS^{\tall}$ and $\matS^{\wide}$, respectively, and the off-diagonal entries are zeros.
This step of $\tallwide$ requires the existence of the fully observed blocks $\matS^{\tall}$ and $\matS^{\wide}$, i.e., $\cR_{\mathrm{obs}}$ and $\cC_{\mathrm{obs}}$ cannot be empty. 
\item\label{step3_tallwide}
Let $\wtil{\matV}^{\tall} \in \mathbb R^{|\cC_{\mathrm{obs}}| \times r}$ 
be the sub-matrix of $\matV^{\tall}$ that keeps the columns corresponding to the $r$ largest singular values only. 
Let $\wtil{\matV}^{\wide} \in \mathbb R^{|\cC_{\mathrm{obs}}| \times r}$ 
be the sub-matrix of $\matV^{\wide}$ that keeps the columns corresponding to the $r$ largest singular values only and the rows corresponding to the indices in $\cC_{\mathrm{obs}}$ only.
Obtain a rotation matrix $\matR^{\miss} \in \Reals^{r \times r}$ as follows:
\begin{align}\label{eq:TW_rotation}
\matR^{\miss} \defn 
\wtil{\matV}^{\tall\top} \wtil{\matV}^{\wide}\bigparenth{\wtil{\matV}^{\wide\top} \wtil{\matV}^{\wide}}^{-1}.
\end{align}
That is, $\matR^{\miss}$ is obtained by regressing $\wtil{\matV}^{\tall}$ on $\wtil{\matV}^{\wide}$.
In essence, $\matR^{\miss}$ aligns the right singular vectors of  $\matS^{\tall}$ and $\matS^{\wide}$ using the entries that are common between these two matrices, i.e., the entries corresponding to indices $\cR_{\mathrm{obs}} \times \cC_{\mathrm{obs}}$. 
The formal guarantees of the $\tallwide$ algorithm remains unchanged if one alternatively regresses $\wtil{\matV}^{\wide}$ on $\wtil{\matV}^{\tall}$, or uses the left singular vectors of $\matS^{\tall}$ and $\matS^{\wide}$ for alignment.

\item Let $\wbar{\Sigma}^{\tall} \in \mathbb R^{\wbar{r}_N \times r}$ 
be the sub-matrix of $\Sigma^{\tall}$ that keeps the columns corresponding to the $r$ largest singular values only.
Let $\wbar{\matV}^{\wide} \in \mathbb R^{M \times r}$ be the sub-matrix of $\matV^{\wide}$ that keeps the columns corresponding to the $r$ largest singular values only.
Return $\what{\matT} \defn \matU^{\tall} \wbar{\Sigma}^{\tall} \matR^{\miss} \wbar{\matV}^{\wide\top}$ as an estimate for  $\matT$.
\end{enumerate}

\subsubsection{\texorpdfstring{$\ssSVD$}{} algorithm.}\label{sec:ssvd_algorithm} \medskip
\begin{enumerate}
    \item The inputs are $(i)$ $\ta \in \Reals^{N \times M}$, $(ii)$ $\ooa \in \normalbraces{\Reals \cup \{\star\}}^{N \times M}$ for $a \in \normalbraces{0,1}$, $(iii)$ a block partition $\cP$ of the set $[N] \times [M]$ into four blocks as in  \cref{assumption_block_noise}, and $(iv)$ hyper-parameters $r_1$, $r_2$, $r_3$, and $\lbar$ such that $r_1, r_2, r_3 \in [\min\normalbraces{N,M}]$ and $0 < \lbar \leq 1/2$. 
    \item\label{item_ssSVD_P} Return $\emta = \texttt{Proj}_{\lbar} \bigparenth{\ssMC(\tallwide_{r_1}, \ta, \cP)}$ 
    where $\texttt{Proj}_{\lbar}(\cdot)$ projects each entry of its input to the interval $[\lbar,1-\lbar]$.
    \item\label{step_definition} Define $\barz$ as equal to $\ooz$, but with all missing entries in $\ooz$ set to zero. Define $\baro$ analogously with respect to $\ooo$.
    \item\label{item_ssSVD_Theta0} Return $\empoz = \ssMC(\tallwide_{r_2}, \barz, \cP) \odiv (\allones - \emta)$. 
    \item\label{item_ssSVD_Theta1} Return $\empoo = \ssMC(\tallwide_{r_3}, \baro, \cP) \odiv \emta$.
\end{enumerate}
We provide intuition on the key steps of the $\ssSVD$ algorithm next.\medskip

\begin{figure}[t]
    \centering
    \begin{tabular}{ccc}
    \includegraphics[trim={3cm 7cm 2.5cm 2.75cm}, width=0.29\linewidth,clip]{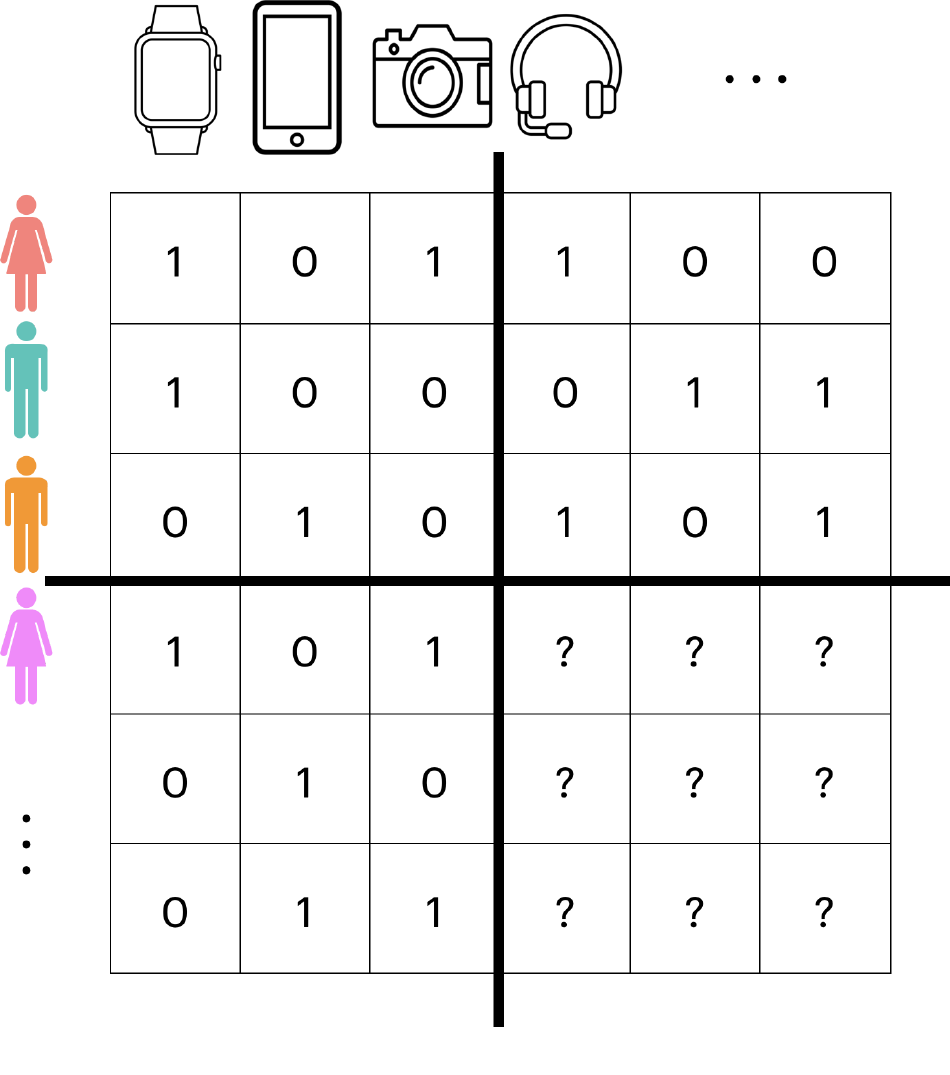}&
    \includegraphics[trim={3cm 7cm 2.5cm 2.75cm}, width=0.29\linewidth,clip]{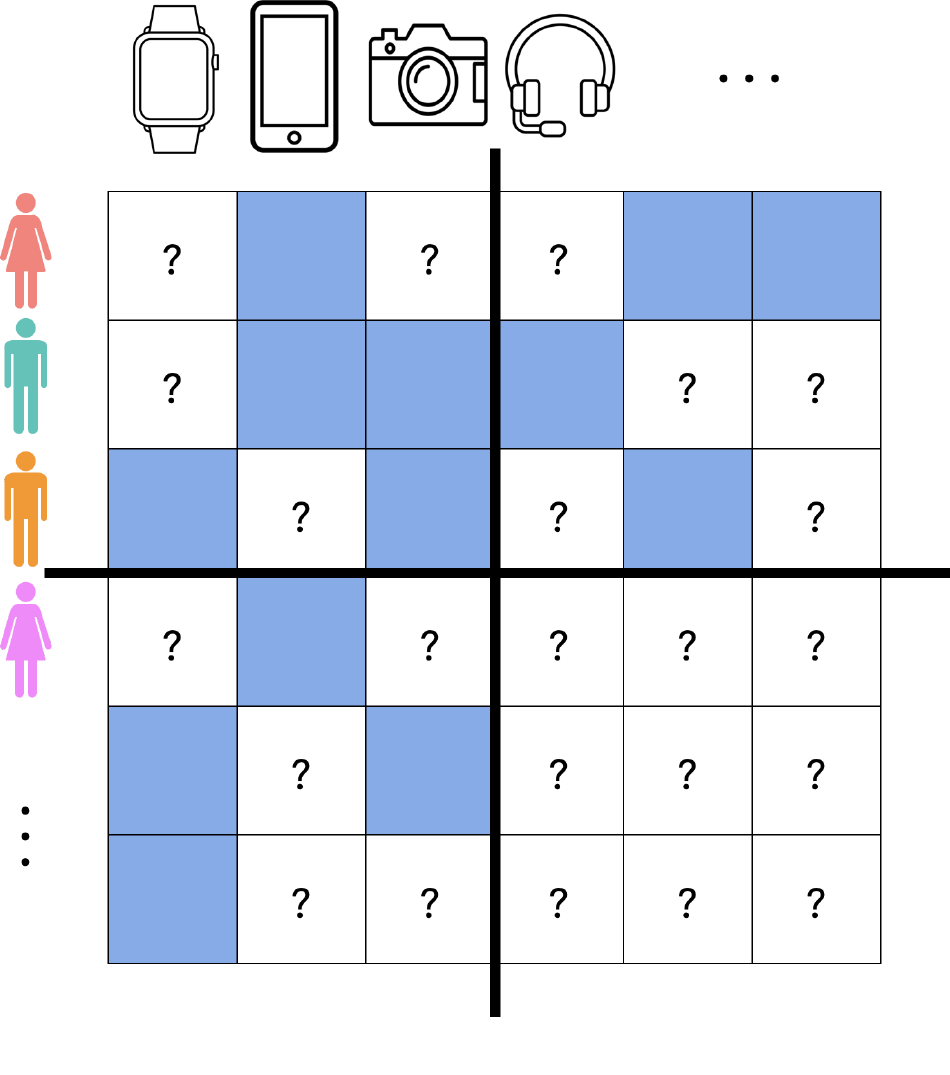}&
    \includegraphics[trim={3cm 7cm 2.5cm 2.75cm}, width=0.29\linewidth,clip]{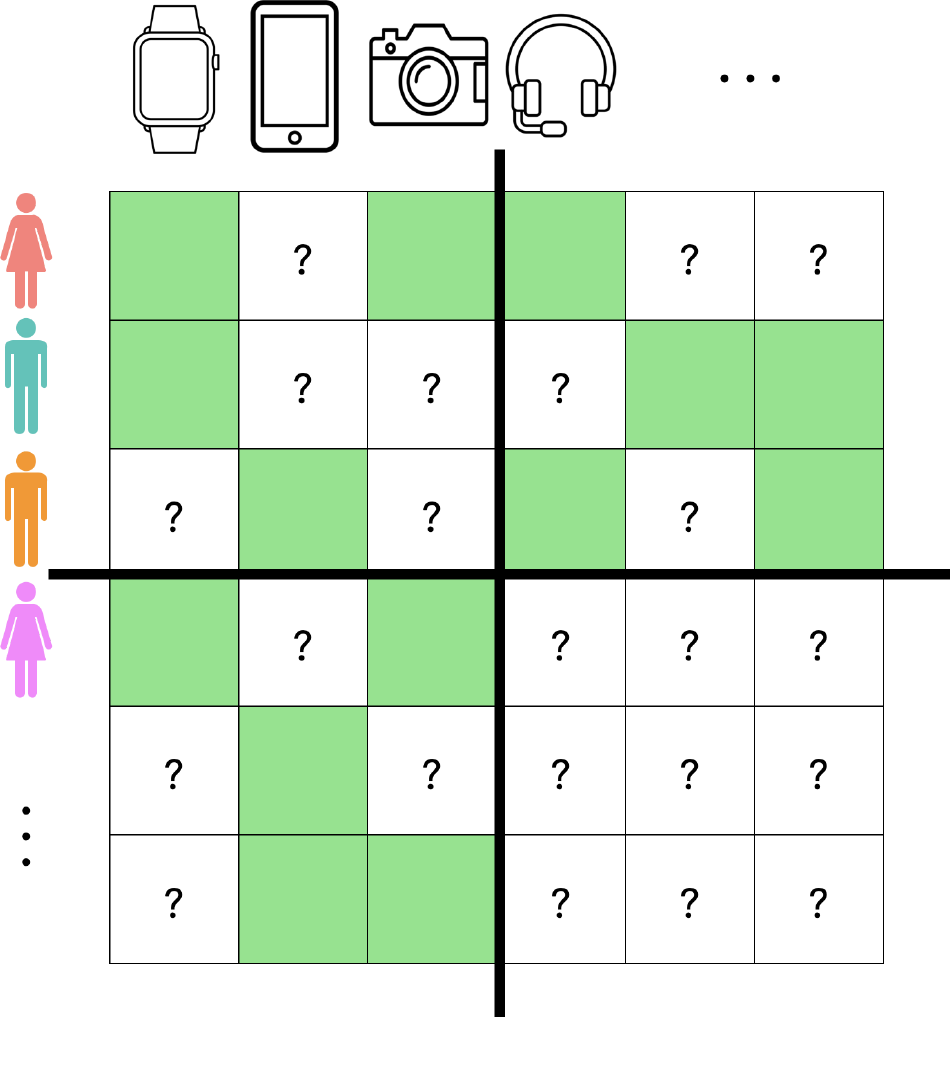}\\  
    $(a)$ $\ta  \otimes \allones^{-\text{Bottom Right}}$ & $(b)$ $\ooz  \otimes \allones^{-\text{Bottom Right}}$ &  $(c)$ $\ooo  \otimes \allones^{-\text{Bottom Right}}$
    \end{tabular}
\caption{
Panels $(a)$, $(b)$, and $(c)$ illustrate the matrices $\ta  \otimes \allones^{-\cI}$, $\ooz  \otimes \allones^{-\cI}$, and $\ooo  \otimes \allones^{-\cI}$ obtained from $\ta$, $\ooz$ and $\ooo$, respectively, for the block partition $\mc P$ in \cref{figure_sample_split}$(a)$ and the block $\cI = \text{Bottom Right}$. Unlike Panels $(b)$ and $(c)$, there exists rows and columns with all entries observed in Panel $(a)$. To enable the application of $\tallwide$ for Panels $(b)$ and $(c)$, we replace missing entries in blocks $\text{Top Left}$, $\text{Top Right}$, and $\text{Bottom Left}$ with zeros.}
\label{figure_algorithm}
\end{figure}

\noindent {\bf Computing $\emta$.} 
The estimate $\emta$ comes from applying $\ssMC$ with $\tallwide$ on $\ta$ and truncating the entries of the resulting matrix to the range $[\lbar, 1-\lbar]$, in accordance with \cref{assumption_pos_estimated}.
The $\tallwide$ sub-routine is directly applicable to $\ta$, because for any block $\cI = \cR_s \times \cC_k \in \cP$ the masked matrix $\ta \otimes \allones^{-\mc I}$ has $[N] \setminus \cR_s$ fully observed rows and $[M] \setminus \cC_k$ fully observed columns.
See \cref{figure_algorithm}$(a)$ for a visualization of $\ta  \otimes \allones^{-\mc I}$. \medskip

\noindent {\bf Computing $\empoz$ and $\empoo$.} 
The estimates $\empoz$ and $\empoo$ are constructed by applying $\ssMC$ with $\tallwide$ on $\barz$ and $\baro$, which do not have missing entries.
$\tallwide$ is not directly applicable on $\ooz$ and $\ooo$, as both matrices may not have any rows and columns that are fully observed. See \cref{figure_algorithm}$(b)$ and \cref{figure_algorithm}$(c)$ for visualizations of $\ooz \otimes \allones^{-\cI}$ and $\ooo \otimes \allones^{-\cI}$, respectively.
However, notice that, due to \cref{assumption_noise}\cref{item_ass_2aa} and \cref{assumption_noise}\cref{item_ass_2bb},
\begin{align}
\Expectation[\barz] & = \Expectation[\oo \odot (\allones - \ta)] = \mpoz \odot (\allones - \mta), 
\shortintertext{and}
\Expectation[\baro] & = \Expectation[\oo \odot \ta] = \mpoo \odot \mta.
\end{align}
As a result, $\mca(\barz)$ and $\mca(\baro)$ provide estimates of $\mpoz \odot (\allones - \mta)$ and $\mpoo \odot \mta$, respectively---recall the discussion in \cref{subsec_mc_primer}. 
To construct $\empoz$ and $\empoo$, we divide the entries of $\mca(\barz)$ and $\mca(\baro)$ by the entries of $(\allones - \emta)$ and $\emta$, respectively, to adjust for heterogeneous probabilities of missingness \citep[see, e.g.,][for related procedures]{jin2021factor,bhattacharya2022matrix,xiong2023large}. This inverse probability of treatment weighting adjustment to estimate $\empoz$ and $\empoo$  is distinct and in addition to the augmented IPW procedure that generates $\ATEDR$ from estimates $\empoz$, $\empoo$ and $\emta$.

\subsection{Theoretical guarantees for \texorpdfstring{$\ssSVD$}{}}
\label{sub:lf}

To establish theoretical guarantees for $\ssSVD$, we adopt three assumptions from \citet{bai2021matrix}. The first assumption imposes a low-rank structure on the matrices $\mta$, $\mpoz$, and $\mpoo$, namely that their entries are given by an inner product of latent factors. 

\begin{assumption}[Linear latent factor model on the confounders]
\label{assumption_lf}
There exist constants $r_{p}, r_{\theta_0}, r_{\theta_1} \in [\min\sbraces{N, M}]$ and a collection of latent factors
\begin{align}
	\matU \in \Reals^{N \times r_p}, 
	\quad \matV \in \Reals^{M \times r_p},
	\quad
	\matU^{(a)} \in \Reals^{N \times r_{\theta_a}},
	\qtext{and} 
	\matV^{(a)} \in \Reals^{M \times r_{\theta_a}}
	\qtext{for} a \in \sbraces{0, 1},
\end{align}
such that the unobserved confounders $(\mpob, \mta)$ satisfy the following factorization:
\begin{align}\label{eq:LFM_explicit}
	\mta = \matU \matV^{\top}
	\qtext{and}
	\mpo = \matU^{(a)} \matV^{(a)\top}
	\qtext{for} a \in \sbraces{0, 1}.
\end{align}
\end{assumption}

\cref{assumption_lf} decomposes each of the unobserved confounders ($\mta$, $\mpoz$, and $\mpoz$) into low-dimensional unit-dependent latent factors ($\matU$, $\matU^{(0)}$, and $\matU^{(1)}$) and measurement-dependent latent factors ($\matV$, $\matV^{(0)}$, and $\matV^{(1)}$). 
In particular, every unit $i \in [N]$ is associated with three low-dimensional factors: $(i)$ $U_i \in \Reals^{r_p}$, $(ii)$ $U_i^{(0)} \in \Reals^{r_{\theta_0}}$, and $(iii)$ $U_i^{(1)} \in \Reals^{r_{\theta_1}}$. Similarly, every measurement $j \in [M]$ is associated with three factors: $(i)$ $V_j \in \Reals^{r_p}$, $(ii)$ $V_j^{(0)} \in \Reals^{r_{\theta_0}}$, and $(iii)$ $V_j^{(1)} \in \Reals^{r_{\theta_1}}$. Low-rank assumptions are standard in the matrix completion literature.

The second assumption requires that the factors that determine $\mta$, $\mpoz \odot (\allones - \mta)$, and $\mpoo \odot \mta$ explain a sufficiently large amount of the variation in the data. This assumption is made on the factors of $\mpoz \odot (\allones - \mta)$ and $\mpoo \odot \mta$ instead of $\mpoz$ and $\mpoo$ as the $\tallwide$ algorithm is applied on $\barz = \oo \odot (\allones - \ta)$ and $\baro = \oo \odot \ta$, instead of $\ooz$ and $\ooo$ (see steps \cref{item_ssSVD_Theta0} and \cref{item_ssSVD_Theta1} of $\ssSVD$). To determine the factors of $\mpoz \odot (\allones - \mta)$ and $\mpoo \odot \mta$, let
\begin{align}
    \wbar{\matU} \defn \normalbrackets{\ones, -\matU} \in \Reals^{N \times (r_p + 1)} \qtext{and} \wbar{\matV} \defn \normalbrackets{\ones[M], \matV} \in \Reals^{M \times (r_p + 1)},
\end{align}
where $\ones \in \Reals^N$ and $\ones[M] \in \Reals^M$ are vectors of all $1$'s.
Then, 
\begin{align}
    \mpoz \odot (\allones - \mta) = \wbar{\matU}^{(0)} \wbar{\matV}^{(0)\top} \qtext{and} \mpoo \odot \mta = \wbar{\matU}^{(1)} \wbar{\matV}^{(1)\top}, \label{eq_factors_theta_p}
\end{align}
where $\wbar{\matU}^{(0)} \defn \wbar{\matU} * \matU^{(0)} \in \Reals^{N \times r_{\theta_0}(r_p + 1)}$, $\wbar{\matV}^{(0)} \defn \wbar{\matV} * \matV^{(0)}  \in \Reals^{M \times r_{\theta_0}(r_p + 1)}$, $\wbar{\matU}^{(1)} \defn \matU * \matU^{(1)} \in \Reals^{N \times r_{\theta_1}r_p}$, and  $\wbar{\matV}^{(1)} \defn \matV * \matV^{(1)} \in \Reals^{M \times r_{\theta_1}r_p}$, with the operator $*$ denoting the Khatri-Rao product (see \cref{section_introduction}).
We provide details of the derivation of these factors in \cref{subsubsec_claim_p_t_factors}.
\begin{assumption}[Strong factors]
\label{assumption_lf_strong}
There exists a positive constant $c$ such that
\begin{align}
    \stwoinfnorm{\matU} \leq c, \quad \stwoinfnorm{\matV} \leq c, \quad \stwoinfnorm{\matU^{(a)}} \leq c, \qtext{and}  \stwoinfnorm{\matV^{(a)}} \leq c \qtext{for} a \in \sbraces{0, 1}.\label{eq_strong_factors}
\end{align}
Further, the matrices defined below exist and are  positive definite: 
\begin{align}
\lim_{N \to \infty} \frac{\matU^{\top} \matU}{N},  \!\!\quad  
\lim_{M \to \infty} \frac{\matV^{\top} \matV}{M}, \!\!\quad 
\lim_{N \to \infty} \frac{\wbar{\matU}^{(a)\top} \wbar{\matU}^{(a)}}{N}, \!\!\qtext{and}
\lim_{M \to \infty} \frac{\wbar{\matV}^{(a)\top} \wbar{\matV}^{(a)}}{M} \!\!\qtext{for} a \in \sbraces{0, 1}. \label{eq_psd_factors}
\end{align}
\end{assumption}
\cref{assumption_lf_strong}, a classic assumption in the literature on latent factor models, ensures that the factor structure is strong. Specifically, it ensures that each eigenvector of $\mta$, $\mpoz \odot (\allones - \mta)$, and $\mpoo \odot \mta$ carries sufficiently large signal.

The third assumption requires a strong factor structure on the sub-matrices of $\mta$, $\mpoz \odot (\allones - \mta)$, and $\mpoo \odot \mta$ corresponding to every block $\cI$ in the block partition $\cP$ from \cref{assumption_block_noise}. Further, it also requires that the size $\cI$ grows linearly in $N$ and $M$.
\begin{assumption}[Strong block factors]
\label{ass_stron_block_factors}
Consider the block partition $\cP \defn \normalbraces{\cR_s \times \cC_k: s,k \in \normalbraces{0,1}}$ from \cref{assumption_block_noise}. For every $s \in \normalbraces{0,1}$, let $\matU_{(s)} \in \Reals^{\normalabs{\cR_s} \times r_p}$, $\wbar{\matU}^{(0)}_{(s)} \in \Reals^{\normalabs{\cR_s} \times r_{\theta_0}(r_p + 1)}$, and $\wbar{\matU}^{(1)}_{(s)} \in \Reals^{\normalabs{\cR_s} \times r_{\theta_1}r_p}$ be the sub-matrices of $\matU$, $\wbar{\matU}^{(0)}$, and $\wbar{\matU}^{(1)}$, respectively, that keeps the rows corresponding to the indices in $\cR_s$. For every $k \in \normalbraces{0,1}$, let $\matV_{(k)} \in \Reals^{\normalabs{\cC_k} \times r_p}$, $\wbar{\matV}^{(0)}_{(k)} \in \Reals^{\normalabs{\cC_k} \times r_{\theta_0}(r_p + 1)}$, and $\wbar{\matV}^{(1)}_{(k)} \in \Reals^{\normalabs{\cC_k} \times r_{\theta_1}r_p}$ be the sub-matrices of $\matV$, $\wbar{\matV}^{(0)}$, and $\wbar{\matV}^{(1)}$, respectively, that keeps the rows corresponding to the indices in $\cC_k$.
    Then, for every $s,k \in \normalbraces{0,1}$, the matrices defined below exist and are positive definite:  
    \begin{align}
    \lim_{N \to \infty} \frac{\matU_{(s)}^{\top} \matU_{(s)}}{\normalabs{\cR_s}},  \!\!\quad  
\lim_{M \to \infty} \frac{\matV_{(k)}^{\top} \matV_{(k)}}{\normalabs{\cC_k}}, \!\!\quad 
\lim_{N \to \infty} \frac{\wbar{\matU}_{(s)}^{(a)\top} \wbar{\matU}_{(s)}^{(a)}}{\normalabs{\cR_s}}, \!\!\qtext{and}
\lim_{M \to \infty} \frac{\wbar{\matV}_{(k)}^{(a)\top} \wbar{\matV}_{(k)}^{(a)}}{\normalabs{\cC_k}} \!\!\qtext{for} a \in \sbraces{0, 1}. \label{eq_strong_block_factors_main}
    \end{align}
    Further, for every $s,k \in \normalbraces{0,1}$, 
    $\normalabs{\cR_s}= \Omega(N)$ and $\normalabs{\cC_k}= \Omega(M)$.
\end{assumption}

The subsequent assumption introduces additional conditions on the noise variables in \cite{bai2021matrix} than those specified in \cref{assumption_noise,assumption_block_noise}.
\begin{assumption}[Weak dependence in noise across measurements and independence in noise across units]
\label{ass_weak_noise}
{\color{white}.}
\begin{enumerate}[itemsep=1mm, topsep=2mm, label=(\alph*)]
    \item\label{item_weak_ass_2aa} $\sum_{j' \in [M]} \bigabs{\Expectation\normalbrackets{\eta_{i,j} \eta_{i,j'}}} \leq c$ for every $i \in [N]$ and $j \in [M]$,
    \item\label{item_weak_ass_2bb} $\sum_{j' \in [M]} \bigabs{\Expectation\normalbrackets{\wbar{\varepsilon}_{i,j}^{(a)} \wbar{\varepsilon}_{i,j'}^{(a)}}} \leq c$ for every $i \in [N]$, $j \in [M]$, and $a \in \normalbraces{0,1}$, where $ \wbar{\varepsilon}_{i,j}^{(a)} \defn \theta_{i,j} \eta_{i,j} + \varepsilon_{i,j}^{(a)} p_{i,j} + \varepsilon_{i,j}^{(a)} \eta_{i,j}$, and
    \item\label{item_weak_ass_2cc}  The elements of $\sbraces{\normalparenth{\noisey_{i,\cdot}, \noisea_{i,\cdot}} : i\in[N]}$ are mutually independent (across $i$) for $a \in \normalbraces{0,1}$.
\end{enumerate}
\end{assumption}
\cref{ass_weak_noise}\cref{item_weak_ass_2aa} and \cref{ass_weak_noise}\cref{item_weak_ass_2bb} requires the noise variables to exhibit only weak dependency across measurements. 
Still, these assumptions allow the existence of pairs of perfectly correlated outcomes (e.g., $j,j'\in [M]$ such that $a_{i,j}=a_{i,j'}$).   \cref{ass_weak_noise}\cref{item_weak_ass_2cc}  requires the noise $\normalparenth{\noisey, \noisea}$ to be jointly independent across units, for every $a \in \normalbraces{0,1}$.
We are now ready to provide guarantees on the estimates produced by $\ssSVD$.
The proof can be found in \cref{proof_bai_ng_prop}.
\begin{proposition}[\textbf{\propsssvd}]
\label{alg_guarantee}
Suppose \cref{assumption_pos,assumption_noise,assumption_lf,assumption_lf_strong,ass_stron_block_factors,ass_weak_noise} hold. 
\ifx\submission\arxiv
Consider an asymptotic sequence such that $\thetamax$ is bounded as both $N$ and $M$ increase.
\fi
\ifx\submission\ec
Suppose $\thetamax$ is bounded. 
\fi
Let $\emta$, $\empoz$, and $\empoo$ be the estimates returned by $\ssSVD$ with the block partition $\cP$ from \cref{assumption_block_noise}, $r_1 = r_p$, $r_2 = r_{\theta_0}(r_p + 1)$,  $r_3 = r_{\theta_1}r_p$, and any $\lbar$ such that $0 < \lbar \leq \lambda$ with $\lambda$ denoting the constant from \cref{assumption_pos}.
Then, as $N, M  \to \infty$,
\begin{align}
\RP = O_p\biggparenth{\frac{1}{\sqrt{N}} + \frac{1}{\sqrt{M}}} \qtext{and} \RTheta = O_p\biggparenth{\frac{1}{\sqrt{N}} + \frac{1}{\sqrt{M}}}.
\end{align}
\end{proposition}
\cref{alg_guarantee} implies that the conditions \cref{item_error_for_normality,item_product_of_error_for_normality} in \cref{thm_normality} hold whenever $N^{1/2}/M=o(1)$.
Then, the DR estimator from \cref{eq_counterfactual_mean_dr} constructed using $\ssSVD$  estimates $\empoz$, $\empoo$, and $\emta$ exhibits an asymptotic Gaussian distribution centered at the target causal estimand.  Further, 
\cref{alg_guarantee} implies that the estimation errors $\RP$ and $\RTheta$ achieve the parametric rate whenever $N/M = O(1)$.

\subsection{Application to panel data with staggered adoption}
\label{sec_stag_adop_short}
\cref{sub:kbb} considered a setting with block independence between noise (formalized in \cref{assumption_block_noise}). 
\cref{sec_stag_adop_app} discusses how to extend the proposed doubly-robust framework to a setting of panel data with staggered adoption, where this assumption may not hold.
Recall (from \cref{extension_delayed}) that in the panel data setting $M$ measurements correspond to $T$ time periods, and $t$ denotes the time index. Then, \cref{sec_stag_adop_app} considers a setting where a unit remains under control for some period of time, after which it deterministically remains under treatment. In other words, for every unit $i \in [N]$, there exists a time point $t_i \in [T]$ such that $a_{i, t} = 0$ for $t \leq t_i$, and $a_{i, t} = 1$ for $t > t_i$. Such a treatment assignment pattern leads to a heavy dependence in the noise $\{\eta_{i, t}\}_{t \in [T]}$ for every unit $i \in [N]$. \cref{sec_stag_adop_app} describes an alternative approach to the $\ssSVD$ algorithm and shows that \cref{assumption_estimates} still holds for a suitable staggered adoption model.

\section{Simulations}
\label{sec_simulations}
This section reports simulation results on the performance of the DR estimator of \cref{eq_counterfactual_mean_dr} and the OI and IPW  estimators of \cref{eq_counterfactual_mean_oi,eq_counterfactual_mean_ipw}, respectively.\medskip

\noindent {\bf Data Generating Process (DGP).}
We now briefly describe the DGP for our simulations; \cref{app_dgp} provides details.
All simulations set $N=M$. 
To generate, $\mta$, $\mpoz$, and $\mpoo$, we use the latent factor model given in \cref{eq:LFM_explicit}.
To introduce unobserved confounding, we set the unit-specific latent factors to be the same across $\mta$, $\mpoz$, and $\mpoo$, i.e., $U = U^{(0)} = U^{(1)}$.
The entries of $U$ and the measurement-specific latent factors, $V, V^{(0)}, V^{(1)}$ are each sampled independently from a uniform distribution, with hyperparameter $r_p$ equal to the dimension of $U$ and $V$, and hyperparameter $r_p$ equal to the dimension of $U^{(a)}$ and $V^{(a)}$ for $a=0,1$. 
Further, the entries of the noise matrices $\noiseyz$ and $\noiseyo$ are sampled independently from a normal distribution, and the entries of $\noisea$ are sampled independently as in \cref{eq_eta_defn}. Then, $y^{(a)}_{i, j}$, $a_{i,j}$, and $y_{i,j}$ are determined from \cref{eq_outcome_model,eq_treatment_model,eq_consistency}, respectively. 
The simulation generates $\mta$, $\mpoz$, and $\mpoo$ once. Given the fixed values of $\mta$, $\mpoz$, and $\mpoo$, the simulation generates $2500$ realizations of $(Y,A)$---that is, only the noise matrices $\noiseyz, \noiseyo, \noisea$ are resampled for each of the $2500$ realizations.
For each simulation realization, we apply 
the $\ssSVD$ algorithm with hyper-parameters as in \cref{alg_guarantee} and $\lbar = \lambda = 0.05$ to obtain
 $\emta$, $\empoz$, and $\empoo$, and compute $\ATETrue$ from \cref{eq_ate_parameter_combined}, and $\ATEOI$, $\ATEIPW$ and $\ATEDR$ from
\cref{eq_counterfactual_mean_oi,eq_counterfactual_mean_ipw,eq_counterfactual_mean_dr}. 

\begin{figure}[t!]
    \centering
    \begin{tabular}{cc}
    \includegraphics[width=0.45\linewidth,clip]{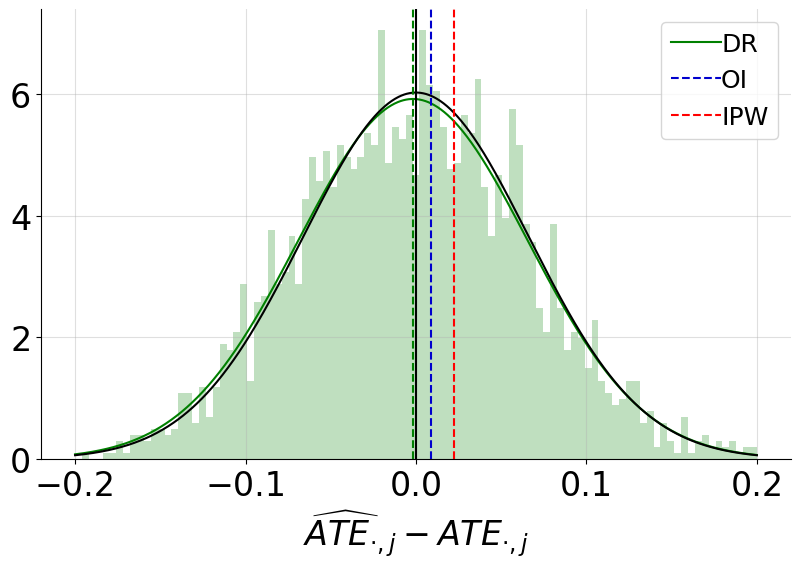} &
    \includegraphics[width=0.45\linewidth,clip]{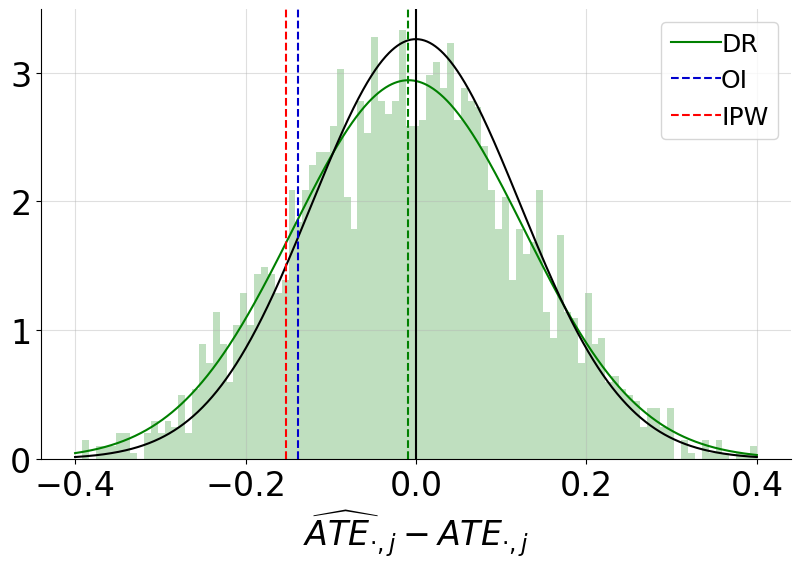}\\
    $(a)$ $r_p = 3 $, $ r_{\theta} = 3$ & $(b)$ $r_p = 5 $, $ r_{\theta} = 3$
    \end{tabular}
\caption{Empirical illustration of the asymptotic performance of DR as in \cref{thm_normality}. The histogram corresponds to the errors of 2500 independent instances of DR estimates, the green curve represents the (best) fitted Gaussian distribution, and the black curve represents the Gaussian approximation from \cref{thm_normality}. The dashed green, blue, and red lines represent the biases of DR, OI, and IPW estimators.
}
\label{figure_thm2_simulations}
\end{figure}

\begin{figure}[ht!]
    \centering
    \begin{tabular}{cc}
    \includegraphics[width=0.45\linewidth,clip]{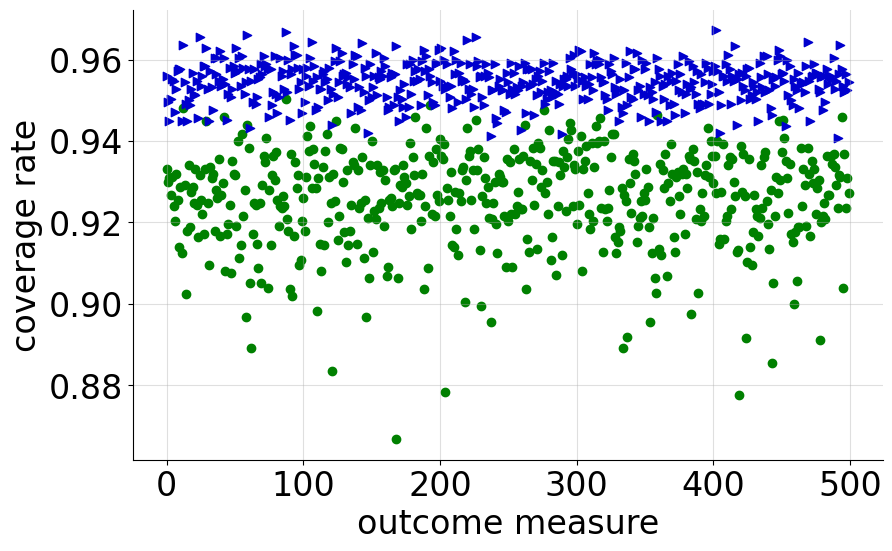} &
    \includegraphics[width=0.45\linewidth,clip]{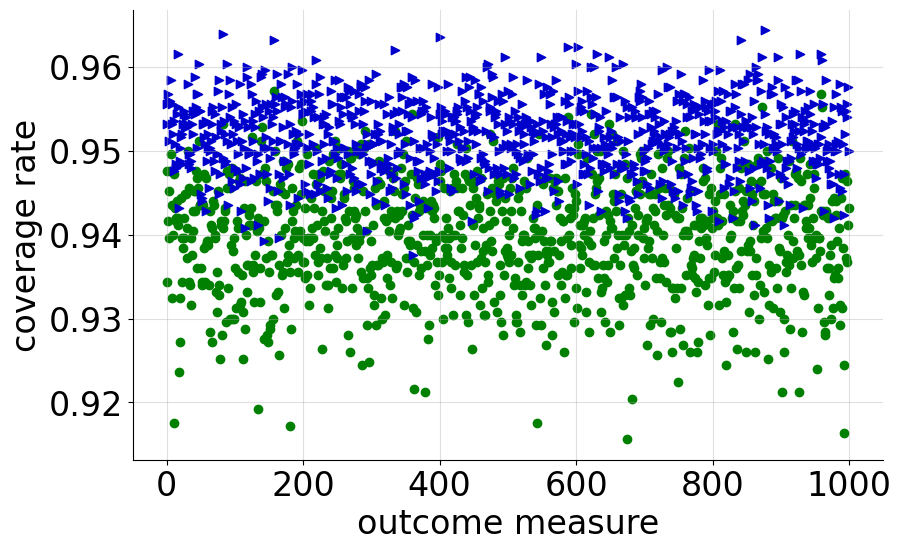}\\
    $(a)$ $N = 500$ & $(b)$ $N = 1000$ \\[1.5ex] \includegraphics[width=0.45\linewidth,clip]{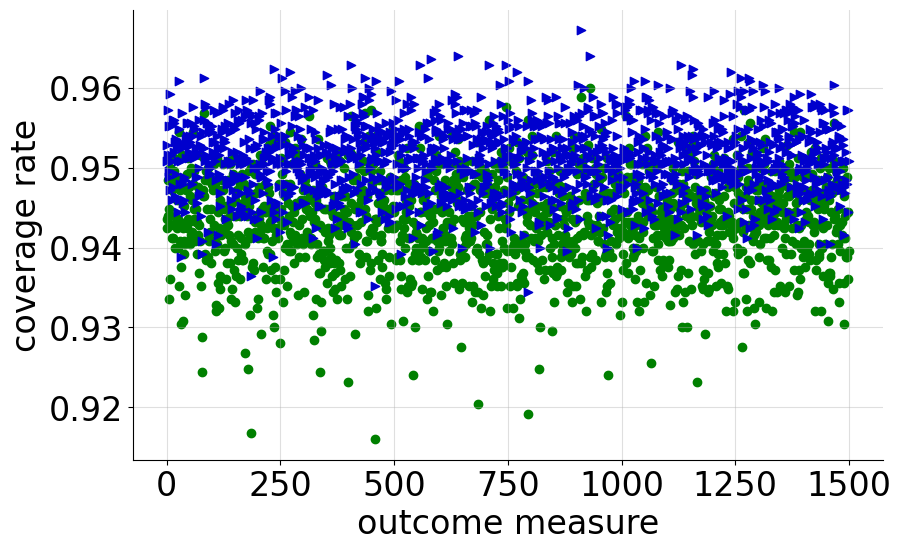} &
    \includegraphics[width=0.45\linewidth,clip]{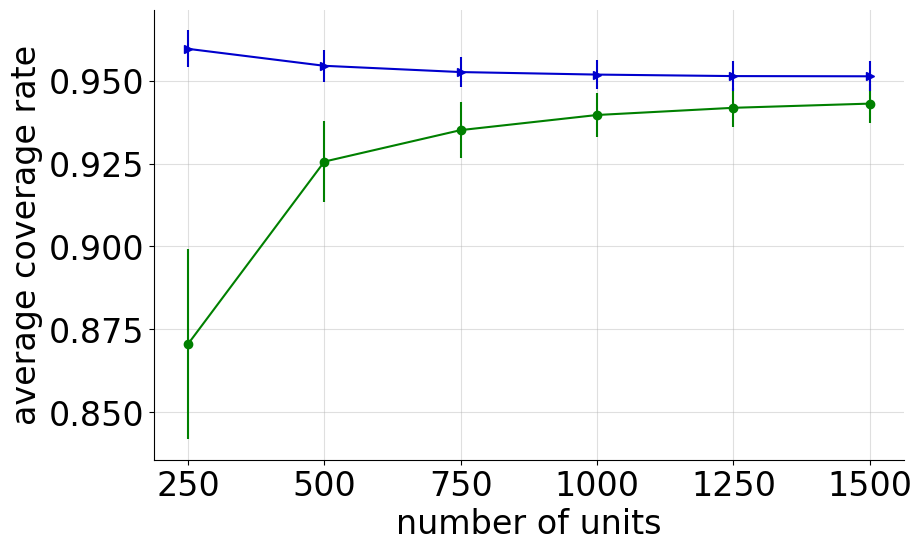} \\
    $(c)$ $N = 1500$ & $(d)$ Average coverage across outcomes
    \end{tabular}
    \vspace{2mm}
    \centering
    \vspace{2mm}
    \includegraphics[width=0.20\linewidth,clip]{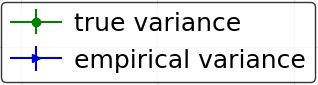} 
\caption{Panels $(a)$, $(b)$, and $(c)$ report  coverage rates for  nominal $95\%$ confidence intervals constructed using the estimated variance from \cref{eq_estimated_variance} (in blue) and the true variance from \cref{eq_normality_variance_thm} (in green) for $N \in \normalbraces{500, 1000, 1500}$ and $M=N$. Panel $(d)$ shows the means and standard deviations of coverage rates across outcomes for different values of $N$.}
\label{figure_coverage}
\end{figure}
\medskip

\noindent {\bf Results.}
\cref{figure_thm2_simulations} reports simulation results for $N = 1000$, with $r_p = 3$, $ r_{\theta} = 3$ in Panel $(a)$, and $r_p = 5$, $ r_{\theta} = 3$ in Panel $(b)$. \cref{figure_comparison} in \cref{sec_estimators} reports simulation results for $r_p = 3$, $ r_{\theta} = 5$. In each case, the figure shows a histogram of the distribution of $\ATEDR - \ATETrue$ across $2500$ simulation instances for a fixed $j$, along with the best fitting Gaussian distribution (green curve). The histogram counts are normalized so that the area under the histogram integrates to one. \cref{figure_thm2_simulations} plots the Gaussian distribution in the result of \cref{thm_normality} (black curve). 
The dashed blue, red and green lines in \cref{figure_thm2_simulations,figure_comparison} indicate the values of the means of the OI, IPW, and DR error, respectively, across simulation instances. For reference, we place a black solid line at zero.
The DR estimator has minimal bias and a close-to-Gaussian distribution. The biases of OI and IPW are non-negligible. In \cref{app_dgp}, we compare the biases and the standard deviations of OI, IPW, and DR across many $j$.

Panels $(a)$, $(b)$, and $(c)$ of \cref{figure_coverage} report coverage rates over the 2500 simulations for $\ATEDR$-centered nominal $95\%$ confidence intervals with  $N = 500$, $N = 1000$, and $N = 1500$, respectively, all with $M=N$ and $r_p = r_{\theta} = 3$. For every $j \in [M]$, panels $(a)$, $(b)$ and $(c)$ show $\what{\mathrm{c}}_j$, the percentage of times $[\ATEDR \pm 1.96 \what{\sigma}_j/\sqrt{N}]$ covers $\ATETrue$ (in blue), and ${\mathrm{c}}_j$, the percentage of times $[\ATEDR \pm 1.96 {\sigma}_j/\sqrt{N}]$ covers $\ATETrue$ (in green). Panel $(d)$ shows the means and standard deviations of $\normalbraces{\what{\mathrm{c}}_j}_{j \in [M]}$ and $\normalbraces{{\mathrm{c}}_j}_{j \in [M]}$ for different values of $N$.  
Confidence intervals based on the large-sample approximation results of \cref{sec_main_results} exhibit small size distortion even for fairly small values of $N$.

\section{Conclusion}
\label{sec_discussion}

This article introduces a new framework to estimate treatment effects in the presence unobserved confounding. We consider modern data-rich environments, where there are many units, and outcomes of interest per unit. We show it is possible to control for the confounding effects of a set of latent variables when this set is low-dimensional relative to the number of observed treatments and outcomes.

Our proposed estimator is doubly-robust, combining outcome imputation and inverse probability weighting with matrix completion. Analytical tractability of its distribution is gained through 
a novel cross-fitting procedure for causal matrix completion. We study the properties of the doubly-robust estimator, along with the outcome imputation and inverse probability weighting-based estimators under black-box matrix completion error rates. We show that the decay rate of the error of the doubly-robust estimator dominates those of the outcome imputation and the inverse probability weighting estimators. Moreover, we establish a Gaussian approximation to the distribution of the doubly-robust estimator. 
Simulation results demonstrate the practical relevance of the formal properties of the doubly-robust estimator.

\pagebreak

\appendix
\section*{Appendices}
\renewcommand{\thesubsection}{\thesection.\arabic{subsection}}
\renewcommand{\theequation}{A.\arabic{equation}}
\setcounter{myassumption}{0}
\renewcommand{\themyassumption}{A.\arabic{myassumption}}
\setcounter{equation}{0}
\counterwithin{mylemma}{section}
\counterwithin{myproposition}{section}
\counterwithin{mycorollary}{section}
\section{Supporting Concentration and Convergence Results}

This section presents known results on $\subGaussian$, $\subExponential$, and $\subWeibull$ random variables (defined below), along with few basic results on convergence of random variables.

We use $\subG[\sigma]$ to represent a $\subGaussian$ random variable, where $\sigma$ is a bound on the $\subGaussian$ norm; and $\subE[\sigma]$ to represent a $\subExponential$ random variable,  where $\sigma$ is a bound on the $\subExponential$ norm. Recall the definitions of the norms from \cref{section_introduction}. 

\begin{lemma}[$\subGaussian$ concentration: Theorem 2.6.3 of \citet{vershynin2018high}]
\label{subg_conc}
    Let $x \in \Reals^{n}$ be a random vector whose entries are independent, zero-mean, $\subG[\sigma]$ random variables. 
    Then, for any $b \in \Reals^n$ and $t \geq 0$,
    \begin{align}
        \Probability\Bigbraces{\bigabs{b^\top x} \geq t} \leq 2\exp\Bigparenth{\frac{-ct^2}{\sigma^2 \stwonorm{b}^2}}.
    \end{align}
\end{lemma}
\noindent The following corollary expresses the bound in \cref{subg_conc} in a convenient form.
\begin{corollary}[$\subGaussian$ concentration]
\label{cor_subg_conc}
    Let $x \in \Reals^{n}$ be a random vector whose entries are independent, zero-mean, $\subG[\sigma]$ random variables. 
    Then, for any $b \in \Reals^n$ and any $\delta \in (0,1)$, with probability at least $1 - \delta$,
    \begin{align}
    \bigabs{b^\top x} \leq \sigma \sqrt{c \ld} \cdot \stwonorm{b}.
\end{align} 
\end{corollary}
\begin{proof}
    The proof follows from \cref{subg_conc} by choosing $\delta \defn 2\exp\normalparenth{-ct^2/\sigma^2 \stwonorm{b}^2}$.
\end{proof}

\begin{lemma}[$\subExponential$ concentration: Theorem 2.8.2 of \citet{vershynin2018high}]
\label{sube_conc}
    Let $x \in \Reals^{n}$ be a random vector whose entries are independent, zero-mean, $\subE[\sigma]$ random variables. 
    Then, for any $b \in \Reals^n$ and $t \geq 0$,
    \begin{align}
        \Probability\Bigbraces{\bigabs{b^\top x} \geq t} \leq 2\exp\biggparenth{-c\min\Bigparenth{\frac{t^2}{\sigma^2 \stwonorm{b}^2}, \frac{t}{\sigma \sinfnorm{b}}}}.
    \end{align}
\end{lemma}
\noindent The following corollary expresses the bound in \cref{sube_conc} in a convenient form.
\begin{corollary}[$\subExponential$ concentration]
\label{cor_sube_conc}
    Let $x \in \Reals^{n}$ be a random vector whose entries are independent, zero-mean, $\subE[\sigma]$ random variables. 
    Then, for any $b \in \Reals^n$ and any $\delta \in (0,1)$, with probability at least $1 - \delta$,
    \begin{align}
    \bigabs{b^\top x} \leq \sigma \m(c \ld)  \cdot \stwonorm{b},
\end{align} 
where recall that $\m(c \ld) = \max\bigparenth{c \ld, \sqrt{c \ld}}$.
\end{corollary}
\begin{proof}
    Choosing $t = t_0 \sigma \stwonorm{b}$ in \cref{sube_conc}, we have
    \begin{align}
        \Probability\Bigbraces{\bigabs{b^\top x} \geq t_0 \sigma \stwonorm{b}} & \leq 2\exp\Bigparenth{-c t_0 \min\Bigparenth{t_0, \frac{\stwonorm{b}}{\sinfnorm{b}}}}\\
        & \leq 2\exp\Bigparenth{-c t_0 \min\bigparenth{t_0, 1}},
    \end{align}
    where the second inequality follows from $\min\normalbraces{t_0, c} \geq \min\normalbraces{t_0, 1}$ for any $c \geq 1$ and $\stwonorm{b} \geq \sinfnorm{b}$. Then, the proof follows by choosing $\delta \defn 2\exp\bigparenth{-c t_0 \min\bigparenth{t_0, 1}}$ which fixes $t_0 = \max\normalbraces{\sqrt{c \ld}, c\ld} = \m(c\ld)$.
\end{proof}

\begin{lemma}[Product of $\subGaussian$s is $\subExponential$: Lemma. 2.7.7 of \citet{vershynin2018high}]
\label{lem_prod_subG}
    Let $x_1$ and $x_2$ be $\subG[\sigma_1]$ and $\subG[\sigma_2]$ random variables, respectively. Then, $x_1 x_2$ is $\subE[\sigma_1\sigma_2]$ random variable.
\end{lemma}

Next, we provide the definition of a $\subWeibull$ random variable.
\begin{definition}[$\subWeibull$ random variable: Definition 1 of \cite{zhang2022sharper}]
    For $\rho > 0$, a random variable $x$ is $\subWeibull$ with index $\rho$ if it has a bounded $\subWeibull$ norm defined as follows:
    \begin{align}
        \snorm{x}_{\psi_{\rho}} \defn \inf\normalbraces{t > 0: \Expectation[\exp(|x|^\rho/t^\rho)] \leq 2}.
    \end{align}
\end{definition}
We use $\subW[\sigma]$ to represent a $\subWeibull$ random variable with index $\rho$, where $\sigma$ is a bound on the $\subWeibull$ norm. We note that $\subGaussian$ and $\subExponential$ random variables are $\subWeibull$ random variable with indices $2$ and $1$, respectively.

\begin{lemma}[Product of $\subWeibull$s is $\subWeibull$: Proposition 2 of \citet{zhang2022sharper}]
\label{lem_prod_subW}
    For $i \in [d]$, let $x_i$ be a $\subWeibull_{\rho_i}(\sigma_i)$ random variable.  Then, $\Pi_{i \in [d]} x_i$ is $\subW[\sigma]$ random variable where
    \begin{align}
        \sigma = \Pi_{i \in [d]} \sigma_i \qtext{and} \rho = \Bigg(\sum_{i \in [d]} 1/\rho_i\Bigg)^{-1}.
    \end{align}
\end{lemma}

\noindent Next set of lemmas provide useful intermediate results on stochastic convergence.
\begin{lemma}
\label{lemma_conv_prob}
Let $X_n$ and $\wbar{X}_n$  be sequences of random variables. Let $\delta_n$ be a deterministic sequence such that $0\leq \delta_n \leq 1$ and $\delta_n \rightarrow 0$. Suppose $X_n=o_p(1)$ and $\Probability(|\wbar{X}_n|\leq |X_n|) \geq 1-\delta_n$.
Then, $\wbar{X}_n = o_p(1)$.
\end{lemma}

\begin{proof}
We need to show that for any $\epsilon>0$ and $\delta > 0$, there exist finite $\wbar n$, such that 
\[
\Probability(|\wbar{X}_n|>\delta)<\epsilon
\]
for all $n\geq \wbar n$. Fix any $\epsilon>0$.
Because $\delta_n$ converges to zero, there exists a finite $n_0$ such that $\delta_n<\epsilon/2$, for all $n\geq n_0$.
Because $X_n$ is converges to zero in probability, there exists finite $n_1$, such that $\Probability(|X_n| > \delta) < \epsilon/2$ for all $n \geq n_1$. Now, the event 
$\{|\wbar{X}_n|>\delta\}$ belongs to the union of $\{|\wbar{X}_n|>|X_n|\}$ and $\{|X_n|>\delta\}$.
As a result, we obtain
\[
\Probability(|\wbar{X}_n|>\delta)\leq \Probability (|\wbar{X}_n|>|X_n|) + \Probability(|X_n|>\delta)\leq \delta_n + \Probability(|X_n|>\delta)<\epsilon, 
\]
for $n\geq \wbar{n}=\max\{n_0,n_1\}$. Therefore, $\wbar{X}_n = o_p(1)$.
\end{proof}

\begin{lemma}
\label{lemma_bigOMarkov}
Let $X_n$ and $\wbar{X}_n$  be sequences of random variables. Suppose $\Expectation\bigbrackets{|X_n| \big|\wbar{X}_n} = o_p(1)$. Then, $X_n = o_p(1)$.
\end{lemma}
\begin{proof}
Fix any $\delta > 0$. 
Markov's inequality implies
\begin{align}   \Probability\Bigparenth{\normalabs{X_n} \geq \delta \Big|\wbar{X}_n} \leq \frac{1}{\delta}  \Expectation\Bigbrackets{|X_n|  \Big| \wbar{X}_n} = o_p(1).\label{eq_conv_2}
\end{align}
The law of total probability and the fact that conditional probabilities are bounded yield
\[
\Probability\Bigparenth{\normalabs{X_n} \geq \delta}=\Expectation\Bigbrackets{\Probability\Bigparenth{\normalabs{X_n}\geq \delta \Big|\wbar{X}_n}}\longrightarrow 0.
\]
\end{proof}

\begin{lemma}
\label{lemma_conv_coro_support}
Let $X_n$ and $\wbar{X}_n$  be sequences of random variables. Suppose $X_n=O_p(1)$ and $\Probability\bigparenth{|\wbar{X}_n|\geq |X_n| + f(\epsilon)} < \epsilon$ for some positive function $f$ and every $\epsilon \in \normalparenth{0,1}$. Then, $\wbar{X}_n = O_p(1)$.
\end{lemma}
\begin{proof}
We need to show that for any $\epsilon>0$, there exist finite $\wbar \delta>0$ and $\wbar n>0$, such that 
\[
\Probability(|\wbar{X}_n|>\wbar \delta)<\epsilon
\]
for all $n\geq \wbar n$. Fix any $\epsilon > 0$. Because $X_n$ is bounded in probability, there exist finite $\delta$ and $n_0$, such that $\Probability(|X_n| > \delta) < \epsilon/2$ for all $n \geq n_0$. 
Further, we have $\Probability\bigparenth{|\wbar{X}_n|\geq |X_n| + f(\epsilon/2)} < \epsilon/2$. Now, the event $\{|\wbar{X}_n| > \delta + f(\epsilon/2)\}$ belongs to the union of $\{|\wbar{X}_n| > |X_n| + f(\epsilon/2)\}$ and $\{|X_n| > \delta\}$.
As a result, we obtain
\begin{align}
    \Probability\bigparenth{|\wbar{X}_n| > \delta + f(\epsilon/2)} \leq \Probability\bigparenth{|\wbar{X}_n| > |X_n| + f(\epsilon/2)} + \Probability\bigparenth{|X_n| > \delta} < \epsilon.
\end{align}
    for all $n\geq n_0$. In other words, $\Probability(|\wbar{X}_n| > \wbar \delta) < \epsilon$ for all $n \geq \wbar n$, where $\wbar \delta = \delta + f(\epsilon/2)>0$ and $\wbar n = n_0$. Therefore, $\wbar{X}_n = O_p(1)$.
\end{proof}

\section{Proof of \texorpdfstring{Theorem \ref{thm_fsg}}{}: \fsg}
\label{sec_proof_thm_fsg}
Fix any $j \in [M]$. Recall the definitions of the parameter $\ATETrue$ and corresponding doubly-robust estimate $\ATEDR$ from \cref{eq_ate_parameter_combined,eq_counterfactual_mean_dr}, respectively. The error $\Delta \ATETrue^{\mathrm{DR}} = \ATEDR - \ATETrue$ 
can be re-expressed as
\begin{align}
\Delta \ATETrue^{\mathrm{DR}} & = \frac{1}{N}\sum_{i \in [N]} \Bigparenth{\what{\theta}_{i,j}^{(1,\mathrm{DR})} - \what{\theta}_{i,j}^{(0,\mathrm{DR})}} - \frac{1}{N}\sum_{i \in [N]} \Bigparenth{\theta_{i,j}^{(1)} - \theta_{i,j}^{(0)}}\\
& =  \frac{1}{N} \sum_{i \in [N]} \biggparenth{ \bigparenth{\what{\theta}_{i,j}^{(1,\mathrm{DR})} - \theta_{i,j}^{(1)} } - \bigparenth{\what{\theta}_{i,j}^{(0,\mathrm{DR})} - \theta_{i,j}^{(0)} }}\\
&\sequal{(a)} \frac{1}{N}  \sum_{i \in [N]} \bigparenth{\termDR[1] + \termDR[0]},
\label{eq_main_dr}
\end{align}
where $(a)$ follows after defining $\termDR[1] \defn \bigparenth{\what{\theta}_{i,j}^{(1,\mathrm{DR})} - \theta_{i,j}^{(1)}}$ and $\termDR[0] \defn - \bigparenth{\what{\theta}_{i,j}^{(0,\mathrm{DR})} - \theta_{i,j}^{(0)} }$ for every $(i,j) \in [N] \times [M]$. Then, we have
\begin{align}
    \termDR[1]  & = \what{\theta}_{i,j}^{(1,\mathrm{DR})} - \theta_{i,j}^{(1)}\\
    &\sequal{(a)}
    \what{\theta}_{i,j}^{(1)} + \bigparenth{y_{i,j} - \what{\theta}_{i,j}^{(1)}} \frac{a_{i,j}}{\what{p}_{i,j}}
    - \theta_{i,j}^{(1)} \\ 
    &\sequal{(b)} \what{\theta}_{i,j}^{(1)} + \bigparenth{\theta_{i, j}^{(1)} + \vareps_{i, j}^{(1)} - \what{\theta}_{i,j}^{(1)}} \frac{p_{i,j}+\eta_{i, j}}{\what{p}_{i,j}}
    - \theta_{i,j}^{(1)} \label{eq:int_error_dr_nondynamic} \\ 
    &= (\what{\theta}_{i,j}^{(1)} - \theta_{i,j}^{(1)}) \Bigparenth{1-\frac{p_{i,j}+\eta_{i, j}}{\what{p}_{i,j}}} +  \vareps_{i, j}^{(1)} \Bigparenth{\frac{p_{i,j}+\eta_{i, j}}{\what{p}_{i,j}}} \\ 
    &= \frac{ (\what{\theta}_{i,j}^{(1)} - \theta_{i,j}^{(1)}) (\what{p}_{i, j}-p_{i, j})}{\what{p}_{i,j}} - 
    \frac{ (\what{\theta}_{i,j}^{(1)} - \theta_{i,j}^{(1)}) \eta_{i, j}}{\what{p}_{i,j}}
    + \frac{\vareps_{i, j}^{(1)}p_{i,j}}{\what{p}_{i,j}} 
    + \frac{\vareps_{i, j}^{(1)}\eta_{i, j}}{\what{p}_{i,j}},
    \label{eq:mu1_err}
\end{align}
where $(a)$ follows from \cref{eq_counterfactual_mean_dr_1}, and $(b)$ follows from \cref{eq_consistency,eq_outcome_model,eq_treatment_model}. A similar derivation for $a = 0$ implies that  
\begin{align}
    \termDR[0] & = - \what{\theta}_{i,j}^{(0,\mathrm{DR})} + \theta_{i,j}^{(0)}\\
    &= - \frac{ (\what{\theta}_{i,j}^{(0)} - \theta_{i,j}^{(0)}) (1-\what{p}_{i, j}\!-\!(1-p_{i, j}))}{1-\what{p}_{i,j}} + 
    \frac{ (\what{\theta}_{i,j}^{(0)} - \theta_{i,j}^{(0)}) (-\eta_{i, j})}{1-\what{p}_{i,j}}
    - \frac{\vareps_{i, j}^{(0)}(1-p_{i,j})}{1-\what{p}_{i,j}} 
    - \frac{\vareps_{i, j}^{(0)}(-\eta_{i, j})}{1-\what{p}_{i,j}}  \\
    &= \frac{ (\what{\theta}_{i,j}^{(0)} - \theta_{i,j}^{(0)}) (\what{p}_{i, j}-p_{i, j})}{1-\what{p}_{i,j}} - 
    \frac{ (\what{\theta}_{i,j}^{(0)} - \theta_{i,j}^{(0)}) \eta_{i, j}}{1-\what{p}_{i,j}}
    - \frac{\vareps_{i, j}^{(0)}(1-p_{i,j})}{1-\what{p}_{i,j}} 
    + \frac{\vareps_{i, j}^{(0)}\eta_{i, j}}{1-\what{p}_{i,j}}.
    \label{eq:mu0_err}
\end{align}
Consider any $a \in \normalbraces{0,1}$ and any $\delta \in (0,1)$. We claim that, with probability at least $1 - 6\delta$, 
\begin{align}
    \frac{1}{N} \Bigabs{\sum_{i \in [N]} \termDR[a]} \leq &   \frac{2}{\lbar} \Rcol\bigparenth{\empo}   \RP + \frac{2\sqrt{c \ld}}{\lbar \sqrt{\lone N}}\Rcol\bigparenth{\empo} + \frac{2\sigmax \sqrt{c \ld}}{\lbar \sqrt{N}} + \frac{2 \sigmax \m(c \ld)}{\lbar \sqrt{\lone N}}, \label{eq_bound_11_dr}
\end{align}
where recall that $\m(c \ld) = \max\bigparenth{c\ld, \sqrt{c \ld}}$. We provide a proof of this claim at the end of this section. 
Applying triangle inequality in \cref{eq_main_dr} 
and using \cref{eq_bound_11_dr} with a union bound, we obtain that
\begin{align}
    \bigabs{\Delta \ATETrue^{\mathrm{DR}}} \leq \frac{2}{\lbar}  \RTheta  \RP + \frac{2 \sqrt{c \ld}}{\lbar \sqrt{\lone N}}\RTheta + \frac{4\sigmax \sqrt{c \ld}}{\lbar \sqrt{N}} + \frac{4\sigmax \m(c \ld)}{\lbar \sqrt{\lone N}}, \label{eq_fsg_final}
\end{align}
with probability at least $1-12\delta$.
The claim in \cref{eq_combined_error_dr} follows  by  re-parameterizing $\delta$. \medskip

\noindent {\bf Proof of bound~\texorpdfstring{\eqref{eq_bound_11_dr}}{}.}
Recall the partitioning of the units $[N]$ into $\cR_0$ and $\cR_1$ from \cref{assumption_estimates}. 
Now, to enable the application of concentration bounds, we split the summation over $i \in [N]$ in the left hand side of \cref{eq_bound_11_dr} into two parts---one over $i \in \cR_0$ and the other over $i \in \cR_1$---such that the noise terms are independent of the estimates of $\mpoz, \mpoo, \mta$ in each of these parts as in \cref{eq_independence_requirement_estimates_dr,eq_independence_requirement_estimates_dr_p}. 

Fix $a = 1$ and note that $|\sum_{i\in[N]}\termDR[1]| \leq |\sum_{i\in\cR_0} \termDR[1]|+|\sum_{i\in\cR_1} \termDR[1]|$. Fix any  $s \in \normalbraces{0,1}$. Then, \cref{eq:mu1_err} and triangle inequality imply
\begin{align}
    \Bigabs{\sum_{i\in\cR_s}\termDR[1]} \leq   & ~  \Bigabs{\sum_{i\in\cR_s}\frac{ \bigparenth{\what{\theta}_{i,j}^{(1)} \!-\! \theta_{i,j}^{(1)}}\bigparenth{\what{p}_{i,j} \!-\! p_{i,j}}}{\what{p}_{i,j}}} \!+\!\Bigabs{\sum_{i\in\cR_s}\frac{\bigparenth{\what{\theta}_{i,j}^{(1)} \!-\! \theta_{i,j}^{(1)}} \eta_{i,j}}{\what{p}_{i,j}}} \\
     & + \Bigabs{\sum_{i\in\cR_s}\frac{\varepsilon_{i,j}^{(1)} p_{i,j}}{\what{p}_{i,j}}}
     \!+\! \Bigabs{\sum_{i\in\cR_s}\frac{\varepsilon_{i,j}^{(1)} \eta_{i,j}}{\what{p}_{i,j}}}.
     \label{eq:decomp-T}
\end{align}
Applying the Cauchy-Schwarz inequality to bound the first term yields that
\begin{align}
    \biggabs{\sum_{i\in\cR_s}\frac{ \bigparenth{\what{\theta}_{i,j}^{(1)} \!-\! \theta_{i,j}^{(1)}}\bigparenth{\what{p}_{i,j} \!-\! p_{i,j}}}{\what{p}_{i,j}}}
    &\leq \sqrt{\sum_{i\in \cR_s} \biggparenth{\frac{\what{\theta}_{i,j}^{(1)} \!-\! \theta_{i,j}^{(1)}}{\what{p}_{i, j}}}^2
     \sum_{i\in \cR_s} \bigparenth{\what{p}_{i, j}-p_{i, j}}^2} \\
    &\leq \twonorm{ \bigparenth{\empoo_{\cdot,j} \!-\! \mpoo_{\cdot,j}} \odiv \emta_{\cdot,j}} \twonorm{\emta_{\cdot,j} \!-\! P_{\cdot,j}}.
    \label{eq:t1}
\end{align}

To bound the second term in \cref{eq:decomp-T}, note that $\eta_{i,j}$ is $ \subG[1/\sqrt{\lone}]$ (see Example 2.5.8 in \citet{vershynin2018high}) as well as zero-mean and independent across all $i \in [N]$ due to \cref{assumption_noise}\cref{item_ass_2aa}. By \cref{assumption_estimates}, $\normalbraces{\normalparenth{\what{p}_{i, j}, \what{\theta}_{i,j}^{(1)}}}_{i\in\cR_s} \indep \normalbraces{\eta_{i, j}}_{i\in\cR_s}$. The $\subGaussian$ concentration result in \cref{cor_subg_conc} yields
\begin{align}
\biggabs{\sum_{i\in\cR_s}\frac{\bigparenth{\what{\theta}_{i,j}^{(1)} \!-\! \theta_{i,j}^{(1)}} \eta_{i,j}}{\what{p}_{i,j}}}
    \leq \frac{\sqrt{c \ld}}{\sqrt{\lone}}  \sqrt{\sum_{i\in \cR_s} \biggparenth{\frac{\what{\theta}_{i,j}^{(1)} \!-\! \theta_{i,j}^{(1)}}{\what{p}_{i, j}}}^2}
    \leq \frac{\sqrt{c \ld}}{\sqrt{\lone}}   \twonorm{ \bigparenth{\empoo_{\cdot,j} \!-\! \mpoo_{\cdot,j}} \odiv \emta_{\cdot,j}}, 
    \label{eq:t2}
\end{align}
with probability at least $1-\delta$.

To bound the third term in \cref{eq:decomp-T}, note that $\varepsilon_{i,j}^{(1)}$ is $ \subG[\sigmax]$, zero-mean, and independent across all $i \in [N]$ due to \cref{assumption_noise}. By Assumption \cref{assumption_estimates}, $\normalbraces{\what{p}_{i, j}}_{i\in\cR_s} \indep \normalbraces{\vareps_{i, j}^{(1)}}_{i\in\cR_s}$. 
The $\subGaussian$ concentration result in \cref{cor_subg_conc} yields
\begin{align}
    \Bigabs{\sum_{i\in\cR_s}\frac{\varepsilon_{i,j}^{(1)} p_{i,j}}{\what{p}_{i,j}}}
    \leq \sigmax \sqrt{c \ld} \sqrt{\sum_{i\in \cR_s} \Bigparenth{\frac{p_{i, j}}{\what{p}_{i, j}}}^2}
    \leq \sigmax \sqrt{c \ld}  \twonorm{ \mta_{\cdot,j} \odiv \emta_{\cdot,j}},
    \label{eq:t3}
\end{align}
with probability at least $1-\delta$.

Finally, to bound the fourth term in \cref{eq:decomp-T},
note that $\varepsilon_{i,j}^{(1)}\eta_{i,j}$ is $\subE[\sigmax/\sqrt{\lone}]$ because of \cref{lem_prod_subG} as well as zero-mean and independent across all $i\in[N]$ due to \cref{assumption_noise}.
By \cref{assumption_estimates}, $\normalbraces{ \what{p}_{i, j} }_{i\in\cR_s} \indep \normalbraces{\normalparenth{\eta_{i, j}, \varepsilon^{(1)}_{i,j}}}_{i\in\cR_s}$. 
The $\subExponential$ concentration result in \cref{cor_sube_conc} yields that
\begin{align}
    \Bigabs{\sum_{i\in\cR_s}\frac{\varepsilon_{i,j}^{(1)} \eta_{i,j}}{\what{p}_{i,j}}}
    \leq \frac{\sigmax \m(c \ld)}{\sqrt{\lone}}  \stwonorm{ \ones \odiv \emta_{\cdot,j}},
    \label{eq:t4}
\end{align}
with probability at least $1-\delta$. 
Putting together \cref{eq:decomp-T,eq:t1,eq:t2,eq:t3,eq:t4}, we conclude that, with probability at least $1-3\delta$, 
\begin{align}
    \frac{1}{N} \Bigabs{\sum_{i \in \cR_s} \termDR[1]} &  
    \leq 
    \frac{1}{N} \twonorm{ \bigparenth{\empoo_{\cdot,j} \!-\! \mpoo_{\cdot,j}} \odiv \emta_{\cdot,j}} \twonorm{\emta_{\cdot,j} \!-\! P_{\cdot,j}} + \frac{\sqrt{c \ld}}{\sqrt{\lone} N} \twonorm{ \bigparenth{\empoo_{\cdot,j} \!-\! \mpoo_{\cdot,j}} \odiv \emta_{\cdot,j}}  \\
    & \qquad +  \frac{\sigmax \sqrt{c \ld}}{N}  \twonorm{ \mta_{\cdot,j} \odiv \emta_{\cdot,j}} + \frac{\sigmax \m(c \ld)}{\sqrt{\lone} N}  \twonorm{ \ones \odiv \emta_{\cdot,j}} \label{eq:decomp_T_first_bound}.
\end{align}
Then, noting that $1/\what{p}_{i,j} \leq 1/\lbar$ for every $i \in [N]$ and $j \in [M]$ from \cref{assumption_pos_estimated}, and consequently that $\stwonorm{B_{\cdot, j} \odiv \what P_{\cdot, j}} \leq  \onetwonorm{B}/\lbar$ for any matrix $B$ and every $j\in[M]$, we obtain the following bound, with probability at least $1-3\delta$, 
\begin{align}
    \frac{1}{N} \Bigabs{\sum_{i \in \cR_s} \termDR[1]} & \leq  
    \frac{1}{\lbar N} \onetwonorm{\empoo \!-\! \mpoo} \onetwonorm{\emta \!-\! \mta} + \frac{\sqrt{c \ld}}{ \lbar  \sqrt{\lone} N} \onetwonorm{\empoo \!-\! \mpoo} \\
    & \qquad + \frac{\sigmax \sqrt{c \ld}}{\lbar N}\onetwonorm{ \mta } + \frac{\sigmax \m(c \ld)}{\lbar \sqrt{\lone} N}\onetwonorm{ \allones }  \label{eq_dr_almost_final} \\
    & \sless{(a)} 
    \frac{1}{\lbar} \Rcol\bigparenth{\empo[(1)]}   \RP + \frac{\sqrt{c \ld}}{\lbar \sqrt{\lone N}}\Rcol\bigparenth{\empo[(1)]} + \frac{\sigmax \sqrt{c \ld}}{\lbar \sqrt{N}} + \frac{ \sigmax \m(c \ld)}{\lbar \sqrt{\lone N}},
    \label{eq_dr_final}
\end{align}
where $(a)$ follows from \cref{eq_notation} and because $\onetwonorm{ \mta } \leq \sqrt{N}$ and $\onetwonorm{ \allones } = \sqrt{N}$.
Then, the claim in \cref{eq_bound_11_dr} follows for $a=1$ by using \cref{eq_dr_final} and applying a union bound over $s \in \normalbraces{0,1}$.
The proof of \cref{eq_bound_11_dr} for $a = 0$ follows similarly.
\section{Proofs of \texorpdfstring{\cref{coro_compare,coro_consistency}}{}}
\label{proofs_coro}

\subsection{Proof of \texorpdfstring{\cref{coro_compare}}{}: \textbf{Gains of DR over OI and IPW}}
\label{subsec_proof_compare}
Fix any $j \in [M]$ and any $\delta \in (0,1)$. First, consider IPW. Take any $\alpha \in [0,1/2]$. From \cref{fsg_ipw_ate_t}, with probability at least $1-\delta$,
\begin{align}
    N^\alpha\bigabs{\ATEIPW - \ATETrue} &\leq   \frac{2\thetamax}{\lbar} N^\alpha\RP + f_1(\delta) N^{\alpha-1/2}\\
    &\leq \frac{2\thetamax}{\lbar} N^\alpha\RP + f_1(\delta)
\end{align}
where 
\begin{align}
    f_1(\delta) \defn \frac{2}{\lbar}  \biggparenth{\frac{\sqrt{c \ld[/12]}}{\sqrt{\lone}} \thetamax + 2\sigmax \sqrt{c \ld[/12]} + \frac{2\sigmax \m(c \ld[/12])}{\sqrt{\lone}}},
\end{align}
for $\m(c)$ and $\ell_c$ as defined in \cref{section_introduction}. Then, if $\RP = O_p\bigparenth{ N^{-\alpha}} $, \cref{lemma_conv_coro_support} implies
\begin{align}
    \bigabs{\ATEIPW - \ATETrue}  = O_p  \bigparenth{N^{-\alpha}}.
\end{align}
Next, consider DR. From \cref{fsg_dr_ate_t}, with probability at least $1-\delta$,
\begin{align}
    \bigabs{\ATEDR - \ATETrue} \leq \frac{2}{\lbar}  \RTheta  \RP +  f_2(\delta) N^{-1/2},
\end{align}
where 
\begin{align}
    f_2(\delta) \defn \frac{2}{\lbar}  \biggparenth{\frac{\sqrt{c \ld[/12]}}{\sqrt{\lone}} \RTheta + 2\sigmax \sqrt{c \ld[/12]} + \frac{2\sigmax \m(c \ld[/12])}{\sqrt{\lone}} }.
\end{align}
Suppose $\RP = O_p\bigparenth{ N^{-\alpha}} $ and $\RTheta = O_p\bigparenth{ N^{-\beta}}$. Consider two cases. First, suppose $\alpha + \beta \leq 0.5$. Then, with probability at least $1-\delta$,
\begin{align}
N^{\alpha + \beta}\bigabs{\ATEDR - \ATETrue} &\leq \frac{2}{\lbar}  N^{\alpha + \beta}\RTheta  \RP +  f_2(\delta) N^{\alpha+\beta-1/2}\\
&\leq \frac{2}{\lbar}  N^{\alpha + \beta}\RTheta  \RP +  f_2(\delta).
\end{align}
\cref{lemma_conv_coro_support} implies
$\bigabs{\ATEDR - \ATETrue} = O_p(N^{-(\alpha + \beta)})$. Next, suppose $\alpha + \beta > 0.5$. With probability at least $1-\delta$,
\begin{align}
N^{1/2}\bigabs{\ATEDR - \ATETrue} &\leq \frac{2}{\lbar}  N^{1/2}\RTheta  \RP +  f_2(\delta)\\
&\leq \frac{2}{\lbar}  N^{\alpha + \beta}\RTheta  \RP +  f_2(\delta).
\end{align}
\cref{lemma_conv_coro_support} implies
$\bigabs{\ATEDR - \ATETrue} = O_p(N^{-1/2})$.

\subsection{Proof of \texorpdfstring{\cref{coro_consistency}}{}: \textbf{Consistency for DR}}
\label{subsec_proof_consistency}
Fix any $j \in [M]$. Then, choose $\delta = 1/N$ in \cref{eq_combined_error_dr} and 
note that every term in the right hand side of \cref{eq_combined_error_dr} is $o_p(1)$ under the conditions on $\RTheta$ and $\RP$. Then, \cref{1con_dr_ate_t} follows from \cref{lemma_conv_prob}.

\section{Proof of \texorpdfstring{Theorem \ref{thm_normality}}{}: \normality}
\label{sec_proof_thm_normality}
For every $(i,j) \in [N] \times [M]$, recall the definitions of $\termDR[1]$ and $\termDR[0]$ from \cref{eq:mu1_err} and \cref{eq:mu0_err}, respectively. Then, define
\begin{align}
    \biasterm^{(1, \mathrm{DR})}_{i,j} & \defn \termDR[1] - \varepsilon_{i,j}^{(1)} - \frac{\varepsilon_{i,j}^{(1)} \eta_{i,j}}{p_{i,j}} \label{eq_bias1_ij_term}\\    
    \biasterm^{(0, \mathrm{DR})}_{i,j} & \defn \termDR[0]  + \varepsilon_{i,j}^{(0)} - \frac{\varepsilon_{i,j}^{(0)} \eta_{i,j}}{1-p_{i,j}}, \label{eq_bias0_ij_term}
    \shortintertext{and}
    \variance_{i,j} & \defn \varepsilon_{i,j}^{(1)} + \frac{\varepsilon_{i,j}^{(1)} \eta_{i,j}}{p_{i,j}} - \varepsilon_{i,j}^{(0)} + \frac{\varepsilon_{i,j}^{(0)} \eta_{i,j}}{1-p_{i,j}}. \label{eq_variance_ij_term}
\end{align}
Then, $\Delta \ATETrue^{\mathrm{DR}}$ in \cref{eq_main_dr} can be expressed as
\begin{align}
    \Delta \ATETrue^{\mathrm{DR}} & = \frac{1}{N}  \sum_{i \in [N]} \Bigparenth{\biasterm^{(1, \mathrm{DR})}_{i,j} + \biasterm^{(0, \mathrm{DR})}_{i,j} + \variance_{i,j}}. \label{eq_main_biasvariance}
\end{align}
We obtain the following convergence results.
\begin{lemma}[Convergence of $\bias_j$]
\label{lemma_bias_dr}
    Fix any $j \in [M]$. Suppose  \cref{assumption_pos,assumption_noise,assumption_estimates,assumption_pos_estimated} and conditions \cref{item_error_for_normality,item_product_of_error_for_normality,item_sigma_bounded_below} in \cref{thm_normality} hold. Then, 
    \begin{align}
        \frac{1}{\wbar{\sigma}_j \sqrt{N}}  \sum_{i \in [N]} \Bigparenth{\biasterm^{(1, \mathrm{DR})}_{i,j} + \biasterm^{(0, \mathrm{DR})}_{i,j}} = o_p(1).
    \end{align}
\end{lemma}
\begin{lemma}[Convergence of $\variance_j$]
\label{lemma_variance_dr}
    Fix any $j \in [M]$. Suppose  \cref{assumption_pos,assumption_noise} hold and condition \cref{item_sigma_bounded_below} in \cref{thm_normality} hold. Then, 
    \begin{align}
        \frac{1}{\wbar{\sigma}_j \sqrt{N}}  \sum_{i \in [N]} \variance_{i,j} \stackrel{d}{\longrightarrow} \mathcal{N}(0,1).
    \end{align}
\end{lemma}
Now, the result in \cref{thm_normality} follows from Slutsky's theorem.

\subsection{Proof of \texorpdfstring{\cref{lemma_bias_dr}}{}}
\label{subsec_bias}
Fix any $j \in [M]$. Consider any $a \in \normalbraces{0,1}$. We claim that
\begin{align}
    \frac{1}{\sqrt{N}} \sum_{i\in [N]} \biasterm^{(a, \mathrm{DR})}_{i,j} \leq O\Bigparenth{ \sqrt{N} \Rcol\bigparenth{\empo} \RP} + o_p(1).
    \label{eq_bias_expression_fsg_dr}
\end{align}
    We provide a proof of this claim at the end of this section.
    Then, using \cref{eq_bias_expression_fsg_dr} and the fact that $\wbar{\sigma}_j \geq c > 0$ as per condition \cref{item_sigma_bounded_below}, we obtain the following,
    \begin{align}
         \frac{1}{\wbar{\sigma}_j \sqrt{N}}  \sum_{i \in [N]} \!\!\Bigparenth{\biasterm^{(1, \mathrm{DR})}_{i,j} \!+\! \biasterm^{(0, \mathrm{DR})}_{i,j}} & \!\leq \frac{1}{c} \biggparenth{O\Bigparenth{ \sqrt{N} \RTheta \RP} + o_p(1)}\\
         & \sequal{(a)} \frac{1}{c} \biggparenth{ \sqrt{N} o_p(N^{-1/2}) + o_p(1)} \sequal{(b)} o_p(1),
         \label{eq_without_conditions}
    \end{align}
    where $(a)$ follows from \cref{item_product_of_error_for_normality}, and $(b)$ follows because $o_p(1) + o_p(1) = o_p(1)$. 
    \medskip

\noindent {\bf Proof of \texorpdfstring{\cref{eq_bias_expression_fsg_dr}}{}}
Recall the partitioning of the units $[N]$ into $\cR_0$ and $\cR_1$ from \cref{assumption_estimates}. Now, to enable the application of concentration bounds, we split the summation over $i \in [N]$ in the left hand side of \cref{eq_bias_expression_fsg_dr} into two parts---one over $i \in \cR_0$ and the other over $i \in \cR_1$---such that the noise terms are independent of the estimates of $\mpoz, \mpoo, \mta$ in each of these parts as in \cref{eq_independence_requirement_estimates_dr,eq_independence_requirement_estimates_dr_p}. 

Fix $a = 1$. Then, \cref{eq:mu1_err,eq_bias1_ij_term} imply that 
\begin{align}
    \biasterm^{(1, \mathrm{DR})}_{i,j}
    &= \frac{ \bigparenth{\what{\theta}_{i,j}^{(1)} \!-\! \theta_{i,j}^{(1)}}\bigparenth{\what{p}_{i,j} \!-\! p_{i,j}}}{\what{p}_{i,j}} \!-\! \frac{\bigparenth{\what{\theta}_{i,j}^{(1)} \!-\! \theta_{i,j}^{(1)}} \eta_{i,j}}{\what{p}_{i,j}} 
    + \frac{\vareps_{i, j}^{(1)}p_{i,j}}{\what{p}_{i,j}} 
    + \frac{\vareps_{i, j}^{(1)}\eta_{i, j}}{\what{p}_{i,j}} 
    -\varepsilon_{i,j}^{(1)} - \frac{\varepsilon_{i,j}^{(1)} \eta_{i,j}}{p_{i,j}}
    \\
    &= \frac{ \bigparenth{\what{\theta}_{i,j}^{(1)} \!-\! \theta_{i,j}^{(1)}}\bigparenth{\what{p}_{i,j} \!-\! p_{i,j}}}{\what{p}_{i,j}} \!-\! \frac{\bigparenth{\what{\theta}_{i,j}^{(1)} \!-\! \theta_{i,j}^{(1)}} \eta_{i,j}}{\what{p}_{i,j}} \!-\! \frac{\varepsilon_{i,j}^{(1)} \bigparenth{\what{p}_{i,j} \!-\! p_{i,j}}}{\what{p}_{i,j}}  \!-\! \frac{\varepsilon_{i,j}^{(1)} \eta_{i,j} \bigparenth{\what{p}_{i,j} \!-\! p_{i,j}}}{\what{p}_{i,j} p_{i,j}}.  \label{eq_bias_1}
\end{align}
Now, note that $| \sum_{i\in[N]}\biasterm^{(1, \mathrm{DR})}_{i,j} | \leq |\sum_{i\in\cR_0} \biasterm^{(1, \mathrm{DR})}_{i,j}|+|\sum_{i\in\cR_1} \biasterm^{(1, \mathrm{DR})}_{i,j}|$. Fix any $s \in \normalbraces{0,1}$. Then, triangle inequality implies that
\begin{align}
  \frac{1}{\sqrt{N}} \Bigabs{\sum_{i\in\cR_s} \biasterm^{(1, \mathrm{DR})}_{i,j}} &  \leq \frac{1}{\sqrt{N}} \Bigabs{\sum_{i\in\cR_s}\frac{ \bigparenth{\what{\theta}_{i,j}^{(1)} \!-\! \theta_{i,j}^{(1)}}\bigparenth{\what{p}_{i,j} \!-\! p_{i,j}}}{\what{p}_{i,j}}} + \frac{1}{\sqrt{N}} \Bigabs{\sum_{i\in\cR_s} \frac{\bigparenth{\what{\theta}_{i,j}^{(1)} \!-\! \theta_{i,j}^{(1)}} \eta_{i,j}}{\what{p}_{i,j}}} \\
    &+ \frac{1}{\sqrt{N}} \Bigabs{\sum_{i\in\cR_s}\frac{\varepsilon_{i,j}^{(1)} \bigparenth{\what{p}_{i,j} \!-\! p_{i,j}}}{\what{p}_{i,j}}}  + \frac{1}{\sqrt{N}} \Bigabs{\sum_{i\in\cR_s}\frac{\varepsilon_{i,j}^{(1)} \eta_{i,j} \bigparenth{\what{p}_{i,j} \!-\! p_{i,j}}}{\what{p}_{i,j} p_{i,j}}}.
    \label{eq_bias_1_triangled_bias}
\end{align}
To control the first term in \cref{eq_bias_1_triangled_bias}, we use 
the Cauchy-Schwarz inequality and \cref{assumption_pos_estimated} as in \cref{sec_proof_thm_fsg} (see \cref{eq:t1,eq_dr_almost_final,eq_dr_final}). 

To control the second term in \cref{eq_bias_1_triangled_bias}, we condition on $\normalbraces{\normalparenth{\what{p}_{i, j}, \what{\theta}_{i,j}^{(1)}}}_{i\in\cR_s}$.
Then, \cref{assumption_estimates} (i.e., \cref{eq_independence_requirement_estimates_dr}) provides that $\normalbraces{\normalparenth{\what{p}_{i, j}, \what{\theta}_{i,j}^{(1)}}}_{i\in\cR_s} \indep \normalbraces{\eta_{i, j}}_{i\in\cR_s}$. As a result, $\sum_{i\in\cR_s} \!\! \bigparenth{\what{\theta}_{i,j}^{(1)} - \theta_{i,j}^{(1)}} \eta_{i,j} / \what{p}_{i,j}$ is $\subGaussian\bigparenth{\bigbrackets{\sum_{i\in\cR_s} \bigparenth{\what{\theta}_{i,j}^{(1)} \!-\! \theta_{i,j}^{(1)}}^2 / \bigparenth{\what{p}_{i,j}}^2}^{1/2}/\sqrt{\lone}}$ because $\eta_{i,j}$ is $ \subG[1/\sqrt{\lone}]$ (see Example 2.5.8 in \citet{vershynin2018high}) as well as zero-mean and independent across all $i \in [N]$ due to \cref{assumption_noise}\cref{item_ass_2aa}.
Then, we have 
\begin{align}
\frac{1}{\sqrt{N}}  \Expectation\biggbrackets{\biggabs{\! \sum_{i\in\cR_s} \!\! \frac{\bigparenth{\what{\theta}_{i,j}^{(1)} \!-\! \theta_{i,j}^{(1)}} \eta_{i,j}}{\what{p}_{i,j}}} \Big| \normalbraces{\normalparenth{\what{p}_{i, j}, \what{\theta}_{i,j}^{(1)}}}_{i\in\cR_s}} & 
\sless{(a)} \frac{c}{\sqrt{N}}  \sqrt{\sum_{i\in\cR_s} \biggparenth{\frac{\what{\theta}_{i,j}^{(1)} \!-\! \theta_{i,j}^{(1)} }{\what{p}_{i,j}}}^2} \\
& \leq \frac{c}{\sqrt{N}}  \twonorm{ \bigparenth{\empoo_{\cdot,j} \!-\! \mpoo_{\cdot,j}} \odiv \emta_{\cdot,j}} \\
& \sless{(b)}  \frac{c}{\lbar} \Rcol\bigparenth{\empoo} \leq \frac{c}{\lbar} \RTheta \sequal{(c)}  o_p(1),
\label{eq_bias_1_triangled_2_bias}
\end{align}
where $(a)$ follows as the first moment of $\subG[\sigma]$ is $O(\sigma)$, $(b)$ follows from \cref{assumption_pos_estimated} and \cref{eq_notation}, and $(c)$ follows from \cref{item_error_for_normality}.

To control the third term in \cref{eq_bias_1_triangled_bias}, we condition on $\normalbraces{\what{p}_{i, j}}_{i\in\cR_s}$. 
Then, \cref{assumption_estimates} (i.e., \cref{eq_independence_requirement_estimates_dr_p}) provides that $\normalbraces{\what{p}_{i, j}}_{i\in\cR_s} \indep \normalbraces{\vareps_{i, j}^{(1)}}_{i\in\cR_s}$. As a result, $\sum_{i\in\cR_s} \!\! \varepsilon_{i,j}^{(1)} \bigparenth{\what{p}_{i,j} \!-\! p_{i,j}}/ \what{p}_{i,j}$ is $\subGaussian\bigparenth{\sigmax\bigbrackets{\sum_{i\in\cR_s} \bigparenth{\what{p}_{i,j} \!-\! p_{i,j}}^2 / \bigparenth{\what{p}_{i,j}}^2}^{1/2}}$ because $\varepsilon_{i,j}^{(1)}$ is $ \subG[\sigmax]$, zero-mean, and independent across all $i \in [N]$ due to \cref{assumption_noise}. Then, we have
\begin{align}
\frac{1}{\sqrt{N}} \Expectation\biggbrackets{ \biggabs{  \sum_{i\in\cR_s}\frac{\varepsilon_{i,j}^{(1)} \bigparenth{\what{p}_{i,j} \!-\! p_{i,j}}}{\what{p}_{i,j}}} \Big| \normalbraces{\what{p}_{i, j}}_{i\in\cR_s}} 
& \sless{(a)} \frac{c \sigmax}{\sqrt{N}} 
\sqrt{\sum_{i\in\cR_s}
 \biggparenth{\frac{\what{p}_{i,j} \!-\! p_{i,j}}{\what{p}_{i,j}}}^{2}}\\
& \leq \frac{c \sigmax}{\sqrt{N}}   \twonorm{ \bigparenth{\emta_{\cdot,j} \!-\! \mta_{\cdot,j}} \odiv \emta_{\cdot,j}} \\
& \sless{(b)}  \frac{c\sigmax}{\lbar} \RP \sequal{(c)}  o_p(1),
\label{eq_bias_1_triangled_3_bias}
\end{align}
where $(a)$ follows as the first moment of $\subG[\sigma]$ is $O(\sigma)$, $(b)$ follows from \cref{assumption_pos_estimated} and \cref{eq_notation}, and $(c)$ follows from \cref{item_error_for_normality}.
To control the fourth term in \cref{eq_bias_1_triangled_bias}, we condition on $\normalbraces{\what{p}_{i, j}}_{i\in\cR_s}$.
Then, \cref{assumption_estimates} (i.e., \cref{eq_independence_requirement_estimates_dr_p}) provides that $\normalbraces{ \what{p}_{i, j} }_{i\in\cR_s} \indep \normalbraces{\normalparenth{\eta_{i, j}, \varepsilon^{(1)}_{i,j}}}_{i\in\cR_s}$. As a result, $\sum_{i\in\cR_s}  \varepsilon_{i,j}^{(1)} \eta_{i,j} \bigparenth{\what{p}_{i,j} - p_{i,j}} / \what{p}_{i,j} p_{i,j}$ is $\subExponential\bigparenth{\sigmax\bigbrackets{\sum_{i\in\cR_s} \bigparenth{\what{p}_{i,j} - p_{i,j}}^2 / \bigparenth{\what{p}_{i,j} p_{i,j}}^2}^{1/2}/\sqrt{\lone}}$ because  $\varepsilon_{i,j}^{(1)}\eta_{i,j}$ is $\subE[\sigmax/\sqrt{\lone}]$ due to \cref{lem_prod_subG} as well as zero-mean and independent across all $i\in[N]$ due to \cref{assumption_noise}.
Then, we have
\begin{align}
\frac{1}{\sqrt{N}}  \Expectation\biggbrackets{ \biggabs{ \sum_{i\in\cR_s}\frac{\varepsilon_{i,j}^{(1)} \eta_{i,j} \bigparenth{\what{p}_{i,j} \!-\! p_{i,j}}}{\what{p}_{i,j}p_{i,j}}} \Big| \normalbraces{\what{p}_{i, j}}_{i\in\cR_s}} 
& \sless{(a)} \frac{c \sigmax}{\sqrt{N}} \sqrt{ \sum_{i\in\cR_s} \biggparenth{\frac{\what{p}_{i,j} \!-\! p_{i,j}}{\what{p}_{i,j}p_{i,j}}}^{2}}
\\
& \leq
\frac{c \sigmax}{\sqrt{N}} \twonorm{ \bigparenth{\emta_{\cdot,j} \!-\! \mta_{\cdot,j}} \odiv \bigparenth{\emta_{\cdot,j} \odot \mta_{\cdot,j}}} \\
 & \sless{(b)}  \frac{c \sigmax}{\lbar \lambda} \RP \sequal{(c)}  o_p(1),
\label{eq_bias_1_triangled_4_bias}
\end{align}
where $(a)$ follows as the first moment of $\subE[\sigma]$ is $O(\sigma)$, $(b)$ follows from \cref{assumption_pos_estimated} and \cref{eq_notation}, and $(c)$ follows from \cref{item_error_for_normality}.

Putting together \cref{eq_bias_1_triangled_bias,eq_bias_1_triangled_2_bias,eq_bias_1_triangled_3_bias,eq_bias_1_triangled_4_bias} using \cref{lemma_bigOMarkov}, we have
\begin{align}
    \frac{1}{\sqrt{N}} \Bigabs{\sum_{i\in\cR_s} \biasterm^{(1, \mathrm{DR})}_{i,j}} \leq O\Bigparenth{ \sqrt{N} \Rcol\bigparenth{\empoo} \RP} + o_p(1).
\end{align}
Then, the claim in \cref{eq_bias_expression_fsg_dr} follows for $a=1$ by using $| \sum_{i\in[N]}\biasterm^{(1, \mathrm{DR})}_{i,j} |  \leq |\sum_{i\in\cR_0} \biasterm^{(1, \mathrm{DR})}_{i,j}|+|\sum_{i\in\cR_1} \biasterm^{(1, \mathrm{DR})}_{i,j}|$. The proof of \cref{eq_bias_expression_fsg_dr} for $a = 0$ follows similarly.

\subsection{Proof of \texorpdfstring{\cref{lemma_variance_dr}}{}} 
\label{subsec_variance}
To prove this result, we invoke Lyapunov central limit theorem (CLT).
\begin{lemma}[Lyapunov CLT, see Theorem 27.3 of \citet{billingsley2017probability}]
\label{lemma_clt}
    Consider a sequence $x_1, x_2, \cdots $ of mean-zero independent random variables such that the moments
    $\Expectation[|x_i|^{2+\omega}]$ are finite for some $\omega>0$.
    Moreover,  assume that the \textit{Lyapunov’s condition} is satisfied, i.e., 
    \begin{align}
\displaystyle\sum_{i=1}^N \Expectation[{|{x_i}|^{2+\omega}}]\Big/\Big(\displaystyle\sum_{i=1}^N \E[x_i^2]\Big)^{\frac{2+\omega}{2}} \longrightarrow 0, \label{eq:lya_cond}
        \shortintertext{as $N \to \infty$. Then,}
\displaystyle\sum_{i=1}^N x_i\Big/\Big(\displaystyle\sum_{i=1}^N \E[x_i^2]\Big)^{\frac12} \stackrel{d}{\longrightarrow} \mathcal{N}(0,1),
    \end{align}
as $N \to \infty$.
\end{lemma}
Fix any $j \in [M]$. We apply Lyapunov CLT in \cref{lemma_clt} on the sequence $\variance_{1,j}, \variance_{2,j}, \cdots$ where $\variance_{i,j}$ is as defined in \cref{eq_variance_ij_term}. Note that this sequence is zero-mean from \cref{assumption_noise}\cref{item_ass_2aa} and \cref{assumption_noise}\cref{item_ass_2bb}, and independent from \cref{assumption_noise}\cref{item_ass_2bb}. 
First, we show in \cref{subsubsec_proof_1} that
\begin{align}
    \Variance(\variance_{i,j}) = \frac{(\sigma_{i,j}^{(1)})^2}{p_{i,j}}+\frac{(\sigma_{i,j}^{(0)})^2}{1-p_{i,j}}, \label{eq_variance_claim}
\end{align}
for each $i\in [N]$. Next, we show in \cref{subsubsec_proof_2} that Lyapunov’s condition~\eqref{eq:lya_cond} holds for the sequence $\variance_{1,j}, \variance_{2,j}, \cdots$ with $\omega = 1$.
Finally, applying \cref{lemma_clt} and using the definition of $\wbar{\sigma}_j$ from \cref{eq_normality_variance_thm} yields \cref{lemma_variance_dr}.

\subsubsection{Proof of \texorpdfstring{\cref{eq_variance_claim}}{}} 
\label{subsubsec_proof_1} 
Fix any $i \in [N]$ and consider $\Variance(\variance_{i,j})$. We have
\begin{align}
    \Variance\Bigparenth{\variance_{i,j}} = & \Variance\biggparenth{\varepsilon_{i,j}^{(1)}\Bigparenth{1 + \frac{ \eta_{i,j}}{p_{i,j}}} - \varepsilon_{i,j}^{(0)}\Bigparenth{1 - \frac{ \eta_{i,j}}{1-p_{i,j}}}}. \label{eq_var_0}
\end{align}
We claim the following:
\begin{gather}
 \Variance\biggparenth{\varepsilon_{i,j}^{(1)}\Bigparenth{1 + \frac{ \eta_{i,j}}{p_{i,j}}}}  = \frac{(\sigma_{i,j}^{(1)})^2}{p_{i,j}}, \label{eq_var_1}\\
 \Variance\biggparenth{\varepsilon_{i,j}^{(0)}\Bigparenth{1 - \frac{ \eta_{i,j}}{1-p_{i,j}}}} = \frac{(\sigma_{i,j}^{(0)})^2}{1-p_{i,j}}, \label{eq_var_2}\\
\shortintertext{and}
 \Covariance\biggparenth{\varepsilon_{i,j}^{(1)}\Bigparenth{1 + \frac{ \eta_{i,j}}{p_{i,j}}}, \varepsilon_{i,j}^{(0)}\Bigparenth{1 - \frac{ \eta_{i,j}}{1-p_{i,j}}}} = 0,  \label{eq_var_3}
\end{gather}
with \cref{eq_variance_claim} following from  \cref{eq_var_0,eq_var_1,eq_var_2,eq_var_3}. 

To establish \cref{eq_var_1}, notice that  \cref{assumption_noise}\cref{item_ass_2aa,item_ass_2bb} imply $\varepsilon_{i,j}^{(1)}\indep \eta_{i,j}$ and $\Expectation[\varepsilon_{i,j}^{(1)}] = \Expectation[\eta_{i,j}] = 0 $ so that $\Expectation[\varepsilon_{i,j}^{(1)}(1 + \eta_{i,j}/p_{i,j})] = 0$. Then,
\begin{align}
\Variance\biggparenth{\varepsilon_{i,j}^{(1)}\Bigparenth{1 + \frac{ \eta_{i,j}}{p_{i,j}}}} 
&=\Expectation\biggbrackets{\Bigparenth{\varepsilon_{i,j}^{(1)}\Bigparenth{1 + \frac{ \eta_{i,j}}{p_{i,j}}}}^2} 
= \Expectation\biggbrackets{\Bigparenth{\varepsilon_{i,j}^{(1)}}^2}\Expectation\biggbrackets{\Bigparenth{1 + \frac{ \eta_{i,j}}{p_{i,j}}}^2} \\
& 
= \Expectation\biggbrackets{\Bigparenth{\varepsilon_{i,j}^{(1)}}^2}\biggbrackets{1 + \Expectation\biggbrackets{\frac{ \eta_{i,j}^2}{p_{i,j}^2}}} 
\sequal{(a)} (\sigma_{i,j}^{(1)})^2 \brackets{1 + \frac{p_{i,j}(1-p_{i,j})}{p_{i,j}^2}}\\ &=\frac{(\sigma_{i,j}^{(1)})^2}{p_{i,j}},
\end{align}
where $(a)$ follows because $\Expectation\normalbrackets{\eta_{i,j}^2} = \Variance\normalparenth{\eta_{i,j}} =p_{i,j}(1-p_{i,j})$ from \cref{eq_treatment_model}, and $\Expectation\bigbrackets{\normalparenth{\varepsilon_{i,j}^{(1)}}^2} = \Variance\normalparenth{\varepsilon_{i,j}^{(1)}} = (\sigma_{i,j}^{(1)})^2$ from condition \cref{item_sigma_bounded_below}. A similar argument establishes \cref{eq_var_2}. \cref{eq_var_3} follows from,
\begin{align}
\Covariance\biggparenth{\varepsilon_{i,j}^{(1)}\Bigparenth{1 + \frac{ \eta_{i,j}}{p_{i,j}}}, \varepsilon_{i,j}^{(0)}\Bigparenth{1 - \frac{ \eta_{i,j}}{1-p_{i,j}}}} 
&= 
\Expectation\biggbrackets{\varepsilon_{i,j}^{(1)}\Bigparenth{1 + \frac{ \eta_{i,j}}{p_{i,j}}} \varepsilon_{i,j}^{(0)}\Bigparenth{1 - \frac{ \eta_{i,j}}{1-p_{i,j}}}} \\
&
\sequal{(a)}
\Expectation\biggbrackets{\Bigparenth{1 + \frac{ \eta_{i,j}}{p_{i,j}}} \Bigparenth{1 - \frac{ \eta_{i,j}}{1-p_{i,j}}}} \Expectation[{\varepsilon_{i,j}^{(1)}\varepsilon_{i,j}^{(0)}}] \\
&=
\biggparenth{1 - \Expectation\biggbrackets{\frac{\eta^2_{i,j}}{{p_{i,j} \bigparenth{1-p_{i,j}}}} }}\Expectation[{\varepsilon_{i,j}^{(1)}\varepsilon_{i,j}^{(0)}}] \\
&\sequal{(b)} 0 \cdot\Expectation[{\varepsilon_{i,j}^{(1)}\varepsilon_{i,j}^{(0)}}] = 0,
\end{align}
where $(a)$ follows because $(\varepsilon^{(0)}_{i, j}, \varepsilon^{(1)}_{i, j}) \indep \eta_{i, j}$ from \cref{assumption_noise}\cref{item_ass_2bb} and $(b)$ follows because $\Expectation\normalbrackets{\eta_{i,j}^2} = \Variance\normalparenth{\eta_{i,j}} =p_{i,j}(1-p_{i,j})$.
\subsubsection{Proof of Lyapunov's condition with \texorpdfstring{$\omega=1$}{}}
\label{subsubsec_proof_2}
We have
\begin{align}
    \frac{\sum_{i \in [N]} \Expectation\bigbrackets{\normalabs{\variance_{i,j}}^{3}}}{\bigparenth{\sum_{i \in [N]}  \Variance(\variance_{i,j})}^{3/2}} &= \frac{1}{N^{3/2}}  \frac{\sum_{i \in [N]} \Expectation\bigbrackets{\normalabs{\variance_{i,j}}^{3}}}{\bigparenth{\frac{1}{N}\sum_{i \in [N]}  \Variance(\variance_{i,j})}^{3/2}}
     \sequal{(a)} \frac{1}{N^{3/2}}  \frac{\sum_{i \in [N]} \Expectation\bigbrackets{\normalabs{\variance_{i,j}}^{3}}}{\bigparenth{\wbar{\sigma}_j }^{3/2}} 
\\
    & \sless{(b)} \frac{1}{N^{3/2}}  \frac{\sum_{i \in [N]} \Expectation\bigbrackets{\normalabs{\variance_{i,j}}^{3}}}{c_1^{3/2}}  \sless{(c)} \frac{1}{N^{1/2}} \frac{c_2}{c_1^{3/2}}, \label{eq_pre_limit}
\end{align}
where $(a)$ follows by putting together \cref{eq_normality_variance_thm,eq_variance_claim}, $(b)$ follows because $\wbar{\sigma}_j \geq c_1 > 0$ as per condition \cref{item_sigma_bounded_below}, $(c)$ follows because  the absolute third moments of $\subExponential$ random variables are bounded, after noting that $\variance_{i,j}$ is a $\subExponential$ random variable. Then, condition \eqref{eq:lya_cond} holds for $\omega=1$ as the right hand side of \cref{eq_pre_limit} goes to zero as $N\to \infty$.

\subsection{Proof of \texorpdfstring{\cref{prop_variance_estimate}: \textbf{Consistent variance estimation}}{}}
\label{subsec_proof_variance_estimate}
Fix any $j \in [M]$ and recall the definitions of $\wbar{\sigma}_j^2$ and $\what{\sigma}_j^2$ from \cref{eq_normality_variance_thm,eq_estimated_variance}, respectively. The error $\Delta_j = \what{\sigma}_j^2 - \wbar{\sigma}_j^2$ can be expressed as
\begin{align}
    \Delta_j & = \frac{1}{N}\sum_{i \in [N]} \biggparenth{\frac{ \bigparenth{\what{\theta}^{(1)}_{i,j} - y_{i,j}}^2 a_{i,j}}{\bigparenth{\what{p}_{i,j}}^2} + \frac{ \bigparenth{\what{\theta}^{(0)}_{i,j} - y_{i,j}}^2 \normalparenth{1-a_{i,j}}}{\bigparenth{1-\what{p}_{i,j}}^2}} - \biggparenth{ \frac{(\sigma_{i,j}^{(1)})^2}{p_{i,j}} + \frac{(\sigma_{i,j}^{(0)})^2}{1-p_{i,j}}} \\
    & = \frac{1}{N}\sum_{i \in [N]} \biggparenth{\frac{ \bigparenth{\what{\theta}^{(1)}_{i,j} - y_{i,j}}^2 a_{i,j}}{\bigparenth{\what{p}_{i,j}}^2} - \frac{(\sigma_{i,j}^{(1)})^2}{p_{i,j}}} + \biggparenth{\frac{ \bigparenth{\what{\theta}^{(0)}_{i,j} - y_{i,j}}^2 \normalparenth{1-a_{i,j}}}{\bigparenth{1-\what{p}_{i,j}}^2} - \frac{(\sigma_{i,j}^{(0)})^2}{1-p_{i,j}}} \\
    &\sequal{(a)} \frac{1}{N}  \sum_{i \in [N]} \Bigparenth{\termV[1] + \termV[0]}, \label{eq_var_dr}
\end{align}
where $(a)$ follows after defining
\begin{align}
    \termV[1] \defn \frac{ \bigparenth{\what{\theta}^{(1)}_{i,j} - y_{i,j}}^2 a_{i,j}}{\bigparenth{\what{p}_{i,j}}^2} - \frac{(\sigma_{i,j}^{(1)})^2}{p_{i,j}} \qtext{and} \termV[0]  \defn \frac{ \bigparenth{\what{\theta}^{(0)}_{i,j} - y_{i,j}}^2 \normalparenth{1-a_{i,j}}}{\bigparenth{1-\what{p}_{i,j}}^2} - \frac{(\sigma_{i,j}^{(0)})^2}{1-p_{i,j}}.
\end{align}
for every $(i,j) \in [N] \times [M]$. Then, we have
\begin{align}
    \termV[1] & \sequal{(a)} \frac{\bigparenth{\what{\theta}^{(1)}_{i,j} - \theta_{i,j}^{(1)} -  \varepsilon_{i,j}^{(1)}}^2 \bigparenth{p_{i,j} + \eta_{i,j}} }{\bigparenth{\what{p}_{i,j}}^2} - \frac{(\sigma_{i,j}^{(1)})^2}{p_{i,j}} \\
    & = 
    \frac{ \bigparenth{\what{\theta}_{i,j}^{(1)} \!-\! \theta_{i,j}^{(1)}}^2 a_{i,j}}{\bigparenth{\what{p}_{i,j}}^2} - 
    \!-\! \frac{2 \varepsilon_{i,j}^{(1)} p_{i,j} \bigparenth{\what{\theta}_{i,j}^{(1)} \!-\! \theta_{i,j}^{(1)}}}{\bigparenth{\what{p}_{i,j}}^2} \!-\! \frac{2 \varepsilon_{i,j}^{(1)} \eta_{i,j} \bigparenth{\what{\theta}_{i,j}^{(1)} \!-\! \theta_{i,j}^{(1)}}}{\bigparenth{\what{p}_{i,j}}^2}\\
    & \qquad \qquad   + \frac{ \bigparenth{\varepsilon_{i,j}^{(1)}}^2 p_{i,j}}{\bigparenth{\what{p}_{i,j}}^2} + \frac{ \bigparenth{\varepsilon_{i,j}^{(1)}}^2 \eta_{i,j}}{\bigparenth{\what{p}_{i,j}}^2} - \frac{(\sigma_{i,j}^{(1)})^2}{p_{i,j}},
\end{align}
where $(a)$ follows from \cref{eq_consistency,eq_outcome_model,eq_treatment_model}. A similar derivation for $a = 0$ implies that 
\begin{align}
    \termV[0] & = \frac{ \bigparenth{\what{\theta}^{(0)}_{i,j} - \theta_{i,j}^{(0)} -  \varepsilon_{i,j}^{(0)}}^2 \normalparenth{1-p_{i,j} - \eta_{i,j}}}{\bigparenth{1-\what{p}_{i,j}}^2} - \frac{(\sigma_{i,j}^{(0)})^2}{1-p_{i,j}} \\
    & = 
    \frac{ \bigparenth{\what{\theta}_{i,j}^{(0)} \!-\! \theta_{i,j}^{(0)}}^2 \bigparenth{1-a_{i,j}}}{\bigparenth{1-\what{p}_{i,j}}^2}
    \!-\! \frac{2 \varepsilon_{i,j}^{(0)} \bigparenth{1-p_{i,j}} \bigparenth{\what{\theta}_{i,j}^{(0)} \!-\! \theta_{i,j}^{(0)}}}{\bigparenth{1-\what{p}_{i,j}}^2}  \!+\! \frac{2 \varepsilon_{i,j}^{(0)} \eta_{i,j} \bigparenth{\what{\theta}_{i,j}^{(0)} \!-\! \theta_{i,j}^{(0)}}}{\bigparenth{1-\what{p}_{i,j}}^2}  \\
    & \qquad \qquad + \frac{ \bigparenth{\varepsilon_{i,j}^{(0)}}^2 \bigparenth{1-p_{i,j}}}{\bigparenth{1-\what{p}_{i,j}}^2} - \frac{ \bigparenth{\varepsilon_{i,j}^{(0)}}^2 \eta_{i,j}}{\bigparenth{1-\what{p}_{i,j}}^2} - \frac{(\sigma_{i,j}^{(0)})^2}{1-p_{i,j}}.
\end{align}
Consider any $a \in \normalbraces{0,1}$. We claim that 
\begin{align}
     \frac{1}{N} \Bigabs{\sum_{i \in [N]} \termV[a]} = o_p(1).
     \label{eq_bound_11_var}
\end{align}
We provide a proof of this claim at the end of this section. Then, applying triangle inequality in \cref{eq_var_dr}, we obtain the following
\begin{align}
    \Delta_j & \leq o_p(1) + o_p(1)  
     \sequal{(a)} o_p(1), 
\end{align}
where $(a)$ follows 
because 
$o_p(1) + o_p(1) = o_p(1)$. 

\paragraph{Proof of bound \cref{eq_bound_11_var}.}
This proof follows a very similar road map to that used for establishing the inequality in \cref{eq_bias_expression_fsg_dr}.
Recall the partitioning of the units $[N]$ into $\cR_0$ and $\cR_1$ from \cref{assumption_estimates}. Now, to enable the application of concentration bounds, we split the summation over $i \in [N]$ in the left hand side of \cref{eq_bound_11_var} into two parts---one over $i \in \cR_0$ and the other over $i \in \cR_1$---such that the noise terms are independent of the estimates of $\mpoz, \mpoo, \mta$ in each of these parts as in \cref{eq_independence_requirement_estimates_dr,eq_independence_requirement_estimates_dr_p}. 

Fix $a = 1$. Now, note that $|\sum_{i\in[N]}\termV[1]| \leq |\sum_{i\in\cR_0} \termV[1]|+|\sum_{i\in\cR_1} \termV[1]|$. Fix any $s \in \normalbraces{0,1}$. Then, triangle inequality implies that
\begin{align}
     \frac{1}{N} \Bigabs{\!\sum_{i \in \cR_s} \termV[1]} & \leq \frac{1}{N} \Bigabs{\!\sum_{i \in \cR_s} \frac{ \bigparenth{\what{\theta}_{i,j}^{(1)} \!-\! \theta_{i,j}^{(1)}}^2 a_{i,j}}{\bigparenth{\what{p}_{i,j}}^2}} 
     \!+\!  \frac{1}{N} \Bigabs{\!\sum_{i \in \cR_s}\frac{2 \varepsilon_{i,j}^{(1)} p_{i,j} \bigparenth{\what{\theta}_{i,j}^{(1)} \!-\! \theta_{i,j}^{(1)}}}{\bigparenth{\what{p}_{i,j}}^2}} \\
    &   \!+\!  \frac{1}{N} \Bigabs{\!\sum_{i \in \cR_s}\frac{2 \varepsilon_{i,j}^{(1)} \eta_{i,j} \bigparenth{\what{\theta}_{i,j}^{(1)} \!-\! \theta_{i,j}^{(1)}}}{\bigparenth{\what{p}_{i,j}}^2}} \!+\!  \frac{1}{N} \Bigabs{\!\sum_{i \in \cR_s}\frac{ \bigparenth{\varepsilon_{i,j}^{(1)}}^2 \eta_{i,j}}{\bigparenth{\what{p}_{i,j}}^2}} \!+\!  \frac{1}{N} \Bigabs{\!\sum_{i \in \cR_s} \frac{ \bigparenth{\varepsilon_{i,j}^{(1)}}^2 p_{i,j}}{\bigparenth{\what{p}_{i,j}}^2} \!-\! \frac{(\sigma_{i,j}^{(1)})^2}{p_{i,j}}}. \label{eq_decomp_variance_0}
\end{align}
To bound the first term in \cref{eq_decomp_variance_0}, we have
\begin{align}
    \frac{1}{N} \Bigabs{\!\sum_{i \in \cR_s}\!\! \frac{ \bigparenth{\what{\theta}_{i,j}^{(1)} \!-\! \theta_{i,j}^{(1)}}^2 a_{i,j}}{\bigparenth{\what{p}_{i,j}}^2}} \sless{(a)} \frac{1}{N} \Bigabs{\!\sum_{i \in \cR_s}\!\! \frac{ \bigparenth{\what{\theta}_{i,j}^{(1)} \!-\! \theta_{i,j}^{(1)}}^2}{\bigparenth{\what{p}_{i,j}}^2}} 
    & \sless{(b)} \frac{1}{\lbar^2 N} \twonorm{\empoo_{\cdot,j} \!-\! \mpoo_{\cdot,j}}^2 \\
    & \sequal{(c)} \frac{1}{\lbar^2}  \Bigbrackets{\Rcol\bigparenth{\empoo}}^2  \\
    & \leq \frac{1}{\lbar^2}  \Bigbrackets{\RTheta}^2  \sequal{(d)}  o_p(1) o_p(1) \sequal{(e)}  o_p(1), \label{eq_decomp_variance_1}
\end{align}
where $(a)$ follows as $a_{i,j} \in \normalbraces{0,1}$, $(b)$ follows from \cref{assumption_pos_estimated}, $(c)$ follows from \cref{eq_notation}, $(d)$ follows from \cref{item_error_for_normality}, and $(e)$ follows because $o_p(1) o_p(1) = o_p(1)$.

To control second term in \cref{eq_decomp_variance_0}, 
we condition on $\normalbraces{\bigparenth{\what{p}_{i, j}, \what{\theta}^{(1)}_{i,j}}}_{i\in\cR_s}$. Then, \cref{eq_independence_requirement_estimates_new} provides that $\normalbraces{\bigparenth{\what{p}_{i, j}, \what{\theta}^{(1)}_{i,j}}}_{i\in\cR_s} \indep \normalbraces{\vareps_{i, j}^{(1)}}_{i\in\cR_s}$. As a result, $\sum_{i\in\cR_s} \!\! \varepsilon_{i,j}^{(1)} p_{i,j}  \bigparenth{\what{\theta}_{i,j}^{(1)} - \theta_{i,j}^{(1)}} / \bigparenth{\what{p}_{i,j}}^2$ is $\subGaussian\bigparenth{\sigmax\bigbrackets{\sum_{i\in\cR_s} \bigparenth{p_{i,j}}^2  \bigparenth{\what{\theta}_{i,j}^{(1)} - \theta_{i,j}^{(1)}}^2 / \bigparenth{\what{p}_{i,j}}^4}^{1/2}}$ because $\varepsilon_{i,j}^{(1)}$ is $ \subG[\sigmax]$, zero-mean and independent across all $i \in [N]$ due to \cref{assumption_noise}. Then, we have 
\begin{align}
    \frac{1}{N}  \Expectation\biggbrackets{\biggabs{\! \sum_{i\in\cR_s} \! \frac{2 \varepsilon_{i,j}^{(1)} p_{i,j} \bigparenth{\what{\theta}_{i,j}^{(1)} \!-\! \theta_{i,j}^{(1)}}}{\bigparenth{\what{p}_{i,j}}^2}} \Big| \normalbraces{\normalparenth{\what{p}_{i, j}, \what{\theta}_{i,j}^{(1)}}}_{i\in\cR_s}} & \!\sless{(a)}\! \frac{c\sigmax}{N} \sqrt{\sum_{i\in \cR_s} \biggparenth{\frac{p_{i,j} \bigparenth{\what{\theta}_{i,j}^{(1)} \!-\! \theta_{i,j}^{(1)}}}{\bigparenth{\what{p}_{i,j}}^2}}^2} \\
    & \! \sless{(b)} \! \frac{c \sigmax}{\lbar^2 N} \twonorm{\empoo_{\cdot,j} \!-\! \mpoo_{\cdot,j}}   \\
   &  \sequal{(c)} \frac{c \sigmax}{\lbar^2} \frac{\Rcol\bigparenth{\empoo}}{\sqrt{N}} \leq \frac{c \sigmax}{\lbar^2} \frac{\RTheta}{\sqrt{N}} \sequal{(d)} o_p(1), \label{eq_decomp_variance_2}
\end{align}
where $(a)$ follows as the first moment of $\subG[\sigma]$ is $O(\sigma)$, $(b)$ follows from \cref{assumption_pos,assumption_pos_estimated}, $(c)$ follows from \cref{eq_notation}, and $(d)$ follows from \cref{item_error_for_normality}.

To control third term in \cref{eq_decomp_variance_0}, we condition on $\normalbraces{\bigparenth{\what{p}_{i, j}, \what{\theta}^{(1)}_{i,j}}}_{i\in\cR_s}$. Then, \cref{eq_independence_requirement_estimates_new} provides that $\normalbraces{\bigparenth{\what{p}_{i, j}, \what{\theta}^{(1)}_{i,j}}}_{i\in\cR_s} \indep \normalbraces{\normalparenth{\eta_{i, j}, \varepsilon^{(1)}_{i,j}}}_{i\in\cR_s}$. As a result, $\sum_{i\in\cR_s} \!\! \varepsilon_{i,j}^{(1)} \eta_{i,j}  \bigparenth{\what{\theta}_{i,j}^{(1)} - \theta_{i,j}^{(1)}} / \bigparenth{\what{p}_{i,j}}^2$ is $\subExponential\bigparenth{\sigmax\bigbrackets{\sum_{i\in\cR_s} \bigparenth{\what{\theta}_{i,j}^{(1)} - \theta_{i,j}^{(1)}}^2 / \bigparenth{\what{p}_{i,j}}^4}^{1/2} /\sqrt{\lone}}$ because  $\varepsilon_{i,j}^{(1)}\eta_{i,j}$ is $\subExponential(\sigmax/\sqrt{\lone})$ due to \cref{lem_prod_subG} as well as zero-mean and independent across all $i\in[N]$ due to \cref{assumption_noise}. Then, we have 
\begin{align}
    \frac{1}{N}  \Expectation\biggbrackets{\biggabs{\! \sum_{i\in\cR_s}\! \frac{2 \varepsilon_{i,j}^{(1)} \eta_{i,j} \bigparenth{\what{\theta}_{i,j}^{(1)} \!-\! \theta_{i,j}^{(1)}}}{\bigparenth{\what{p}_{i,j}}^2}} \Big| \normalbraces{\normalparenth{\what{p}_{i, j}, \what{\theta}_{i,j}^{(1)}}}_{i\in\cR_s}} & \!\sless{(a)} \frac{c\sigmax}{N} 
    \sqrt{\sum_{i\in \cR_s} \biggparenth{\frac{\what{\theta}_{i,j}^{(1)} \!-\! \theta_{i,j}^{(1)}}{\bigparenth{\what{p}_{i,j}}^2}}^2} \\
    & \! \sless{(b)} \! \frac{c \sigmax}{\lbar^2 N} \twonorm{\empoo_{\cdot,j} \!-\! \mpoo_{\cdot,j}} \\
    &  \sequal{(c)} \frac{c \sigmax}{\lbar^2} \frac{\Rcol\bigparenth{\empoo}}{\sqrt{N}} \leq \frac{c \sigmax}{\lbar^2} \frac{\RTheta}{\sqrt{N}} \sequal{(d)} o_p(1), \label{eq_decomp_variance_3}
\end{align}
where $(a)$ follows as the first moment of $\subE[\sigma]$ is $O(\sigma)$ \citep[Corollary 3]{zhang2022sharper}, $(b)$ follows from \cref{assumption_pos_estimated}, $(c)$ follows from \cref{eq_notation}, and $(d)$ follows from \cref{item_error_for_normality}.

To control fourth term in \cref{eq_decomp_variance_0}, 
we condition on $\normalbraces{\what{p}_{i, j}}_{i\in\cR_s}$. Then, \cref{eq_independence_requirement_estimates_new} provides that $\normalbraces{ \what{p}_{i, j} }_{i\in\cR_s} \indep \normalbraces{\normalparenth{\eta_{i, j}, \varepsilon^{(1)}_{i,j}}}_{i\in\cR_s}$. As a result, $\sum_{i \in \cR_s} \bigparenth{\varepsilon_{i,j}^{(1)}}^2 \eta_{i,j} / \bigparenth{\what{p}_{i,j}}^2$ is $\subWeibull_{2/3}$ $\bigparenth{\sigmax^2\bigbrackets{\sum_{i\in\cR_s} 1 / \bigparenth{\what{p}_{i,j}}^4}^{1/2} /\sqrt{\lone}}$ because $\normalparenth{\varepsilon_{i,j}^{(1)} }^2 \eta_{i,j}$ is $\subWeibull_{2/3}(\sigmax^2/\sqrt{\lone})$ due to \cref{lem_prod_subW} as well as zero-mean and independent across all $i\in[N]$ due to \cref{assumption_noise}. Then, we have
\begin{align}
\frac{1}{N}  \Expectation\biggbrackets{\biggabs{ \sum_{i \in \cR_s} \frac{ \bigparenth{\varepsilon_{i,j}^{(1)}}^2 \eta_{i,j}}{\bigparenth{\what{p}_{i,j}}^2}} \Big| \normalbraces{\what{p}_{i, j}}_{i\in\cR_s}} 
\sless{(a)} \frac{c\sigmax^2}{N}  \sqrt{\sum_{i\in \cR_s} \frac{1}{\bigparenth{\what{p}_{i,j}}^4}} \sless{(b)} \frac{c\sigmax^2}{\lbar^2\sqrt{N} } = o_p(1), 
    \label{eq_decomp_variance_4}
\end{align}
where $(a)$ follows as the first moment of $\subWeibull_{2/3}(\sigma)$ is $O(\sigma)$ \citep[Corollary 3]{zhang2022sharper} and $(b)$ follows from \cref{assumption_pos_estimated}.

To control fifth term in \cref{eq_decomp_variance_0}, we have
\begin{align}
     \biggabs{\!\sum_{i \in \cR_s} \!\! \biggparenth{\frac{ \bigparenth{\varepsilon_{i,j}^{(1)}}^2 p_{i,j}}{\bigparenth{\what{p}_{i,j}}^2} \!-\! \frac{(\sigma_{i,j}^{(1)})^2}{p_{i,j}}}} & =  \biggabs{\!\sum_{i \in \cR_s} \!\! \biggparenth{\frac{ \bigparenth{\varepsilon_{i,j}^{(1)}}^2 p_{i,j}}{\bigparenth{\what{p}_{i,j}}^2} \!-\! \frac{ \bigparenth{\sigma_{i,j}^{(1)}}^2 p_{i,j}}{\bigparenth{\what{p}_{i,j}}^2} \!+\! \frac{ \bigparenth{\sigma_{i,j}^{(1)}}^2 p_{i,j}}{\bigparenth{\what{p}_{i,j}}^2} \!-\! \frac{(\sigma_{i,j}^{(1)})^2}{p_{i,j}} } } \\
     & \sless{(a)}  \biggabs{\!\sum_{i \in \cR_s} \!\! \biggparenth{\frac{ \bigbrackets{\bigparenth{\varepsilon_{i,j}^{(1)}}^2 \!-\! \bigparenth{\sigma_{i,j}^{(1)}}^2} p_{i,j}}{\bigparenth{\what{p}_{i,j}}^2}}} \!+\! 
     \biggabs{\!\sum_{i \in \cR_s} \!\! \biggparenth{\frac{ \bigparenth{\sigma_{i,j}^{(1)}}^2 p_{i,j}}{\bigparenth{\what{p}_{i,j}}^2} \!-\! \frac{(\sigma_{i,j}^{(1)})^2}{p_{i,j}} } }, \label{eq_decomp_variance_5}
\end{align}
where $(a)$ follows from the triangle inequality. 
To control the first term in \cref{eq_decomp_variance_5}, we condition on $\normalbraces{\what{p}_{i, j}}_{i\in\cR_s}$. 
Then, \cref{eq_independence_requirement_estimates_new} provides that $\normalbraces{\what{p}_{i, j}}_{i\in\cR_s} \indep \normalbraces{\vareps_{i, j}^{(1)}}_{i\in\cR_s}$. Further, $\Expectation\normalbrackets{\normalparenth{\varepsilon_{i,j}^{(1)}}^2 - (\sigma_{i,j}^{(1)})^2} = 0$
due to \cref{item_sigma_bounded_below} and \cref{assumption_noise}. As a result, $\sum_{i \in \cR_s} \bigbrackets{\bigparenth{\varepsilon_{i,j}^{(1)}}^2 \!-\! \bigparenth{\sigma_{i,j}^{(1)}}^2} p_{i,j} / \bigparenth{\what{p}_{i,j}}^2$ is $\subExponential\bigparenth{\sigmax^2\bigbrackets{\sum_{i\in\cR_s} \bigparenth{p_{i,j}}^2 / \bigparenth{\what{p}_{i,j}}^4}^{1/2}}$ because $\normalparenth{\varepsilon_{i,j}^{(1)}}^2 - (\sigma_{i,j}^{(1)})^2$ is $ \subE[\sigmax^2]$ and independent across all $i \in [N]$ due to \cref{lem_prod_subG}. Then, we have 
\begin{align}
\frac{1}{N}  \Expectation\biggbrackets{\biggabs{ \sum_{i \in \cR_s} \frac{ \bigbrackets{\bigparenth{\varepsilon_{i,j}^{(1)}}^2 \!-\! \bigparenth{\sigma_{i,j}^{(1)}}^2} p_{i,j}}{\bigparenth{\what{p}_{i,j}}^2}} \Big| \normalbraces{\what{p}_{i, j}}_{i\in\cR_s}}
     \sless{(a)} \frac{c\sigmax^2}{N}  \sqrt{\sum_{i\in \cR_s} \biggparenth{\frac{p_{i,j}}{\bigparenth{\what{p}_{i,j}}^2}}^2} \sless{(b)} \frac{c\sigmax^2}{\lbar^2\sqrt{N} } = o_p(1), \label{eq_decomp_variance_51}
\end{align}
where $(a)$ follows as the first moment of $\subE[\sigma]$ is $O(\sigma)$ and $(b)$ follows from \cref{assumption_pos_estimated}. To bound the second term in \cref{eq_decomp_variance_5}, applying the Cauchy-Schwarz inequality yields that
\begin{align}
    \frac{1}{N} \biggabs{\sum_{i \in \cR_s} \! \biggparenth{\frac{ \bigparenth{\sigma_{i,j}^{(1)}}^2 p_{i,j}}{\bigparenth{\what{p}_{i,j}}^2} - \frac{(\sigma_{i,j}^{(1)})^2}{p_{i,j}} } } & = \frac{1}{N} \biggabs{\sum_{i \in \cR_s} {\frac{ \bigparenth{\sigma_{i,j}^{(1)}}^2 \bigparenth{\bigparenth{p_{i,j}}^2 - \bigparenth{\what{p}_{i,j}}^2}}{\bigparenth{\what{p}_{i,j}}^2 p_{i,j}}}}  \\
    & \sless{(a)} \frac{2}{N} {\sum_{i \in \cR_s} {\frac{\bigparenth{\sigma_{i,j}^{(1)}}^2 \bigabs{p_{i,j} - \what{p}_{i,j}}}{\bigparenth{\what{p}_{i,j}}^2 p_{i,j}}}} \\
    & \sless{(b)} \frac{2 \sigmax^2 }{\lambda \lbar^2 N} \sum_{i \in \cR_s}  \bigabs{p_{i,j} - \what{p}_{i,j}} \\
    & \sless{(c)} \frac{2 \sigmax^2 }{\lambda \lbar^2 \sqrt{N} }  \twonorm{ \mta_{\cdot,j} - \emta_{\cdot,j}}  \sequal{(d)} \frac{2 \sigmax^2 }{\lambda \lbar^2}\RP \sequal{(e)} o_p(1), \label{eq_decomp_variance_52}
\end{align}
where $(a)$ follows by using $\bigparenth{p_{i,j}}^2 - \bigparenth{\what{p}_{i,j}}^2 = (p_{i,j} + \what{p}_{i,j})(p_{i,j} - \what{p}_{i,j}) \leq 2 |p_{i,j} - \what{p}_{i,j}|$, $(b)$ follows from \cref{assumption_pos,assumption_pos_estimated}, and because the variance of a $\subGaussian$ random variable is upper bounded by the square of its $\subGaussian$ norm, $(c)$ follows by the relationship between $\ell_1$ and $\ell_2$ norms of a vector, $(d)$ follows from \cref{eq_notation}, and $(e)$ follows from \cref{item_error_for_normality}.

Putting together \cref{eq_decomp_variance_0,eq_decomp_variance_1,eq_decomp_variance_2,eq_decomp_variance_3,eq_decomp_variance_4,eq_decomp_variance_5,eq_decomp_variance_51,eq_decomp_variance_52} using \cref{lemma_bigOMarkov}, 
\begin{align}
     \frac{1}{N} \Bigabs{\sum_{i \in \cR_s} \termV[1]} =  o_p(1).
     \label{eq_var_final}
\end{align}
Then, the claim in \cref{eq_bound_11_var} follows for $a=1$ by using $|\sum_{i\in[N]}\termV[1]| \leq |\sum_{i\in\cR_0} \termV[1]|+|\sum_{i\in\cR_1} \termV[1]|$. The proof of \cref{eq_bound_11_var} for $a = 0$ follows similarly.

\section{Proof of \texorpdfstring{Proposition \ref{thm_fsg_ipw_oi} \MakeLowercase{\eqref{fsg_oi_ate_t}}}{}: \fsgoi}
\label{sec_proof_thm_fsg_oi}
Fix any $j \in [M]$. Recall the definitions of the parameter $\ATETrue$ and corresponding outcome imputation estimate $\ATEOI$ from \cref{eq_ate_parameter_combined,eq_counterfactual_mean_oi}, respectively. The error $\Delta \ATETrue^{\mathrm{OI}} = \ATEOI - \ATETrue$ 
can be re-expressed as
\begin{align}
    \Delta \ATETrue^{\mathrm{OI}} & = \frac{1}{N}\sum_{i \in [N]} \Bigparenth{\what{\theta}_{i,j}^{(1)} - \what{\theta}_{i,j}^{(0)}} - \frac{1}{N}\sum_{i \in [N]} \Bigparenth{\theta_{i,j}^{(1)} - \theta_{i,j}^{(0)}}\\
    & =  \frac{1}{N} \sum_{i \in [N]} \biggparenth{ \bigparenth{\what{\theta}_{i,j}^{(1)} - \theta_{i,j}^{(1)} } - \bigparenth{\what{\theta}_{i,j}^{(0)} - \theta_{i,j}^{(0)} }}.
\label{eq_ate_error_t_decomposition_oi_1}
\end{align}
Using the triangle inequality, we have
\begin{align}
    \bigabs{\Delta \ATETrue^{\mathrm{OI}}} \leq 
    \frac{1}{N} \Bigabs{\sum_{i \in [N]} \bigparenth{\what{\theta}_{i,j}^{(1)} - \theta_{i,j}^{(1)} } } + 
    \frac{1}{N} \Bigabs{\sum_{i \in [N]} \bigparenth{\what{\theta}_{i,j}^{(0)} - \theta_{i,j}^{(0)} } }.  \label{eq_ate_error_t_decomposition_oi_2}
\end{align}
Consider any $a \in \normalbraces{0,1}$. We claim that
\begin{align}
    \frac{1}{N} \Bigabs{\sum_{i \in [N]} \bigparenth{\what{\theta}_{i,j}^{(a)} - \theta_{i,j}^{(a)} } } \leq \Rcol\bigparenth{\empo}. \label{eq_ate_error_t_decomposition_oi_3}
\end{align}
The proof is complete by putting together \cref{eq_ate_error_t_decomposition_oi_2,eq_ate_error_t_decomposition_oi_3}.\medskip

\noindent {\bf Proof of \cref{eq_ate_error_t_decomposition_oi_3}}
Fix any $a \in \normalbraces{0,1}$. Using the Cauchy-Schwarz inequality, we have
\begin{align}
   \frac{1}{N} \Bigabs{ \sum_{i \in [N]} \bigparenth{\what{\theta}_{i,j}^{(1)} - \theta_{i,j}^{(1)}}} \leq \frac{1}{N} \stwonorm{\ones} \stwonorm{\empoo_{\cdot,j} - \mpoo_{\cdot,j}} = \frac{1}{\sqrt{N}}\stwonorm{\empoo_{\cdot,j} - \mpoo_{\cdot,j}} \leq \frac{1}{\sqrt{N}} \onetwonorm{\empoo - \mpoo}. \label{eq_guarantee_counterfactual_mean1_oi_t}
\end{align}

\section{Proof of \texorpdfstring{Proposition \ref{thm_fsg_ipw_oi} \MakeLowercase{\eqref{fsg_ipw_ate_t}}}{}: \fsgipw}
\label{sec_proof_thm_fsg_ipw}
Fix any $j \in [M]$. Recall the definitions of the parameter $\ATETrue$ and corresponding inverse probability weighting estimate $\ATEIPW$ from \cref{eq_ate_parameter_combined,eq_counterfactual_mean_ipw}, respectively. The error $\Delta \ATETrue^{\mathrm{IPW}} = \ATEIPW - \ATETrue$ 
can be re-expressed as
\begin{align}
\Delta \ATETrue^{\mathrm{IPW}} & = \frac{1}{N}\sum_{i \in [N]} \Bigparenth{\frac{ y_{i,j} a_{i,j}}{\what{p}_{i,j}} - \frac{ y_{i,j} \normalparenth{1-a_{i,j}}}{1 - \what{p}_{i,j}} } - \frac{1}{N}\sum_{i \in [N]} \Bigparenth{\theta_{i,j}^{(1)} - \theta_{i,j}^{(0)}}\\
& =  \frac{1}{N} \sum_{i \in [N]} \biggparenth{ \Bigparenth{  \frac{y_{i,j} a_{i,j}}{\what{p}_{i,j}  } - \theta_{i,j}^{(1)}} - \Bigparenth{\frac{ y_{i,j} \normalparenth{1-a_{i,j}}}{1 - \what{p}_{i,j}} - \theta_{i,j}^{(0)} }}\\
& \sequal{(a)}  \frac{1}{N}  \sum_{i \in [N]} \Bigparenth{\termIPW[1] + \termIPW[0] }, \label{eq_main_ipw} 
\end{align}
where $(a)$ follows after defining $\termIPW[1] \defn y_{i,j} a_{i,j} / \what{p}_{i,j}  - \theta_{i,j}^{(1)}$ and $\termIPW[0] \defn \theta_{i,j}^{(0)} - y_{i,j} \normalparenth{1-a_{i,j}} / (1 - \what{p}_{i,j})$. Then, we have
\begin{align}
    \termIPW[1] & = \frac{y_{i,j} a_{i,j}}{\what{p}_{i,j}  } - \theta_{i,j}^{(1)} \\ 
    & \sequal{(a)} \frac{\bigparenth{ \theta_{i,j}^{(1)} +  \varepsilon_{i,j}^{(1)}} \bigparenth{p_{i,j} + \eta_{i,j}} }{\what{p}_{i,j}  } - \theta_{i,j}^{(1)} \\
    & = \theta_{i,j}^{(1)} \Bigparenth{\frac{p_{i,j}+\eta_{i, j}}{\what{p}_{i,j}} - 1} +  \vareps_{i, j}^{(1)} \Bigparenth{\frac{p_{i,j}+\eta_{i, j}}{\what{p}_{i,j}}} \\
    & = \frac{  \theta_{i,j}^{(1)}\bigparenth{p_{i,j} - \what{p}_{i,j} }}{\what{p}_{i,j}} + \frac{\theta_{i,j}^{(1)}  \eta_{i,j}}{\what{p}_{i,j}} + \frac{\varepsilon_{i,j}^{(1)} p_{i,j}}{\what{p}_{i,j}} + \frac{\varepsilon_{i,j}^{(1)} \eta_{i,j}}{\what{p}_{i,j}}, \label{eq:mu1_err_ipw}
\end{align}
where $(a)$ follows from \cref{eq_consistency,eq_outcome_model,eq_treatment_model}. A similar derivation for $a = 0$ implies that  
\begin{align}
    \termIPW[0] & = \theta_{i,j}^{(0)} - \frac{ y_{i,j} \normalparenth{1-a_{i,j}}}{1 - \what{p}_{i,j}} \\
    & = - \frac{  \theta_{i,j}^{(0)}\bigparenth{1-p_{i,j} - \bigparenth{1-\what{p}_{i,j}} }}{1-\what{p}_{i,j}} - \frac{\theta_{i,j}^{(0)} (-\eta_{i, j})}{1-\what{p}_{i,j}} - \frac{\varepsilon_{i,j}^{(0)} \bigparenth{1-p_{i,j}}}{1-\what{p}_{i,j}} - \frac{\varepsilon_{i,j}^{(0)} (-\eta_{i, j})}{1-\what{p}_{i,j}}\\
    & = \frac{\theta_{i,j}^{(0)}\bigparenth{p_{i,j} - \what{p}_{i,j}}}{1-\what{p}_{i,j}} + \frac{\theta_{i,j}^{(0)} \eta_{i,j}}{1-\what{p}_{i,j}} - \frac{\varepsilon_{i,j}^{(0)} \bigparenth{1-p_{i,j}}}{1-\what{p}_{i,j}} + \frac{\varepsilon_{i,j}^{(0)} \eta_{i,j}}{1-\what{p}_{i,j}}. \label{eq:mu0_err_ipw}
\end{align}
Consider any $a \in \normalbraces{0,1}$ and any $\delta \in (0,1)$. We claim that, with probability at least $1 - 6\delta$, 
\begin{align}
    \frac{1}{N} \Bigabs{\sum_{i \in [N]} \termIPW[a]} \leq &   \frac{2}{\lbar} \maxmatnorm{ \mpo }\,  \RP + \frac{2 \sqrt{c \ld}}{\lbar \sqrt{\lone N}}\maxmatnorm{ \mpo } + \frac{2\sigmax \sqrt{c \ld}}{\lbar \sqrt{N}} + \frac{2\sigmax \m(c \ld)}{\lbar \sqrt{\lone N}}. \label{eq_bound_11_ipw}
\end{align}
where recall that $\m(c \ld) = \max\bigparenth{c\ld, \sqrt{c \ld}}$. We provide a proof of this claim at the end of this section. Applying triangle inequality in \cref{eq_main_ipw} and using \cref{eq_bound_11_ipw} with a union bound, we obtain that
\begin{align}
    \bigabs{\Delta \ATETrue^{\mathrm{IPW}}} \leq \frac{2}{\lbar} \thetamax\,  \RP + \frac{2 \sqrt{c \ld}}{\lbar \sqrt{\lone N}}\thetamax + \frac{4\sigmax \sqrt{c \ld}}{\lbar \sqrt{N}} + \frac{4\sigmax \m(c \ld)}{\lbar \sqrt{\lone N}},
\end{align}
with probability at least $1-12\delta$.
The claim in \cref{fsg_ipw_ate_t} follows  by  re-parameterizing $\delta$.\medskip

\noindent {\bf Proof of \cref{eq_bound_11_ipw}.} This proof follows a very similar road map to that used for establishing the inequality in \cref{eq_bound_11_dr}. Recall the partitioning of the units $[N]$ into $\cR_0$ and $\cR_1$ from \cref{assumption_estimates}. Now, to enable the application of concentration bounds, we split the summation over $i \in [N]$ in the left hand side of \cref{eq_bound_11_ipw} into two parts---one over $i \in \cR_0$ and the other over $i \in \cR_1$---such that the noise terms are independent of the estimates of $\mpoz, \mpoo, \mta$ in each of these parts as in \cref{eq_independence_requirement_estimates_dr,eq_independence_requirement_estimates_dr_p}. 

Fix $a = 1$ and note that $|\sum_{i\in[N]}\termIPW[1]| \leq |\sum_{i\in\cR_0} \termIPW[1]|+|\sum_{i\in\cR_1} \termIPW[1]|$. Fix any $s \in \normalbraces{0,1}$. Then, \cref{eq:mu1_err_ipw} and triangle inequality imply that
\begin{align}
    \Bigabs{\sum_{i\in\cR_s}\termIPW[1]} \!\leq\! \Bigabs{\sum_{i\in\cR_s}\frac{ \theta_{i,j}^{(1)}\bigparenth{  p_{i,j} \!-\! \what{p}_{i,j}}}{\what{p}_{i,j}}} \!+\!\Bigabs{\sum_{i\in\cR_s}\frac{\theta_{i,j}^{(1)} \eta_{i,j}}{\what{p}_{i,j}}} + \Bigabs{\sum_{i\in\cR_s}\frac{\varepsilon_{i,j}^{(1)} p_{i,j}}{\what{p}_{i,j}}}
     \!+\! \Bigabs{\sum_{i\in\cR_s}\frac{\varepsilon_{i,j}^{(1)} \eta_{i,j}}{\what{p}_{i,j}}}.
     \label{eq:decomp-T_ipw}
\end{align}

Next, note that the decomposition in \cref{eq:decomp-T_ipw} is identical to the one in \cref{eq:decomp-T}, except for the fact when compared to \cref{eq:decomp-T}, the first two terms in \cref{eq:decomp-T_ipw} have a factor of $\theta_{i,j}^{(1)}$ instead of $\bigparenth{\what{\theta}_{i,j}^{(1)} \!-\! \theta_{i,j}^{(1)}}$. As a result, mimicking steps used to derive \cref{eq_dr_almost_final}, we obtain the following bound, with probability at least $1-3\delta$,  
\begin{align}
    \frac{1}{N} \Bigabs{\sum_{i \in \cR_s} \termIPW[1]} 
    & \!\leq\! 
    \frac{1}{\lbar N} \onetwonorm{ \mpoo } \onetwonorm{\emta \!-\! \mta} \!+\! \frac{\sqrt{c \ld}}{ \lbar  \sqrt{\lone} N}\onetwonorm{ \mpoo } \!+\! \frac{\sigmax \sqrt{c \ld}}{\lbar N}\onetwonorm{ \mta } \!+\! \frac{\sigmax \m(c \ld)}{\lbar \sqrt{\lone} N}\onetwonorm{ \allones } \\
    & \!\sless{(a)}\! 
    \frac{1}{\lbar \sqrt{N}} \maxmatnorm{ \mpoo } \onetwonorm{\emta \!-\! \mta} \!+\! \frac{\sqrt{c \ld}}{\lbar\sqrt{\lone N}}\maxmatnorm{ \mpoo } \!+\! \frac{\sigmax \sqrt{c \ld}}{\lbar \sqrt{N}} \!+\! \frac{\sigmax \m(c \ld)}{\lbar \sqrt{\lone N}}\\
    & \!\sless{(b)}\! \frac{1}{\lbar} \maxmatnorm{ \mpoo }\,  \RP + \frac{\sqrt{c \ld}}{\lbar \sqrt{\lone N}}\maxmatnorm{ \mpoo }  + \frac{\sigmax \sqrt{c \ld}}{\lbar \sqrt{N}} + \frac{ \sigmax \m(c \ld)}{\lbar \sqrt{\lone N}},
    \label{eq_ipw_final}
\end{align}
where $(a)$ follows because $\onetwonorm{\mpoo} \leq \sqrt{N} \maxmatnorm{\mpoo}$, $\onetwonorm{ \mta } \leq \sqrt{N}$ and $\onetwonorm{ \allones } = \sqrt{N}$, and $(b)$ follows from \cref{eq_notation}. Then, the claim in \cref{eq_bound_11_ipw} follows for $a=1$ by using \cref{eq_ipw_final} and applying a union bound over $s \in \normalbraces{0,1}$. The proof of \cref{eq_bound_11_ipw} for $a = 0$ follows similarly.

\section{Proofs of \texorpdfstring{Propositions \ref{prop_ssmc} and \ref{alg_guarantee}}{}}
\label{sec_alg_guarantee_proof}
In \cref{proof_prop_ssmc}, we prove \cref{prop_ssmc}, i.e., we show that the estimates of $\mta$, $\mpoz$, and $\mpoo$ generated by $\ssMC$ satisfy \cref{assumption_estimates}. Next, we prove \cref{alg_guarantee} implying that the estimates of $\mta$, $\mpoz$, and $\mpoo$ generated by $\ssSVD$ satisfy the condition \cref{item_product_of_error_for_normality} in \cref{thm_normality} as long as $\sqrt{N} / M = o(1)$. 

\subsection{Proof of \texorpdfstring{\cref{prop_ssmc}}{}: \texorpdfstring{\propssmc}{}}
\label{proof_prop_ssmc}
Consider any matrix completion algorithm $\mca$. We show that
\begin{gather}
\emta_{\cI}, \empo_{\cI} \indep \noisea_{\cI}
\label{eq_independence_requirement_estimates_dr_new}
\shortintertext{and}
\emta_{\cI} \indep \noisea_{\cI}, \noisey_{\cI},
\label{eq_independence_requirement_estimates_dr_p_new}
\end{gather}
for every $\cI \in \cP$ and $a \in \normalbraces{0,1}$, where  $\cP$ is the block partition of $[N] \times [M]$ into four blocks from \cref{assumption_block_noise}. Then, \cref{eq_independence_requirement_estimates_dr,eq_independence_requirement_estimates_dr_p} in \cref{assumption_estimates} follow from \cref{eq_independence_requirement_estimates_dr_new,eq_independence_requirement_estimates_dr_p_new}, respectively.

Consider $\empoz, \empoo$, and $\emta$ as in \cref{ssmc_1,ssmc_2,ssmc_3}. Fix any $a \in \normalbraces{0,1}$. From \cref{eq_mc_sample_split}, note that $\emta_{\cI}$ depends only on $\ta  \otimes \allones^{-\cI}$ and $\empo_{\cI}$ depends on $\ooa  \otimes \allones^{-\mc I}$. In other words, the randomness in $\bigparenth{\emta_{\cI}, \empo_{\cI}}$ stems from the randomness in $\bigparenth{\ta_{-\cI}, \ooa[{-\cI}]}$ which in turn stems from the randomness in $\bigparenth{\noisea_{-\cI}, \noisey_{-\cI}}$. Then, \cref{eq_independence_requirement_estimates_dr_new} follows from \cref{eq_noise_independence_1}. Likewise, the randomness in $\emta_{\cI}$ stems from the randomness in $\ta_{-\cI}$ which in turn stems from the randomness in $\noisea_{-\cI}$. Then, \cref{eq_independence_requirement_estimates_dr_p_new} follows from \cref{eq_noise_independence_2}. 

To prove \cref{eq_independence_requirement_estimates_new}, we show that
\begin{gather}
\emta_{\cI}, \empo_{\cI}  \indep \noisea_{\cI}, \noisey_{\cI},
\label{eq_independence_requirement_estimates_dr_3}
\end{gather}
for every $\cI \in \cP$ and $a \in \normalbraces{0,1}$.
As mentioned above, the randomness in $\bigparenth{\emta_{\cI}, \empo_{\cI}}$ stems from the randomness in $\bigparenth{\ta_{-\cI}, \ooa[{-\cI}]}$ which in turn stems from the randomness in $\bigparenth{\noisea_{-\cI}, \noisey_{-\cI}}$. Then, \cref{eq_independence_requirement_estimates_dr_3} follows from \cref{eq_noise_independence_3}.

\subsection{Proof of \texorpdfstring{\cref{alg_guarantee}}{}: \texorpdfstring{\propsssvd}{}}
\label{proof_bai_ng_prop}
To prove this result, we first derive a corollary of Lemma A.1 in \citet{bai2021matrix} for a generic matrix of interest $\matT$, such that $\matS = (\matT + \matH) \otimes \matmask$, and apply it to $\mta$, $\mpoz \odot (\allones - \mta)$, and $\mpoo \odot \mta$. We impose the following restrictions on $\matT$, $\matH$, and $\matmask$.
\begin{assumption}[Strong linear latent factors]
\label{ass_T_factors}
    There exist a constant $r_{T} \in [\min\sbraces{N, M}]$ and a collection of latent factors
\begin{align}
	\wtil{\matU} \in \Reals^{N \times r_{T}} 
	\qtext{and} \wtil{\matV} \in \Reals^{M \times r_{T}},
\end{align}
such that, 
\begin{enumerate}[itemsep=1mm, topsep=2mm, label=(\alph*)]
    \item \label{item_ass_factors_a} $\matT$ satisfies the factorization: $\matT = \wtil{\matU} \wtil{\matV}^{\top}$,
    \item \label{item_ass_factors_b} $\twoinfnorm{\wtil{\matU}} \leq c$ and $\twoinfnorm{\wtil{\matV}} \leq c$ for some positive constant $c$, and
    \item \label{item_ass_factors_c} The matrices defined below exist and are positive definite:
    \begin{align}
        \lim_{N \to \infty} \frac{\wtil{\matU}^{\top} \wtil{\matU}}{N} \qtext{and} \lim_{M \to \infty} \frac{\wtil{\matV}^{\top} \wtil{\matV}}{M}.
    \end{align}
\end{enumerate}
\end{assumption}

\begin{assumption}[Zero-mean, weakly dependent, and $\subExponential$ noise]
\label{ass_H_noise} The noise matrix $\matH$ is such that,
\begin{enumerate}[itemsep=1mm, topsep=2mm, label=(\alph*)]
    \item\label{item_h_ass_2a}  $\sbraces{h_{i,j} : i\in[N], j \in [M]}$ are zero-mean $\subExponential$ with the $\subExponential$ norm bounded by a constant $\sigmax$,
    \item\label{item_h_ass_2b} $\sum_{j' \in [M]} \bigabs{\Expectation\normalbrackets{h_{i,j} h_{i,j'}}} \leq c$ for every $i \in [N]$ and $j \in [M]$, and 
    \item\label{item_h_ass_2c}  The elements of $\sbraces{\matH_{i,\cdot} : i\in[N]}$ are mutually independent (across $i$).
\end{enumerate}
\end{assumption}
\begin{assumption}[Strong block factors]
\label{ass_T_factors_extra}
    Consider the latent factors $\wtil{\matU} \in \Reals^{N \times r_{T}}$ and $\wtil{\matV} \in \Reals^{M \times r_{T}}$ from \cref{ass_T_factors}. Let $\cR_{\mathrm{obs}}  \subseteq [N]$ and $\cC_{\mathrm{obs}}  \subseteq [M]$ denote the set of rows and columns of $\matS$, respectively, with all entries observed, and $\cR_{\mathrm{miss}} \defn [N] \setminus \cR_{\mathrm{obs}}$ and $\cC_{\mathrm{miss}} \defn [M]  \setminus \cC_{\mathrm{obs}}$. Let $\wtil{\matU}^{\mathrm{obs}} \in \Reals^{\normalabs{\cR_{\mathrm{obs}}} \times r_T}$ and $\wtil{\matU}^{\mathrm{miss}} \in \Reals^{\normalabs{\cR_{\mathrm{miss}}} \times r_T}$ be the sub-matrices of $\wtil{\matU}$ that keeps the rows corresponding to the indices in $\cR_{\mathrm{obs}}$ and $\cR_{\mathrm{miss}}$, respectively. Let $\wtil{\matV}^{\mathrm{obs}} \in \Reals^{\normalabs{\cC_{\mathrm{obs}}} \times r_T}$ and $\wtil{\matV}^{\mathrm{miss}} \in \Reals^{\normalabs{\cC_{\mathrm{miss}}} \times r_T}$ be the sub-matrices of $\wtil{\matV}$ that keeps the rows corresponding to the indices in $\cC_{\mathrm{obs}}$ and $\cC_{\mathrm{miss}}$, respectively.
    Then, the matrices defined below exist and are positive definite:  
    \begin{align}
        \lim_{N \to \infty}
        \frac{\wtil{\matU}^{\mathrm{obs}\top} \wtil{\matU}^{\mathrm{obs}}}{\normalabs{\cR_{\mathrm{obs}}}}, \!\!\quad 
        \lim_{M \to \infty}
        \frac{\wtil{\matU}^{\mathrm{miss}\top} \wtil{\matU}^{\mathrm{miss}}}{\normalabs{\cR_{\mathrm{miss}}}}, \!\!\quad 
        \lim_{N \to \infty}
        \frac{\wtil{\matV}^{\mathrm{obs}\top} \wtil{\matV}^{\mathrm{obs}}}{\normalabs{\cC_{\mathrm{obs}}}}, \!\!\qtext{and}\!\! 
        \lim_{M \to \infty}
        \frac{\wtil{\matV}^{\mathrm{miss}\top} \wtil{\matV}^{\mathrm{miss}}}{\normalabs{\cC_{\mathrm{miss}}}}. \label{eq_strong_block_factors}
    \end{align}
    Further, the mask matrix $\matmask$ is such that
    \begin{align}
 \normalabs{\cR_{\mathrm{obs}}}= \Omega(N), \qquad \normalabs{\cR_{\mathrm{miss}}}= \Omega(N), \qquad \normalabs{\cC_{\mathrm{obs}}}= \Omega(M), \qtext{and} \normalabs{\cC_{\mathrm{miss}}}= \Omega(M). \label{eq_mask_matrix_size}
    \end{align}
\end{assumption}
The next result characterizes the entry-wise error in recovering the missing entries of a matrix where all entries in one block are deterministically missing (see the discussion in \cref{sub:kbb}) using the $\tallwide$ algorithm (summarized in \cref{sec:TW_algorithm}). Its proof, essentially established as a corollary of \citet[Lemma~A.1]{bai2021matrix}, is provided in \cref{subsec_proof_coro_bai_ng}.
\begin{corollary}
\label{coro_bai_ng}
    Consider a matrix of interest $\matT$, a noise matrix $\matH$, and a mask matrix $\matmask$ such that that \cref{ass_T_factors,ass_H_noise,ass_T_factors_extra} hold.
    Let $\matS  \in \{\Reals \cup \{\star\}\}^{N \times M}$ be the observed matrix as in \cref{eq_mc_main}. 
    Let $\cR_{\mathrm{obs}}  \subseteq [N]$ and $\cC_{\mathrm{obs}}  \subseteq [M]$ denote the set of rows and columns of $\matS$, respectively, with all entries observed.
    Let $\cI = \cR_{\mathrm{miss}} \times \cC_{\mathrm{miss}}$ where $\cR_{\mathrm{miss}} \defn [N] \setminus \cR_{\mathrm{obs}}$ and $\cC_{\mathrm{miss}} \defn [M]  \setminus \cC_{\mathrm{obs}}$. 
    Then, $\tallwide_{r_T}$ produces an estimate $\what{\matT}_{\mc I}$ of $\matT_{\cI}$ such that 
    \begin{align}
        \maxmatnorm{\what{\matT}_{\mc I} - \matT_{\cI}} = O_p\biggparenth{\frac{1}{\sqrt{N}} + \frac{1}{\sqrt{M}}}, \label{eq_bai_ng_coro}
    \end{align}
    as $N, M  \to \infty$.
\end{corollary}
Given this corollary, we now complete the proof of \cref{alg_guarantee}. Consider the partition $\cP$ from \cref{assumption_block_noise} and fix any $\cI \in \cP$. Recall that $\ssSVD$ applies $\tallwide$ on $\mta \otimes \allones^{-\cI}$, $\barz \otimes \allones^{-\cI}$, and $\baro \otimes \allones^{-\cI}$, and note that the mask matrix $\allones^{-\cI}$ satisfies the requirement in \cref{ass_T_factors_extra}, i.e., \cref{eq_mask_matrix_size} under \cref{ass_stron_block_factors}.

\subsubsection{Estimating \texorpdfstring{$\mta$}{}.} Consider estimating $\mta$ using $\ssSVD$. To apply \cref{coro_bai_ng}, we use \cref{assumption_lf,assumption_lf_strong} to note that $\mta$ satisfies \cref{ass_T_factors} with rank parameter $r_p$. Then, we use \cref{eq_eta_defn}, \cref{assumption_noise}, and \cref{ass_weak_noise} to note that $\noisea$ satisfies \cref{ass_H_noise}. Finally, we use \cref{ass_stron_block_factors} to note that \cref{ass_T_factors_extra} holds. Step \cref{item_ssSVD_P} of $\ssSVD$ can be rewritten as $\emta = \texttt{Proj}_{\lbar} \bigparenth{\wbar{\mta}}$ and $\wbar{\mta} = \ssMC(\tallwide_{r_1}, \ta, \cP)$ where $r_1 = r_p$. Then,
\begin{align}
    \maxmatnorm{\emta_{\mc I} - \mta_{\cI}} \sless{(a)} \maxmatnorm{\wbar{\mta}_{\mc I} - \mta_{\cI}} \sequal{(b)} O_p\biggparenth{\frac{1}{\sqrt{N}} + \frac{1}{\sqrt{M}}},
\end{align}
where $(a)$ follows from \cref{assumption_pos}, the choice of $\lbar$, and the definition of $\texttt{Proj}_{\lbar}(\cdot)$, and $(b)$ follows from \cref{coro_bai_ng}. Applying a union bound over all $\cI \in \cP$, we have
\begin{align}
    \RP \sless{(a)} \maxmatnorm{\emta - \mta} = O_p\biggparenth{\frac{1}{\sqrt{N}} + \frac{1}{\sqrt{M}}}, \label{eq_p_final_bai_ng}
\end{align}
where $(a)$ follows from the definition of $(1,2)$ operator norm.
\subsubsection{Estimating \texorpdfstring{$\mpoz$}{} and \texorpdfstring{$\mpoo$}{}.} 
For every $a \in \normalbraces{0,1}$, we show that
\begin{align}
    \Rcol\bigparenth{\empo} = O_p\biggparenth{\frac{1}{\sqrt{N}} + \frac{1}{\sqrt{M}}}. \label{eq_t1_to_show}
\end{align}
We focus on $a = 1$ noting that the proof for $a = 0$ is analogous. We split the proof in two cases: (i) $\maxmatnorm{ \bigparenth{\empoo - \mpoo} \odot \emta} \leq  \maxmatnorm{ \mpoo \odot \bigparenth{\emta - \mta} }$ and (ii) $\maxmatnorm{ \bigparenth{\empoo - \mpoo} \odot \emta} \geq  \maxmatnorm{ \mpoo \odot \bigparenth{\emta - \mta} }$.

In the first case, we have
\begin{align}
    \lbar \maxmatnorm{\empoo - \mpoo} \sless{(a)} \maxmatnorm{ \bigparenth{\empoo - \mpoo} \odot \emta} & \leq  \maxmatnorm{ \mpoo \odot \bigparenth{\emta - \mta} } \\
    & \sless{(b)} \maxmatnorm{\mpoo} \maxmatnorm{\emta - \mta}, \label{eq_t1_inter_bai_ng}
\end{align}
where $(a)$ follows from \cref{assumption_pos_estimated} and $(b)$ follows from the definition of $\maxmatnorm{\mpoo}$. Then, 
\begin{align}
    \Rcol\bigparenth{\empoo} \sless{(a)} \maxmatnorm{\empoo \!-\! \mpoo} \sless{(b)} \frac{\maxmatnorm{\mpoo}}{\lbar} \maxmatnorm{\emta \!-\! \mta} \sequal{(c)} \frac{\maxmatnorm{\mpoo}}{\lbar} O_p\biggparenth{\frac{1}{\sqrt{N}} + \frac{1}{\sqrt{M}}},
\end{align}
where $(a)$ follows from the definition of $(1,2)$ operator norm, $(b)$ follows from \cref{eq_t1_inter_bai_ng}, and $(c)$ follows from \cref{eq_p_final_bai_ng}. 
Then, \cref{eq_t1_to_show} follows as $1/\lbar$ and $\maxmatnorm{\mpoo}$ are assumed to be bounded.

In the second case, 
using \cref{eq_outcome_model,eq_treatment_model} to expand $\baro$, we have
\begin{align}
    \baro & = \mpoo \odot \mta + \mpoo \odot \noisea + \noiseyo \odot \mta + \noiseyo \odot \noisea.
\end{align}
Next, we utilize two claims proven in \cref{subsubsec_claim_p_t_factors,subsubsec_claim_p_t_noise} respectively: $\mpoo \odot \mta$ satisfies \cref{ass_T_factors} with rank parameter $r_{\theta_1} r_p$ and 
\begin{align}
    \wbar{\varepsilon}^{(1)} \defn \mpoo \odot \noisea + \noiseyo \odot \mta + \noiseyo \odot \noisea,
\end{align}
satisfies \cref{ass_H_noise}. Finally, \cref{ass_stron_block_factors} in \cref{sec_algorithm} implies that \cref{ass_T_factors_extra} holds. 

Now, note that step \cref{item_ssSVD_Theta1} of $\ssSVD$ can be rewritten as $\empoo = \bmpoo \odiv \emta$ and $\bmpoo = \ssMC(\tallwide_{r_3}, \baro, \cP)$ where $r_3 = r_{\theta_1} r_p$. Then, from \cref{coro_bai_ng},
\begin{align}
    \maxmatnorm{\bmpoo_{\cI} - \mpoo_{\cI} \odot \mta_{\cI}} = O_p\biggparenth{\frac{1}{\sqrt{N}} + \frac{1}{\sqrt{M}}}.
\end{align}
Applying a union bound over all $\cI \in \cP$ and noting that $\bmpoo = \empoo \odot \emta$, we have
\begin{align}
    \maxmatnorm{\empoo \odot \emta - \mpoo \odot \mta} = O_p\biggparenth{\frac{1}{\sqrt{N}} + \frac{1}{\sqrt{M}}}. \label{eq_union_bound_theta_p}
\end{align}
The left hand side of \cref{eq_union_bound_theta_p} can be written as,
\begin{align}
    \maxmatnorm{\empoo \odot \emta - \mpoo \odot \mta} & = \maxmatnorm{\empoo \odot \emta - \mpoo \odot \emta + \mpoo \odot \emta - \mpoo \odot \mta} \\
    & \sgreat{(a)} \maxmatnorm{ \bigparenth{\empoo - \mpoo} \odot \emta} - \maxmatnorm{ \mpoo \odot \bigparenth{\emta - \mta} } \\
    & \sgreat{(b)} \lbar \maxmatnorm{\empoo - \mpoo}  - \maxmatnorm{\mpoo} \maxmatnorm{\emta - \mta}, \label{eq_triangle_bai_ng}
\end{align}
where $(a)$ follows from triangle inequality as $\maxmatnorm{ \bigparenth{\empoo - \mpoo} \odot \emta} \geq  \maxmatnorm{ \mpoo \odot \bigparenth{\emta - \mta} }$ and $(b)$ follows from the choice of $\lbar$ and the definition of $\maxmatnorm{\mpoo}$. Then,
\begin{align}
    \Rcol\bigparenth{\empoo} \sless{(a)} \maxmatnorm{\empoo - \mpoo} & \sless{(b)} \frac{1}{\lbar} \maxmatnorm{\empoo \odot \emta - \mpoo \odot \mta} + \frac{\maxmatnorm{\mpoo}}{\lbar} \maxmatnorm{\emta \!-\! \mta} \\
    & \sequal{(c)} \frac{1}{\lbar} O_p\biggparenth{\frac{1}{\sqrt{N}} + \frac{1}{\sqrt{M}}} + \frac{\maxmatnorm{\mpoo}}{\lbar} O_p\biggparenth{\frac{1}{\sqrt{N}} + \frac{1}{\sqrt{M}}},
\end{align}
where $(a)$ follows from the definition of $L_{1,2}$ norm, $(b)$ follows from \cref{eq_triangle_bai_ng},
and $(c)$ follows from \cref{eq_p_final_bai_ng,eq_union_bound_theta_p}. Then, \cref{eq_t1_to_show} follows as $1/\lbar$ and $\maxmatnorm{\mpoo}$ are assumed to be bounded.

\subsubsection{Proof that \texorpdfstring{$\mpoz \odot (\allones - \mta)$ \text{ and } $\mpoo \odot \mta$}{} satisfy  \texorpdfstring{\cref{ass_T_factors}}{}.}
\label{subsubsec_claim_p_t_factors}
First, we show that $\wbar{\matU}^{(0)} \in \Reals^{N \times r_{\theta_0}(r_p + 1)}$ and $\wbar{\matV}^{(0)} \in \Reals^{N \times r_{\theta_0}(r_p + 1)}$ are factors of $\mpoz \odot (\allones - \mta)$, and $\wbar{\matU}^{(1)} \in \Reals^{N \times r_{\theta_1}r_p}$ and $\wbar{\matV}^{(1)} \in \Reals^{N \times r_{\theta_1}}$ are factors of $\mpoo \odot \mta$ as claimed in \cref{eq_factors_theta_p}. We have
\begin{align}
    \mpoo \odot \mta = \Bigparenth{\sum_{i \in [r_{\theta_1}]} U_{i, \cdot}^{(1)} V_{i, \cdot}^{(1)^\top}} \odot \Bigparenth{\sum_{j \in [r_p]} U_{j, \cdot} V_{j, \cdot}^{\top}} & = \sum_{i \in [r_{\theta_1}]} \sum_{j \in [r_p]} \Bigparenth{ U_{i, \cdot}^{(1)} \odot U_{j, \cdot}  }  \Bigparenth{V_{i, \cdot}^{(1)} \odot V_{j, \cdot}}^\top \\
    & \sequal{(a)} \bigparenth{\matU * \matU^{(1)}} \bigparenth{\matV * \matV^{(1)} }^\top \sequal{(b)} \wbar{\matU}^{(1)} \wbar{\matV}^{(1)^\top},
\end{align}
where $(a)$ follows from the definition of Khatri-Rao product (see \cref{section_introduction}) and $(b)$ follows from the definitions of $\wbar{\matU}^{(1)}$ and $\wbar{\matV}^{(1)}$. The proof for $\mpoz \odot (\allones - \mta)$ follows similarly.
Then, \cref{ass_T_factors}\cref{item_ass_factors_a} holds from \cref{eq_factors_theta_p}. Next, we note that
\begin{align}
    \twoinfnorm{\wbar{\matU}^{(1)}} =  \twoinfnorm{\matU * \matU^{(1)}} \sequal{(a)} \max_{i \in [N]} \sqrt{\sum_{j \in [r_p]} u_{i,j}^2 \sum_{j' \in [r_{\theta_1} ]} \normalparenth{u^{(1)}_{i,j'}}^2} \leq \twoinfnorm{\matU} \twoinfnorm{\matU^{(1)}} \sless{(b)} c,
\end{align}
where $(a)$ follows from the definition of Khatri-Rao product (see \cref{section_introduction}), and $(b)$ follows from \cref{assumption_lf_strong}. Then, $\mpoo \odot \mta$ satisfies \cref{ass_T_factors}\cref{item_ass_factors_b} by using similar arguments on $\wbar{\matV}^{(1)}$. Further, $\mpoz \odot (\allones - \mta)$ satisfies \cref{ass_T_factors}\cref{item_ass_factors_b} by noting that $\twoinfnorm{\wbar{\matU}}$ and $\twoinfnorm{\wbar{\matV}}$ are bounded whenever $\twoinfnorm{\matU}$ and $\twoinfnorm{\matV}$ are bounded, respectively. Finally, \cref{ass_T_factors}\cref{item_ass_factors_c} holds from \cref{assumption_lf_strong}.

\subsubsection{Proof that \texorpdfstring{$\wbar{\varepsilon}^{(1)}$}{} satisfies \texorpdfstring{\cref{ass_H_noise}}{}}
\label{subsubsec_claim_p_t_noise}
Recall that $\wbar{\varepsilon}^{(1)} \defn \mpoo \odot \noisea + \noiseyo \odot \mta + \noiseyo \odot \noisea$. Then, \cref{ass_H_noise}\cref{item_h_ass_2a} holds as $\wbar{\varepsilon}^{(1)}_{i,j}$ is zero-mean from \cref{assumption_noise} and \cref{eq_treatment_model}, and $\wbar{\varepsilon}^{(1)}_{i,j}$ is $\subExponential$ because $\varepsilon^{(1)}_{i,j} \eta_{i,j}$ is a $\subExponential$ random variable \cref{lem_prod_subG}, every $\subGaussian$ random variable is $\subExponential$ random variable, and sum of $\subExponential$ random variables is a $\subExponential$ random variable. Finally, \cref{ass_H_noise}\cref{item_h_ass_2b} and \cref{ass_H_noise}\cref{item_h_ass_2b} hold from \cref{ass_weak_noise}\cref{item_weak_ass_2bb} and \cref{ass_weak_noise}\cref{item_weak_ass_2cc}, respectively.

\subsection{Proof of \texorpdfstring{\cref{coro_bai_ng}}{}}
\label{subsec_proof_coro_bai_ng}
\cref{coro_bai_ng} is a direct application of \citet[Lemma~A.1]{bai2021matrix}, specialized to our setting. Notably, \cite{bai2021matrix} make three assumptions numbered A, B, and C in their paper to establish the corresponding result. It remains to establish that the conditions assumed in \cref{coro_bai_ng} imply the necessary conditions used in the proof of \citet[Lemma~A.1]{bai2021matrix}. First, note that certain assumptions in \citet{bai2021matrix} are not actually used in their proof of Lemma A.1 (or in the proof of other results used in that proof), namely, the distinct eigenvalue condition in Assumption A(a)(iii), the asymptotic normality conditions in Assumption A(c) and the asymptotic normality conditions in Assumption C. Next, \cref{eq_mask_matrix_size} in  \cref{ass_T_factors_extra} implies Assumption B and \cref{eq_strong_block_factors} in \cref{ass_T_factors_extra} is equivalent to the remaining conditions in Assumption C. 

It remains to show how \cref{ass_T_factors,ass_H_noise} imply the remainder of conditions in \citet[Assumptions A]{bai2021matrix}. 
For completeness, these conditions are collected in the following assumption.
\begin{assumption}
\label{ass_H_noise_extra} The noise matrix $\matH$ is such that,
\begin{enumerate}[itemsep=1mm, topsep=2mm, label=(\alph*)]
    \item\label{item_h_ass_2a_extra}  
    $\max_{j \in [M]} \frac{1}{N} \sum_{j' \in [M]} \bigabs{ \sum_{i \in [N]}  \Expectation\normalbrackets{h_{i,j} h_{i,j'}} } \leq c$,
    \item\label{item_h_ass_2b_extra} $\max_{j \in [M]} \bigabs{\Expectation\normalbrackets{h_{i,j} h_{i',j}}} \leq c_{i,i'}$ and $\max_{i \in [N]}  \sum_{i' \in [N]} c_{i,i'} \leq c$,
    \item\label{item_h_ass_2c_extra}   $\frac{1}{NM} \sum_{i,i' \in [N]} \sum_{j,j' \in [M]} \bigabs{ \Expectation\normalbrackets{h_{i,j} h_{i',j'}} } \leq c$, and
    \item\label{item_h_ass_2d_extra} $\max_{j, j' \in [M]} \frac{1}{N^2} \Expectation\bigbrackets{ \bigabs{ \sum_{i \in [N]} \bigparenth{h_{i,j} h_{i,j'} - \Expectation\normalbrackets{h_{i,j} h_{i,j'}}} }^4 }$.
\end{enumerate}
\end{assumption}
\cref{ass_H_noise_extra} is a restatement of the subset of conditions from \citet[Assumption A]{bai2021matrix} necessary in \citet[proof of Lemma A.1]{bai2021matrix} and it essentially requires weak dependence in the noise across measurements and across units. In particular, \cref{ass_H_noise_extra}\cref{item_h_ass_2a_extra}, \cref{item_h_ass_2b_extra}, \cref{item_h_ass_2c_extra}, and \cref{item_h_ass_2d_extra} correspond to Assumption A(b)(ii), (iii), (iv), (v), respectively, of \citet{bai2021matrix}. 
For the other conditions in \citet[Assumption A]{bai2021matrix}, note that \cref{ass_T_factors} above is equivalent to their Assumption A(a)(i) and (ii) of \citet{bai2021matrix} when the factors are non-random as in this work.  Similarly, \cref{ass_H_noise}\cref{item_h_ass_2a} above is analogous to Assumption A(b)(i) of \citet{bai2021matrix}.  Assumption A(b)(vi) of \citet{bai2021matrix} is implied by their other Assumptions for non-random factors as stated in \citet{bai2003inferential}. 

To establish \cref{coro_bai_ng}, it remains to establish that \cref{ass_H_noise_extra} holds, which is done in \cref{assum_11_holds} below.

\subsubsection{\texorpdfstring{\cref{ass_H_noise_extra}}{} holds}
\label{assum_11_holds}
First, \cref{ass_H_noise_extra}\cref{item_h_ass_2a_extra} holds as follows,
\begin{align}
    \max_{j \in [M]} \frac{1}{N} \sum_{j' \in [M]} \Bigabs{\sum_{i \in [N]}  \Expectation\bigbrackets{h_{i,j} h_{i,j'}} } \sless{(a)} \max_{j \in [M]} \frac{1}{N} \sum_{i \in [N]}  \sum_{j' \in [M]} \Bigabs{ \Expectation\bigbrackets{h_{i,j} h_{i,j'}} } \sless{(b)} \max_{j \in [M]} \frac{1}{N} \sum_{i \in [N]} c = c, 
\end{align}
where $(a)$ follows from triangle inequality and  $(b)$ follows from \cref{ass_H_noise}\cref{item_h_ass_2b}. Next, from 
\cref{ass_H_noise}\cref{item_h_ass_2a} and \cref{ass_H_noise}\cref{item_h_ass_2c}, we have
\begin{align} \label{cases_assum}
    \max_{j \in [M]} \bigabs{\Expectation\normalbrackets{h_{i,j} h_{i',j}}} =  \begin{cases}
        0 & \text{ if } i \neq i'\\
        \max_{j \in [M]} \bigabs{\Expectation\normalbrackets{h_{i,j}^2}} \leq c & \text{ if } i = i'
    \end{cases}
\end{align}
Then,  \cref{ass_H_noise_extra}\cref{item_h_ass_2b_extra} holds as follows,
\begin{align}
    \max_{i \in [N]} \max_{j \in [M]}  \sum_{i' \in [N]} \bigabs{\Expectation\normalbrackets{h_{i,j} h_{i',j}}} \leq c.
\end{align}
Next, \cref{ass_H_noise_extra}\cref{item_h_ass_2c_extra} holds as follows, 
\begin{align}
    \frac{1}{NM} \sum_{i,i' \in [N]} \sum_{j,j' \in [M]} \bigabs{ \Expectation\normalbrackets{h_{i,j} h_{i',j'}} } \sequal{(a)} \frac{1}{NM} \sum_{i \in [N]} \sum_{j,j' \in [M]} \bigabs{ \Expectation\normalbrackets{h_{i,j} h_{i,j'}} } \sless{(b)} \frac{1}{NM} \sum_{i \in [N]} \sum_{j \in [M]} c = c,
\end{align}
where $(a)$ follows from \cref{ass_H_noise}\cref{item_h_ass_2c} and $(b)$ follows from \cref{ass_H_noise}\cref{item_h_ass_2b}. Next, let $\gamma_{i,j,j'} \defn h_{i,j} h_{i,j'} - \Expectation\normalbrackets{h_{i,j} h_{i,j'}}$ and fix any $j,j' \in [M]$. Then, \cref{ass_H_noise_extra}\cref{item_h_ass_2d_extra} holds as follows,
\begin{align}
    \frac{1}{N^2} \Expectation\Bigbrackets{ \Bigparenth{ \sum_{i \in [N]} \gamma_{i,j,j'} }^4 } & = \frac{1}{N^2} \Expectation\Bigbrackets{ \Bigparenth{\sum_{i_1 \in [N]} \gamma_{i_1,j,j'}} \Bigparenth{\sum_{i_2 \in [N]} \gamma_{i_2,j,j'}} \Bigparenth{\sum_{i_3 \in [N]} \gamma_{i_3,j,j'}} \Bigparenth{\sum_{i_4 \in [N]} \gamma_{i_4,j,j'}} } \\
    & \sequal{(a)} \frac{1}{N^2} \sum_{i \in [N]} \Expectation\Bigbrackets{ \gamma_{i,j,j'}^4 } + \frac{3}{N^2} \sum_{i \neq i' \in [N]} \Expectation\Bigbrackets{ \gamma_{i,j,j'}^2 \gamma_{i',j,j'}^2 } \leq  c, 
\end{align}
where $(a)$ follows from linearity of expectation and \cref{ass_H_noise}\cref{item_h_ass_2c} after by noting that $\Expectation[\gamma_{i,j,j'}] = 0$ for all $i, j, j' \in [N] \times [M] \times [M]$ and $(b)$ follows because $\gamma_{i,j,j'}$ has bounded moments due to \cref{ass_H_noise}\cref{item_h_ass_2a}. 

\section{Data generating process for the simulations}
\label{app_dgp}
The inputs of the data generating process (DGP) are: the probability bound $\lambda$; two positive constants $c^{(0)}$ and $c^{(1)}$; and the standard deviations ${\sigma_{i,j}^{(a)}}$ for every $i \in [N], j \in [M], a \in \normalbraces{0,1}$. The DGP is:
\begin{enumerate}
\item\label{dgp_1} For positive integers $r_p$, $r_{\theta}$ and $r = \max\normalbraces{r_p, r_{\theta}}$,
generate a proxy for the common unit-level latent factors $\matU^{\trm{shared}} \in \Reals^{N \times r}$, such that, for all $i \in [N]$ and $j \in [r]$, $u^{\trm{shared}}_{i,j}$ is independently sampled from a $\Uniform(\sqrt{\lambda}, \sqrt{1-\lambda})$ distribution, with $\lambda\in (0,1)$.
\item\label{dgp_2} Generate proxies for the measurement-level latent factors $\matV, \matV^{(0)}, \matV^{(1)} \in \Reals^{M \times r}$, such that, for all $i \in [M]$ and $j \in [r]$, $v_{i,j}, v_{i,j}^{(0)}, v_{i,j}^{(1)}$ are independently sampled from a $\Uniform(\sqrt{\lambda}, \allowbreak\sqrt{1-\lambda})$ distribution.
\item\label{dgp_3} Generate the treatment assignment probability matrix $\mta$
\begin{align}
   \mta = \frac{1}{r_p} \matU^{\trm{shared}}_{[N] \times [r_p]} \matV_{[M] \times [r_p]}^{\top}.
\end{align}
\item\label{dgp_4} For $a \in \normalbraces{0,1}$, run \texttt{SVD} on $\matU^{\trm{shared}} \matV^{{(a)}\top}$, i.e.,
\begin{align}
    \texttt{SVD}(\matU^{\trm{shared}} \matV^{{(a)}\top}) = ( \matU^{(a)}, \Sigma^{(a)}, \matW^{(a)}  ).
\end{align}
Then, generate the mean potential outcome matrices $\mpoz$ and $ \mpoo$:
\begin{align}
    \mpo = \frac{c^{(a)} \texttt{Sum}(\Sigma^{(a)})}{r_{\theta}} \matU^{(a)}_{[N] \times [r_{\theta}]} \matW^{(a)\top}_{[M] \times [r_{\theta}]} ,
\end{align}
where $\texttt{Sum}(\Sigma^{(a)})$ denotes the sum of all entries of $\Sigma^{(a)}$.
\item\label{dgp_5} Generate the noise matrices $\noiseyz$ and $\noiseyo$, such that, for all $i \in [N], j \in [M], a \in \normalbraces{0,1}$, $\varepsilon^{(a)}_{i, j}$ is independently sampled from a $\cN(0, \normalparenth{\sigma_{i,j}^{(a)}}^2)$ distribution.
Then, determine $y^{(a)}_{i, j}$ from \cref{eq_outcome_model}. 
\item\label{dgp_6} Generate the noise matrix $\noisea$, such that, for all $i \in [N], j \in [M]$, $\eta_{i, j}$ is independently sampled as per \cref{eq_eta_defn}. Then, determine $a_{i,j}$ and $y_{i,j}$ from \cref{eq_treatment_model} and \cref{eq_consistency}, respectively.
\end{enumerate}
\begin{figure}[t!]
    \centering
    \begin{tabular}{cc}
    \includegraphics[width=0.45\linewidth,clip]{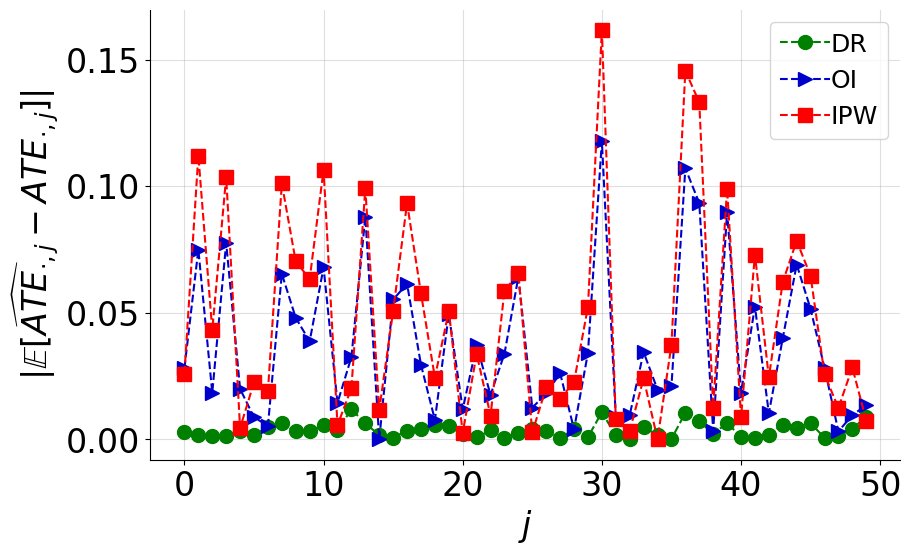} &
    \includegraphics[width=0.45\linewidth,clip]{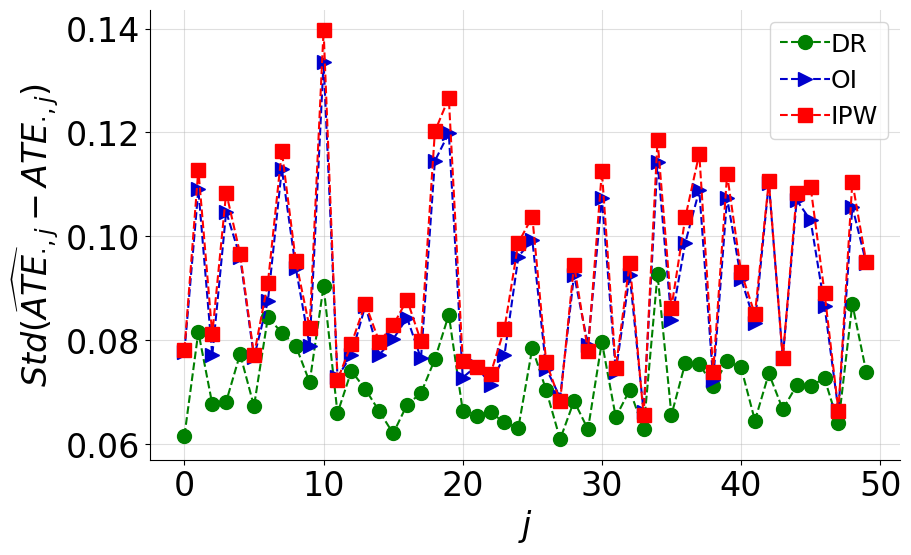}\\
    \multicolumn{2}{c}{$(a)$ $r_p = 3 $, $ r_{\theta} = 3$} \\
    \includegraphics[width=0.45\linewidth,clip]{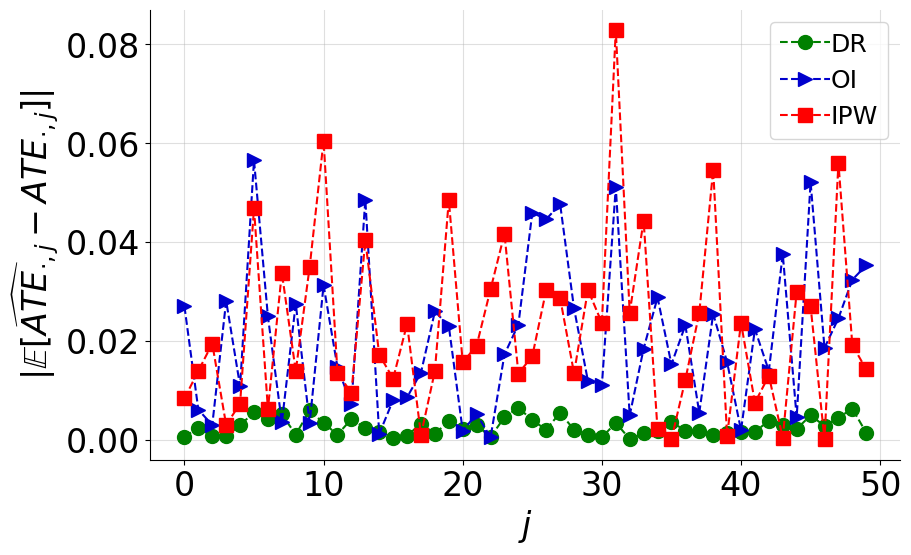} &
    \includegraphics[width=0.45\linewidth,clip]{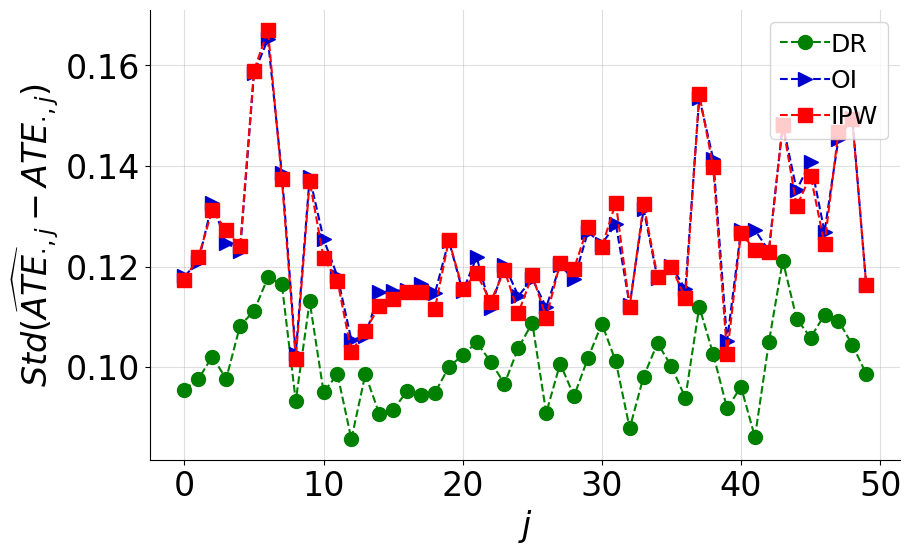}\\
    \multicolumn{2}{c}{$(b)$ $r_p = 3 $, $ r_{\theta} = 5$} \\
    \includegraphics[width=0.45\linewidth,clip]{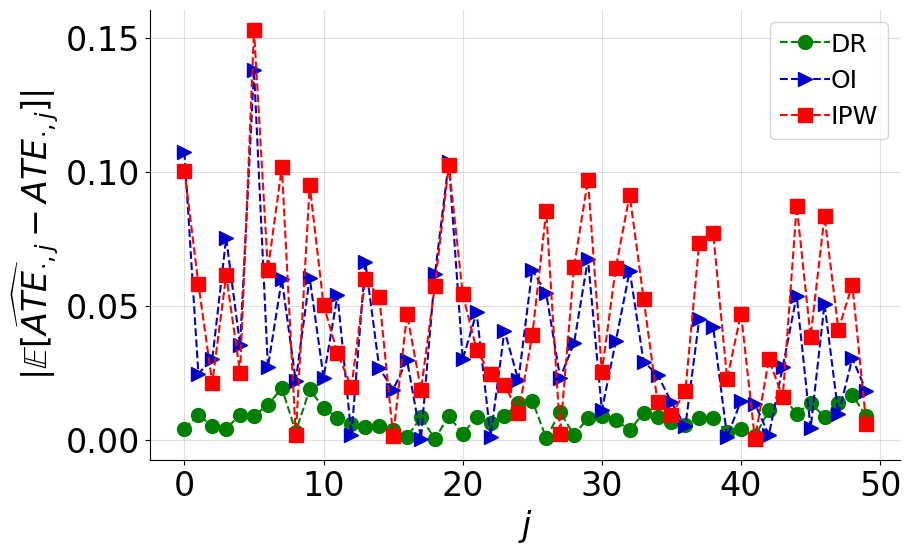} &
    \includegraphics[width=0.45\linewidth,clip]{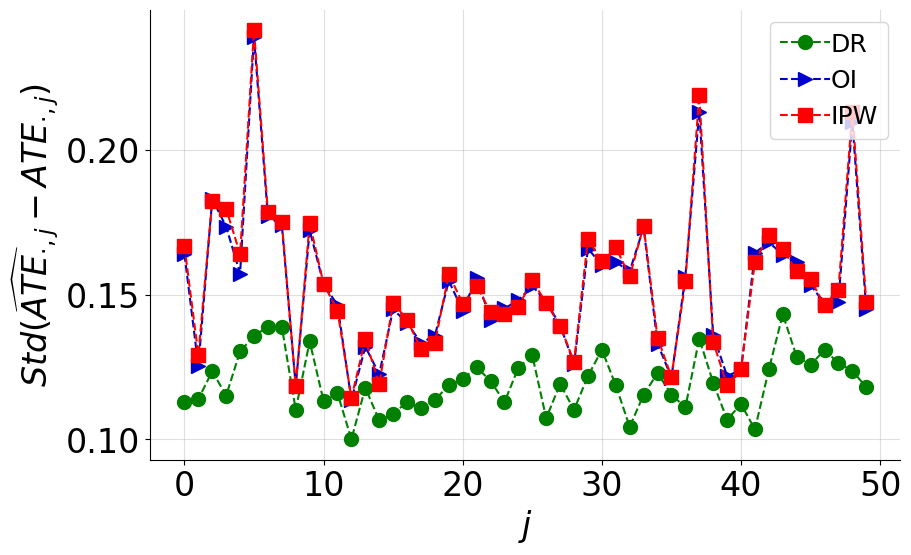}\\
    \multicolumn{2}{c}{$(c)$ $r_p = 5 $, $ r_{\theta} = 3$} \\
    \end{tabular}
\caption{Empirical illustration of the biases and the standard deviations of DR, OI, and IPW estimators for different $j$, and for different $r_p = 5 $ and $ r_{\theta}$.
}
\label{figure_thm2_all}
\end{figure}
In our simulations, we set $\lambda = 0.05$, $c^{(0)} = 1$ and $c^{(1)} = 2$. In practice, instead of choosing the values of $\sigma_{i,j}^{(a)}$ as ex-ante inputs, we make them equal to the standard deviation of all the entries in $\mpo$ for every $i$ and $j$, separately for  $a \in \normalbraces{0,1}$. 

In \cref{figure_thm2_all}, we compare the absolute biases and the standard deviations of OI, IPW, and DR across the first 50 values of $j$ for $N = 1000$, with $r_p = 3$, $ r_{\theta} = 3$ in Panel $(a)$, $r_p = 3$, $ r_{\theta} = 5$ in Panel $(b)$, and $r_p = 5$, $ r_{\theta} = 3$ in Panel $(c)$. For each $j$, the estimate of the biases of OI, IPW, and DR is the average of $\ATEOI-\ATETrue$, $\ATEIPW-\ATETrue$ and $\ATEDR - \ATETrue$ across the $Q$ simulation instances. Likewise, the estimate of the standard deviation of OI, IPW, and DR is the standard deviation of $\ATEOI-\ATETrue$, $\ATEIPW-\ATETrue$ and $\ATEDR - \ATETrue$ across the $Q$ simulation instances. The DR estimator consistently outperforms the OI and IPW estimators in reducing both absolute biases and standard deviations.

\section{Doubly-robust estimation in panel data with lagged effects}
\label{sec:app_dynamic}
\newcommand{\T}{T}
\newcommand{\ra}[1][]{\Rcol(\what\alpha^{#1})}
\newcommand{\cdr}{C_{\mrm{DR}}}
\newcommand{\J}{J}
\newcommand{\est}[1][a]{\what{\mu}_{\cdot, T, J}^{[#1, \mathrm{DR}]}}
\newcommand{\termdefault}{\mbb T}
\newcommand{\tone}[1][a]{\termdefault_{\J}^{(#1)}}
\newcommand{\ttwo}[1][a]{\mbb U_{\J}^{(#1)}}
\newcommand{\tthree}[1][a]{\termdefault_{\J}^{(#1)}}
\newcommand{\ttone}[1][a]{\wtil{\termdefault}_{\J}^{(#1)}}
\newcommand{\htone}[1][a]{\what{\termdefault}_{\J}^{(#1)}}
\newcommand{\etthree}[1][a]{\what{\termdefault}_{\J}^{(#1)}}
\newcommand{\residual}[1][a]{\mbb{V}^{(#1)}}
\newcommand{\eresidual}[1][a]{\what{\mbb{V}}^{(#1)}}

This section describes how the doubly-robust framework of this article can be generalized to a panel data setting with lagged treatment effects. We highlight that, as is the convention in a panel data setting, $t$ denotes the column (time) index and $T$ denotes the total number of columns (time periods). 
\subsection{Setup}
As described in \cref{extension_delayed}, potential outcomes are generated as follows: for all $i \in [N], t \in [T]$, and $a \in \normalbraces{0,1}$, 
\begin{align}
y_{i,t}^{(a|y_{i,t-1})} = \alpha^{(a)} y_{i,t-1} + \theta_{i, t}^{(a)} 
+ \vareps_{i,t}^{(a)},
\label{equation:dynamic}
\end{align}
where $y_{i,t}^{(a|y_{i,t-1})}$ is the potential outcome for unit $i$ at time $t$ given treatment $a\in \{0,1\}$ and lagged outcome $y_{i, t-1}$. This model combines unobserved confounding and lagged treatment effects, where the lagged effect is carried over via the auto-regressive term, $\alpha^{(a)} y_{i,t-1}$, with  $\alpha^{(a)}$ being the auto-regressive parameter for treatment $a\in \{0,1\}$. The treatment possibly starts at $t=1$, and $y_{i,0}$ is assumed to not be affected by any future exposure to the treatment.
Treatment assignments are continually assumed to be generated via \cref{eq_treatment_model}.
As in \cref{eq_consistency}, realized outcomes, $y_{i,t}$, depend on potential outcomes and treatment assignments,
\begin{align}
\label{eq:dynamic_obs}
    y_{i,t} & = y_{i,t}^{(0|y_{i,t-1})} (1-a_{i,t}) + y_{i,t}^{(1|y_{i,t-1})}  a_{i,t},
\end{align}
for all $i \in [N] \stext{and} t \in [T]$. 

\subsection{Target causal estimand}
The lagged effects in \cref{equation:dynamic} imply that the treatment effects need to be defined for sequences of treatments. For concreteness, consider the effect at time $T$ for an always-treat policy, i.e., $a_{i, t}= 1$, versus never-treat, i.e., $a_{i, t}= 0$, for $i \in [N] \stext{and} j \in [T]$. Let $y_{i, \T}^{[1]}$ be the potential outcome for unit $i$ at time $\T$ under always-treat and  $y_{i, \T}^{[0]}$ be the potential outcome for unit $i$ at time $\T$ under never-treat. We aim to estimate the difference in the expected potential outcomes under these two treatment policies averaged over all units,
\begin{gather}
    {\ATETrue[T]} \defeq \muexpected[1] - \muexpected[0], 
    \label{eq:ate_t}
\shortintertext{where} 
    \muexpected[a] \defeq \frac{1}{N}\sum_{i \in [N]} \E[y_{i,\T}^{[a]}],
\end{gather}
with the expectation taken over the distribution of {$\{\varepsilon_{i, t}^{(a)}\}_{i \in [N], t\in [T]}$}, conditioned on the initial outcomes $\sbraces{y_{i, 0}}_{i\in[N]}$. We make the following assumption about the noise in potential outcomes.
\begin{assumption}[Zero-mean noise conditioned on the initial outcomes]
\label{assum:mean_zero_init}
     $\sbraces{{\varepsilon_{i,t}^{(a)}} : i\in[N], t\in [T], a \in \normalbraces{0,1}}$ are mean zero conditioned on $\sbraces{y_{i, 0}}_{i\in[N]}$.
\end{assumption}
\cref{assum:mean_zero_init} holds whenever \cref{assumption_noise}\cref{item_ass_2aa} holds conditioned on the initial outcomes $\sbraces{y_{i, 0}}_{i\in[N]}$. 
{Another sufficient condition for \cref{assum:mean_zero_init} is that 
$(\varepsilon_{i, t}^{(0)},\varepsilon_{i, t}^{(1)})$ are independent in time.} %
Given this, the time dependence in the expected potential outcome $\E[y_{i,\T}^{[a]}]$ is captured as follows: for $a \in \normalbraces{0,1}$
\begin{align}
\E[y_{i, \T}^{[a]}] = (\alpha^{(a)})^{\T}y_{i,0} + \sum_{s=0}^{\T-1} (\alpha^{(a)})^s\theta_{i, \T-s}^{(a)} .
\label{equation:ar_factors}
\end{align}
\cref{equation:ar_factors} forms the basis of our doubly-robust estimator of $\ATETrue[T]$. 

We chose the contrast between always-treat and never-treat for concreteness. However, the framework and the results in this section can be generalized in a straightforward manner to contrast any two pre-specified sequences of treatments, where the treatment can also be chosen stochastically with pre-specified probabilities. For the remainder of this section, we condition on the initial outcomes $\sbraces{y_{i, 0}}_{i\in[N]}$ but omit it from our notation for brevity. 

\subsection{Doubly-robust estimator}
The DR estimator of $\ATETrue[T]$ combines the estimates of $(\alpha^{(0)},\alpha^{(1)})$, $(\mpoz, \mpoo)$, and $\mta$. First, we obtain the estimates $(\what\alpha^{(0)},\what\alpha^{(1)})$. These estimates can be computed using the likelihood approach of \citet{bai2024likelihood} whenever there exists some units such that they all have treatment $a$ for some consecutive time points, for $a \in \normalbraces{0,1}$.

Next, we define the residual matrices $\tooz$ and $\tooo$. Let $\tooz \in \{\Reals\cup \{\star\}\}^{N \times T}$ be a matrix with $(i, t)$-th entry equal to $y_{i, t} - \what \alpha^{(0)} y_{i, t-1}$ if $a_{i, t} = 0$, and equal to $\star$ otherwise. Analogously, let $\tooo  \in \{\Reals\cup \{\star\}\}^{N \times T}$ be a matrix with $(i, t)$-th entry equal to $y_{i, t} - \what \alpha^{(1)} y_{i, t-1}$ if $a_{i, t} = 1$, and equal to $\star$ otherwise. Then, similar to \cref{eq:denoised_estimated}, the application of matrix completion yields the following estimates:
\begin{align}\label{eq:denoised_estimated_dynamic}
\empoz = \mca(\tooz), \quad \empoo = \mca(\tooo), \qtext{and} \emta = \mca(\ta).
\end{align}
Then, the DR estimate is defined as follows:
\begin{align}
 \ATEDR[T,\J] \defeq \est[1] - \est[0] \qtext{where} \est[a] = \frac{1}{N} \sum_{i \in [N]} \brackets{ (\what\alpha^{(a)})^{\T} y_{i, 0} +  \sum_{s=0}^{\J-1} (\what\alpha^{(a)})^s \what{\theta}_{i,\T-s}^{[a,\mathrm{DR}]}}, \label{mu_hat_dynamic}
\end{align}
where
\begin{align}
     \what{\theta}_{i,\T-s}^{[0,\mathrm{DR}]} \defn \what\theta^{(0)}_{i,\T-s} +\big(y_{i, \T-s} -\what\alpha^{(0)} y_{i,\T-s-1}  - \what \theta^{(0)}_{i, \T-s} \big)  \frac{1-a_{i, \T-s}}{1-\what p_{i, \T-s}},
\end{align}
and
\begin{align}
\label{eq_theta_hat_dr_dynamic}
\what{\theta}_{i,\T-s}^{[1,\mathrm{DR}]} \defn \what\theta^{(1)}_{i,\T-s} +\big(y_{i, \T-s}-\what\alpha^{(1)} y_{i,\T-s-1} - \what \theta^{(1)}_{i, \T-s} \big)  \frac{a_{i, \T-s}}{\what p_{i, \T-s}} 
\end{align}
The estimator is parameterized by an integer $\J$, which denotes the contiguous number of time periods preceding time $\T$ that are used to estimate the expectations at time $\T$ (see the summation in \cref{equation:ar_factors}). 
Notably, using preceding $\J$ terms instead of $\T-1$ terms 
allows us to adapt cross-fitting for the setting with lagged treatment effects. Let us briefly elaborate: suppose $(\what\alpha^{(0)},\what\alpha^{(1)})$ are estimated from entries {of $\oo$} in $[N]\times[L]$ for some $L < \T-\J$. Consider the column partitions $\mc C_0 = \sbraces{L+1, \ldots, \T-\J}$ and $\mc C_1 = \sbraces{\T-\J+1, \ldots, \T}$ of times $[T] \setminus [L]$. Suppose \cref{eq_noise_independence_1,eq_noise_independence_2} in \cref{assumption_block_noise} hold for $\cI = \mc R_0 \times \mc C_1$ and $\cI = \mc R_1 \times \mc C_1$ for some row partitions $\mc R_0$ and $\mc R_1$ of units $[N]$. Then, applying \ssMC\ on the residual matrices $\tooz$ and $\tooo$ with row partitions $(\mc R_0, \mc R_1)$ and column partitions $(\mc C_0, \mc C_1)$ ensures that \cref{assumption_estimates} holds for every column in $\mc C_1$ with row partitions $(\mc R_0, \mc R_1)$.

\newcommand{\abound}{\overline{\alpha}}

\subsection{Non-asymptotic guarantees}
Recall the notation for $\RTheta$ and $\RP$ from \cref{eq_notation} and define
\begin{align}
\ra\defn  \sum_{a \in \normalbraces{0,1}} \ra[(a)]  \qtext{where}  \ra[(a)]  \defn |\what\alpha^{(a)} - \alpha^{(a)}|. \label{eq_notation_2}
\end{align}

Our analysis makes two additional assumptions to state a non-asymptotic error bound for $\ATEDR[T,\J] - \ATETrue[T]$.
\begin{assumption}[Bounded auto-regressive parameters and estimates]
    \label{assumption_bounded_ar_params}
    The auto-regressive parameters and their estimates are such that $\normalabs{\alpha^{(a)}} \leq \abound$ and $\normalabs{\what\alpha^{(a)}} \leq \abound$, for all $a\in \sbraces{0, 1}$, where $\abound \in [0, 1)$.
\end{assumption}
\cref{assumption_bounded_ar_params} requires the regression parameters to be bounded by a fixed constant less than $1$. This condition is standard for auto-regressive models, as it implies stability of the outcome process in \cref{equation:dynamic}. The analogous condition on the estimated parameters can be ensured by truncating the estimates to $[0, \abound]$.
\newcommand{\cthree}{C_3}
\newcommand{\ctwo}{C_2}
\begin{assumption}[Bounded observed outcomes, mean potential outcomes, and estimated mean potential outcomes]
\label{assumption_bounded_outcomes}
    The observed outcomes, the mean potential outcomes, and the estimates of the mean potential outcomes are such that $\normalabs{y_{i, t}} \leq C_1$, $\normalabs{\theta^{(a)}_{i, t}} \leq \ctwo$, and $\normalabs{\what\theta^{(a)}_{i, t}} \leq \cthree$, for all 
    $i \in [N]$, $j \in [M]$, and $a\in \sbraces{0, 1}$, where $C_1$, $\ctwo$, and $\cthree$ are universal constants.
\end{assumption}
\cref{assumption_bounded_outcomes} requires the observed outcomes, the mean potential outcomes, and the estimates of the mean potential outcomes to be bounded to simplify our proof.
With a more delicate analysis, \cref{assumption_bounded_outcomes} can be relaxed to require the average observed outcomes over $i\in [N]$, the average mean potential outcomes over $i\in [N]$, and the average estimated mean potential outcomes over $i\in [N]$ to be bounded.

\renewcommand{\thetheorem}{A.\arabic{theorem}}
\setcounter{theorem}{0}

\newcommand{\fsglagged}{Finite Sample Guarantees for DR with lagged effects}
\begin{theorem}[\fsglagged]\label{thm_fsg_dynamic}
Consider the panel data model with lagged effects defined via \cref{equation:dynamic,eq:dynamic_obs}. Suppose 
\cref{assumption_pos,assumption_noise,assumption_pos_estimated,assumption_bounded_ar_params,assumption_bounded_outcomes} hold and \cref{assumption_estimates} holds for $t \in \sbraces{\T-J+1, \ldots, \T}$ for some integer $\J \in [\T]$. Fix $\delta \in (0, 1)$. Then, with probability at least $1-\delta$, we have
\begin{align}
\bigabs{\ATEDR[T,\J]-\ATETrue[T]}  \leq  \frac{\error_{N, \delta/J}}{1-\abound}
+ C \brackets{\frac{\abound^{\J}}{1-\abound} +  \ra \Bigparenth{T\abound^{T-1} + \frac{1}{1-\abound} 
}},
\label{fsg_dr_ate_t_dynamic}
\end{align}
for $ \error_{N, \delta}$ as defined in \cref{eq_combined_error_dr} in \cref{thm_fsg} and a universal constant $C$.
\end{theorem}
The proof of \cref{thm_fsg_dynamic} is given in \cref{sub_proof_of_thm_fsg_dynamic}. For brevity, the finite sample guarantees above uses $\RTheta$ and $\RP$ as defined in \cref{eq_notation}, but the proof can be easily modified to replace the $\max_{j\in [T]}$ appearing in the definition of $\onetwonorm{\cdot}$ in \cref{eq_notation} with $\max_{j\in \sbraces{\T-\J+1,\cdots,\T}}$.

Next, we remark that \cref{thm_fsg_dynamic} is a strict generalization of \cref{thm_fsg}. To this end, note that when $\alpha^{(a)} = 0$ for all $a \in \normalbraces{0,1}$, the 
model considered in \cref{thm_fsg_dynamic}
simplifies to the model considered in \cref{thm_fsg}. For this setting, the assumptions in \cref{thm_fsg} imply that the assumptions in \cref{thm_fsg_dynamic} hold with $\J=1$. First, \cref{assumption_bounded_ar_params} holds with $\abound=0$ when $\alpha^{(a)} = 0$ for all $a \in \normalbraces{0,1}$. Second, the proof of \cref{thm_fsg_dynamic} can be easily modified to drop the requirement of \cref{assumption_bounded_outcomes} when $\J=1$ and $\abound=0$. 
Substituting $\abound=0$, $\ra=0$ (i.e., the auto-regressive parameters are known to be $0$), and $\J=1$ in \cref{fsg_dr_ate_t_dynamic} recovers the guarantee stated in \cref{thm_fsg}. \medskip

\noindent {\bf Doubly-robust behavior of \texorpdfstring{$\ATEDR[T,\J]$}{}.} 
When $\abound  \neq 0$ and bounded away from one, \cref{fsg_dr_ate_t_dynamic} bounds the absolute error of the DR estimator by the rate of
\begin{align}
    \RTheta \Bigparenth{\RP  + \sqrt{\frac{\log \J}{N}}} + \frac{1}{\sqrt N} + \abound^{\J}  + \ra.
\end{align}
Then, if the conditions of \cref{thm_fsg_dynamic} are satisfied for some $\J$ such that $C \log N \geq \J \geq \log N / (2\log (1 / \abound))$,
the error rate of the DR estimator is bounded by $$\RTheta \Bigparenth{\RP  + \sqrt{\frac{\log \log N}{N}}} + \frac{1}{\sqrt N} + \ra,$$ which decays a parametric rate of $O_p(N^{-0.5})$ as long as
\begin{align}
    \RTheta\RP = O_p\Bigparenth{\frac{1}{\sqrt N} }, \quad \RTheta = O_p\Bigparenth{\frac{1}{\sqrt{\log\log N}}},
    \qtext{and}
    \ra = O_{p}\Bigparenth{\frac{1}{\sqrt{N}}}.
\end{align}

Note that \cref{alg_guarantee} still implies that \ssSVD\ achieves $\RP = O_p(N^{-0.5}+\T^{-0.5})$ under suitable conditions. To estimate the auto-regressive parameter $\alpha^{(a)}$ for $a \in \normalbraces{0,1}$, \citet[Section~5]{bai2024likelihood} shows that whenever there exist $K$ units such that they all have treatment $a$ for $L$ consecutive time points, a full information maximum likelihood estimator provides $|\alpha^{(a)}-\what{\alpha}^{(a)}| = O_p((KL)^{-0.5})$. Next, establishing a matrix completion guarantee for the mean potential outcomes by residualizing as in \cref{eq:denoised_estimated_dynamic} can be reduced to deriving a matrix completion guarantee for an approximately low-rank matrix. To this end, \citet[Theorem~5]{agarwal2021causal} suggests that, up to logarithmic factors, an error rate of $N^{-0.5}+ \T^{-0.5} + \ra$ is plausible for $\RTheta$ for our setting. A complete derivation of error guarantees for $\ra$ and $\RTheta$ in the dynamic model is an interesting venue for future work.

\subsection{Proof of \cref{thm_fsg_dynamic}: \fsglagged}
\label{sub_proof_of_thm_fsg_dynamic}
The error $\Delta \ATETrue[T]^{\mathrm{DR}} = \ATEDR[T,\J]-\ATETrue[T]$ can be re-expressed as
    \begin{align}
        \Delta \ATETrue[T]^{\mathrm{DR}} = \bigparenth{\est[1] - \est[0]} - \bigparenth{\muexpected[1] - \muexpected[0]} = \bigparenth{\est[1] - \muexpected[1]} - \bigparenth{\est[0] - \muexpected[0]}. \label{eq_dynamic_error_rearrange}
    \end{align}
We claim that, with probability at least $1-\delta$, 
\begin{align}
     &\abs{\est[1] - \muexpected[1]} \leq  
     C\brackets{
    \frac{|\alpha^{(1)}|^{\J} - |\alpha^{(1)}|^{\T}}{1-|\alpha^{(1)}|} +\ra[(1)] \Bigparenth{T\abound^{T-1} + \frac{1-|\alpha^{(1)}|^{\J}}{1-|\alpha^{(1)}|} + \frac{1}{(1 - |\alpha^{(1)}|)^2} } 
     }\\ 
      &+ \frac{2}{(1-|\alpha^{(1)}|)\lbar}   \bigg[ \Rcol(\empoo)  \RP 
        +  \frac{1}{\sqrt{N}} \Bigparenth{\frac{\sqrt{c \ld[/(12\J)]  }}{\sqrt{\lone}} \Rcol(\empoo) + 2\sigmax \sqrt{c \ld[/(12\J)]  } + \frac{2\sigmax \m(c \ld[/(12\J)] }{\sqrt{\lone}} }  \bigg], \label{eq:dyn_bound_a1}
\end{align}
and 
\begin{align}
    &\abs{\est[0] - \muexpected[0]} \leq  
     C\brackets{
    \frac{|\alpha^{(0)}|^{\J} - |\alpha^{(0)}|^{\T}}{1-|\alpha^{(0)}|} +\ra[(0)] \Bigparenth{T\abound^{T-1} + \frac{1-|\alpha^{(0)}|^{\J}}{1-|\alpha^{(0)}|} + \frac{1}{(1 - |\alpha^{(0)}|)^2} } 
     }\\ 
      &+ \frac{2}{(1-|\alpha^{(0)}|)\lbar}   \bigg[ \Rcol(\empoz)  \RP 
        +  \frac{1}{\sqrt{N}} \Bigparenth{\frac{\sqrt{c \ld[/(12\J)]  }}{\sqrt{\lone}} \Rcol(\empoz) + 2\sigmax \sqrt{c \ld[/(12\J)]  } + \frac{2\sigmax \m(c \ld[/(12\J)] }{\sqrt{\lone}} }  \bigg].
        \label{eq:dyn_bound_a0}
\end{align}
Then, the claim in \cref{fsg_dr_ate_t_dynamic} follows by applying triangle inequality in \cref{eq_dynamic_error_rearrange} and using \cref{assumption_bounded_ar_params}.
We prove the bound \eqref{eq:dyn_bound_a1} in \cref{sub_dynamic_a_1}, and also provide an expression for $C$. The proof of \cref{eq:dyn_bound_a0} follows similarly.
    \subsubsection{Proof of \cref{eq:dyn_bound_a1}}
    \label{sub_dynamic_a_1}
    We start by decomposing $\muexpected[1]$ as follows:
    \begin{align}
    \muexpected[1] = 
        \frac1N \brackets{
        \sum_{i \in [N]} (\alpha^{(1)})^{\T} y_{i, 0}
        + 
        \sum_{s=0}^{\T-1} (\alpha^{(1)})^s \sum_{i \in [N]} 
        \theta^{(1)}_{i, \T-s}}
        &= \tone[1] + \ttwo[1] + \residual[1],
    \end{align}
where 
    \begin{align}
        \tone[1] \defeq \frac1N\sum_{s=0}^{J-1} (\alpha^{(1)})^s \sum_{i \in [N]} 
        \theta^{(1)}_{i, \T-s}, \quad \ttwo[1] \defeq \frac1N\sum_{s=J}^{\T-1} (\alpha^{(1)})^s \sum_{i \in [N]} 
        \theta^{(1)}_{i, \T-s}, \label{def_residual_0}
    \end{align}
and
    \begin{align}
    \residual[1] \defeq   
        (\alpha^{(1)})^{\T} \frac1N \sum_{i \in [N]}  y_{i, 0}.
        \label{def_residual}
    \end{align}
    Next, we decompose $\est[1]$ in \cref{mu_hat_dynamic} as $\est[1] = \htone[1] + \eresidual[1] $, where
    \begin{align}
        \htone[1] \defn \frac1N  \sum_{s=0}^{J-1} (\what\alpha^{(1)})^s  \sum_{i \in [N]} \what{\theta}_{i,\T-s}^{[1,\mathrm{DR}]}, 
        \qtext{and} \eresidual[1] \defeq   
        (\what\alpha^{(1)})^{\T} \frac1N \sum_{i \in [N]}  y_{i, 0}. \label{def_residual_2}
    \end{align}
    Finally, we define 
    \begin{align}
        \ttone[1] &\defeq \frac1N  \sum_{s=0}^{J-1} (\alpha^{(1)})^s  \sum_{i \in [N]} 
        \left[
        \what\theta^{(1)}_{i,\T-s} +\big(y_{i, \T-s}-\alpha^{(1)} y_{i,\T-s-1} - \what \theta^{(1)}_{i, \T-s} \big)  \frac{a_{i, \T-s}}{\what p_{i, \T-s}} 
        \right], \label{def_residual_3}
    \end{align}   
    which is similar to $\htone[1]$ except that $\what{\alpha}^{(1)}$ is replaced by $\alpha^{(1)}$.
    The proof proceeds by bounding each term in the following fundamental decomposition:
    \begin{align}
        \est[1] - \muexpected[1]
        = (\eresidual[1] - \residual[1]) + (\ttone[1] - \tone[1]) + (\htone[1] - \ttone[1]) - \ttwo[1].
        \label{eq:decompose_dyn_a1}
    \end{align}
    With $C_{0} \defeq \max_{i\in[N]}|y_{i, 0}|$ and $\cdr \defeq \cthree + (2C_1+ \cthree)/\lbar$, we claim that the bounds
    \begin{align}
        \bigabs{\eresidual[1] - \residual[1]} \leq C_0 T \ra[(1)] \abound^{T-1}, \qquad |\ttwo[1]| \leq \ctwo   \frac{|\alpha^{(1)}|^{\J}-|\alpha^{(1)}|^{\T}}{1-|\alpha^{(1)}|},
        \label{eq:bound_residual}
    \end{align}
    and
    \begin{align}
        |\htone[1] - \ttone[1]|  \leq  \ra[(1)] \biggparenth{\frac{C_1}{\lambda}  \frac{(1-|\alpha^{(1)}|^{\J})}{1-|\alpha^{(1)}|} + \cdr \frac{1}{(1-|\alpha^{(1)}|)^2}},
        \label{eq:bound_alpha}
    \end{align}
    hold deterministically (conditioned on $\what\alpha^{(1)}$), and that the bound
    \begin{align}
        |\ttone[1] - \tone[1] |  &\leq 
        \frac{2}{(1-|\alpha^{(1)}|)\lbar}   \bigg[ \Rcol(\empoo)  \RP \\
        &\qquad + \Bigparenth{\frac{\sqrt{c \ld[/(12\J)]  }}{\sqrt{\lone}} \Rcol(\empoo) + 2\sigmax \sqrt{c \ld[/(12\J)]  } + \frac{2\sigmax \m(c \ld[/(12\J)] }{\sqrt{\lone}} }  \frac{1}{\sqrt{N}}  \bigg],
        \label{eq:bound_last_dynamic}
    \end{align}
    holds with probability at least $1-\delta/2$.
    The claim in \cref{eq:dyn_bound_a1} follows by
    applying triangle inequality in \cref{eq:decompose_dyn_a1} and using the above bounds.

It remains to establish the intermediate claims~\cref{eq:bound_residual,eq:bound_alpha,eq:bound_last_dynamic}. 
Throughout the rest of the proof, we repeatedly use the inequality below that holds for all $s \in [T]$:
\begin{align}
    \Bigabs{(\what\alpha^{(1)})^{s} - (\alpha^{(1)})^s} = \Bigabs{(\what\alpha^{(1)}-\alpha^{(1)}) \bigparenth{\sum_{l \in [s]} (\what\alpha^{(1)})^{s-l} (\alpha^{(1)})^{l -1}}} 
    & \sless{(a)} s \bigabs{(\what\alpha^{(1)}-\alpha^{(1)})} \abound^{s-1} \\
    & \sequal{(b)} s \ra[(1)] \abound^{s-1}, \label{eq_helper_dynamic}
\end{align}
where $(a)$ follows from \cref{assumption_bounded_ar_params} and $(b)$ follows from \cref{eq_notation_2}. \medskip

\noindent {\bf Proof of \cref{eq:bound_residual}}
First, from \cref{def_residual_0}, we have
\begin{align}
   |\ttwo[1]| = \Bigabs{\frac1N\sum_{s=J}^{\T-1} (\alpha^{(1)})^s \sum_{i \in [N]} 
        \theta^{(1)}_{i, \T-s}} \sless{(a)} \ctwo \sum_{s=J}^{\T-1} \bigabs{\alpha^{(1)}}^s \sequal{(b)} \ctwo \frac{|\alpha^{(1)}|^{\J}-|\alpha^{(1)}|^{\T}}{1-|\alpha^{(1)}|}, 
\end{align}
where $(a)$ follows from \cref{assumption_bounded_outcomes} and $(b)$ follows from the sum of geometric series. Next, from \cref{def_residual,def_residual_2}, we have
\begin{align}
   \bigabs{\eresidual[1] - \residual[1]} = \Bigabs{\Bigparenth{(\what\alpha^{(1)})^{\T} - (\alpha^{(1)})^{\T}} \frac1N \sum_{i \in [N]}  y_{i, 0} }
    \sless{(a)} 
    C_0 T \ra[(1)] \abound^{T-1},
\end{align}
where $(a)$ follows from the definition of $C_0$ and \cref{eq_helper_dynamic}. \medskip

\noindent {\bf Proof of \cref{eq:bound_alpha}}
From  \cref{def_residual_2,def_residual_3}, and the triangle inequality, we have
\begin{align}
    \bigabs{\htone[1] - \ttone[1]} 
    & \leq \frac1N \sum_{i \in [N]} \sum_{s=0}^{\J-1} \bigg | (\what\alpha^{(1)})^s  \Bigparenth{\what\theta^{(1)}_{i,\T-s} + \big(y_{i, \T-s}-\what\alpha^{(1)} y_{i,\T-s-1} - \what\theta^{(1)}_{i, \T-s} \big) \frac{a_{i, \T-s}}{\what p_{i, \T-s}}} \\ 
     &\qquad \qquad \qquad - 
     (\alpha^{(1)})^s  \parenth{\what\theta^{(1)}_{i,\T-s} +  \big(y_{i, \T-s}-\alpha^{(1)} y_{i,\T-s-1} - \what\theta^{(1)}_{i, \T-s} \big) \frac{a_{i, \T-s}}{\what p_{i, \T-s}}} \bigg|\\
     & = \frac1N \sum_{i \in [N]} \sum_{s=0}^{\J-1} \bigg |(\alpha^{(1)})^s (\alpha^{(1)} - \what \alpha^{(1)})  y_{i,\T-s-1} \frac{a_{i, \T-s}}{\what p_{i, \T-s}} \\
     &\qquad+ \Bigparenth{(\what\alpha^{(1)})^s - (\alpha^{(1)})^s } \cdot \Bigparenth{\what\theta^{(1)}_{i,\T-s} + \big(y_{i, \T-s}-\what\alpha^{(1)} y_{i,\T-s-1} - \what\theta^{(1)}_{i, \T-s} \big) \frac{a_{i, \T-s}}{\what p_{i, \T-s}}}
     \bigg|\\
    &\sless{(a)} \frac1N \sum_{i \in [N]} \sum_{s=0}^{\J-1} \Bigabs{\frac{C_1}{\lambda} |\alpha^{(1)}|^s \ra[(1)]  + \cdr s \ra[(1)] \abound^{s-1}}\\
    &= \ra[(1)] \biggparenth{\frac{C_1}{\lambda}  \frac{(1-|\alpha^{(1)}|^{\J})}{1-|\alpha^{(1)}|} + \cdr \frac{1}{(1-|\alpha^{(1)}|)^2}},
\end{align}
where $(a)$ follows from \cref{eq_notation_2}, \cref{assumption_pos_estimated,assumption_bounded_outcomes}, and because $\max_{i \in [N], t \in [T]}\bigabs{\what{\theta}_{i,t}^{[1,\mathrm{DR}]}} \leq \cdr$ from \cref{assumption_bounded_outcomes,assumption_bounded_ar_params,assumption_pos_estimated}, and $(b)$ follows from the sum of geometric and arithmetico-geometric sequences. \medskip

\noindent {\bf Proof of \cref{eq:bound_last_dynamic}}
We start by defining
\begin{align}
    \wtil{\theta}_{i,\T-s}^{[1,\mathrm{DR}]}&\defeq \what\theta^{(1)}_{i,\T-s} +  \big(y_{i, \T-s}-\alpha^{(1)} y_{i,\T-s-1} - \what\theta^{(1)}_{i, \T-s} \big) \frac{a_{i, \T-s}}{\what p_{i, \T-s}}.
    \label{eq:dr_dyn_i}
\end{align}
Then, from \cref{def_residual_0,def_residual_3}, we have
\begin{align}
    |\ttone[1] - \tone[1] | = \biggabs{\sum_{s=0}^{J-1} (\alpha^{(1)})^s 
    \frac1N\sum_{i \in [N]} ( \wtil{\theta}_{i,\T-s}^{[1,\mathrm{DR}]} - \theta_{i, \T-s}^{(1)})
    }
    \sless{(a)} \sum_{s=0}^{J-1} |\alpha^{(1)}|^s \frac1N\Bigabs{\sum_{i \in [N]} ( \wtil{\theta}_{i,\T-s}^{[1,\mathrm{DR}]} - \theta_{i, \T-s}^{(1)} )},
\end{align}
where $(a)$ follows from triangle inequality. From \cref{,eq_treatment_model,equation:dynamic}, we have
\begin{align}
    \wtil{\theta}_{i,\T-s}^{[1,\mathrm{DR}]}- \theta_{i, \T-s}^{(1)} =  \what\theta^{(1)}_{i,\T-s} + (\theta^{(1)}_{i, \T-s} + \vareps^{(1)}_{i, \T-s} - \what\theta^{(1)}_{i, \T-s}) \frac{p_{i, \T-s} + \eta_{i, \T-s}}{\what p_{i, \T-s}} - \theta_{i, \T-s}^{(1)}.
\end{align}
Then, the term $\wtil{\theta}_{i,\T-s}^{[1,\mathrm{DR}]} - \theta_{i, \T-s}^{(1)} $ is analogous to the display \cref{eq:int_error_dr_nondynamic} in the proof of \cref{thm_fsg}. Following similar algebra as in \cref{sec_proof_thm_fsg}, we first obtain
\begin{align}
    \wtil{\theta}_{i,\T-s}^{[1,\mathrm{DR}]} - \theta_{i, \T-s}^{(1)} 
    &=
    \frac{ (\what{\theta}_{i,\T-s}^{(1)} - \theta_{i,\T-s}^{(1)}) (\what{p}_{i,\T-s}-p_{i,\T-s})}{\what{p}_{i,\T-s}} - 
    \frac{ (\what{\theta}_{i,\T-s}^{(1)} - \theta_{i,\T-s}^{(1)}) \eta_{i,\T-s}}{\what{p}_{i,\T-s}} + \frac{\vareps_{i,\T-s}^{(1)}p_{i,\T-s}}{\what{p}_{i,\T-s}} 
    \\
    & \qquad  + \frac{\vareps_{i,\T-s}^{(1)}\eta_{i,\T-s}}{\what{p}_{i,\T-s}}.
\end{align}
Now, note that \cref{assumption_estimates} holds for $j=\T-s$ for all $s \in \sbraces{0, \ldots, \J-1}$. Hence, for any such $s$ and for any $\delta \in (0,1)$, mimicking the derivation of \cref{eq_bound_11_dr} from \cref{sec_proof_thm_fsg}, we obtain, with probability at least $1 - \delta/({2\J})$,
\begin{align}
     \frac{1}{N} \Bigabs{\sum_{i \in [N]} (\wtil{\theta}_{i,\T-s}^{[1,\mathrm{DR}]} - \theta_{i, \T-s}^{(1)} ) } &\leq    \frac{2}{\lbar} \Rcol\bigparenth{\empoo}   \RP + \frac{2\sqrt{c \ld[/(12\J)] }}{\lbar \sqrt{\lone N}}\Rcol\bigparenth{\empoo} + \frac{2\sigmax \sqrt{c \ld[/(12\J)]}}{\lbar \sqrt{N}} + \\
     & \qquad \frac{2 \sigmax \m(c \ld[/(12\J)])}{\lbar \sqrt{\lone N}}. \label{eq_bound_11_dr_dynamic}
\end{align}
Finally, multiplying both sides of \cref{eq_bound_11_dr_dynamic} by $(\alpha^{(1)})^s$, summing it over $s\in\sbraces{0, \ldots, \J-1}$, and using a union bound argument yields that the bound in ~\cref{eq:bound_last_dynamic} holds with probability at least $1-\delta/2$.

\section{Doubly-robust estimation in a panel data model with staggered adoption}
\label{sec_stag_adop_app}
This section shows how to extend the doubly-robust framework of this article to a setting with panel data and staggered adoption. Recall (from \cref{sec:app_dynamic}) that for panel data, $t$ denotes the column (time) index and $T$ denotes the total number of columns (time periods). In a staggered adoption setting, for every unit $i \in [N]$, there exists a time point $t_i \in [T]$ such that $a_{i, t} = 0$ for $t \leq t_i$, and $a_{i, t} = 1$ for $t > t_i$. This defines the observed treatment assignment matrix $\ta$. As mentioned in \cref{sec_stag_adop_short} and illustrated in the example below, a staggered treatment assignment leads to a heavy time-series dependence in $\{\eta_{i, t}\}_{t \in [T]}$.

\begin{example}[Single adoption time]\label{example:staggered_adoption}
    Consider a panel data setting where all units remain in the control group until time $T_0$. At time $T_0 + 1$, each unit $i \in [N]$ receives treatment with probability $p_i$, and remains in treatment until time $T$. With probability $1 - p_i$, each unit $i \in [N]$ stays in the control group until time $T$. In other words, for each unit $i \in [N]$
    \begin{align}
        p_{i,t} & = 0 \qtext{for all} t \leq T_0 \qtext{and} p_{i,t} = p_i \qtext{for all} T_0 < t \leq T.
    \end{align}
    Further, for units remaining in control,
    \begin{align}
        \eta_{i,t} & = 0 \qtext{for all} t \leq T_0 \qtext{and} \eta_{i,t} = -p_i \qtext{for all} T_0 < t \leq T,
    \end{align}
    and for units receiving treatment,
    \begin{align}
        \eta_{i,t} & = 0 \qtext{for all} t \leq T_0 \qtext{and} \eta_{i,t} = 1 - p_i \qtext{for all} T_0 < t \leq T.
    \end{align}
\end{example}
The strong time-series dependence in $\eta_{i,t}$ above implies that \cref{ass_stron_block_factors} or \cref{ass_weak_noise}\cref{item_weak_ass_2aa} do not hold, which in turn implies that the guarantees for $\ssSVD$, as in \cref{alg_guarantee}, may not hold.
To see this, first note that to ensure \cref{assumption_block_noise}, the set of column partitions $\normalbraces{\cC_0, \cC_1}$ must be equal to $\normalbraces{[T_0], [T] \setminus [T_0]}$ due to the dependence in the noise $\noisea$. Now, for \cref{ass_stron_block_factors} to hold, we need $\normalabs{\cC_k}= \Omega(T)$ for every $k \in \normalbraces{0,1}$. However, for \cref{ass_weak_noise}\cref{item_weak_ass_2aa} to hold, we need $T-T_0$ to be a constant with respect to $T$ as, for any $t \in [T] \setminus [T_0]$ and $i \in [N]$,  $\sum_{t' \in [T]} \bigabs{\Expectation\normalbrackets{\eta_{i,t} \eta_{i,t'}}} = (T-T_0) c_i$ where $c_i \in \normalbraces{p_i^2,(1-p_i)^2}$. 

Moreover, in \cref{example:staggered_adoption}, $t_i=T_0$ for all treated units. This allows the choice of $\normalbraces{[T_0], [T] \setminus [T_0]}$ as the set of column partitions $\normalbraces{\cC_0, \cC_1}$ in \cref{assumption_block_noise}. More generally, if treatment adoption times $\normalbraces{t_i}_{i \in [N]}$ differ across units, then it may not be feasible to obtain a partition of $[T]$ into $\normalbraces{\cC_0, \cC_1}$ such that \cref{assumption_block_noise} holds. 

In this section, we propose an alternative approach to the $\ssSVD$ algorithm such that \cref{assumption_estimates} still holds for a suitable staggered adoption model.

\begin{assumption}[Staggered adoption and common unit factors]
\label{defn_stag_adop}
We consider a panel data setting with staggered adoption where 
\begin{enumerate}
    \item all units remain under control till time $T_0$, i.e., for every unit $i \in [N]$, there exists a time point $t_i \geq T_0$ such that $a_{i, t} = 0$ for $t \leq t_i$, and $a_{i, t} = 1$ for $t > t_i$, and 
    \item the unit-dependent latent factors corresponding to $\mta$, $\mpoz$, and $\mpoo$ are the same, i.e., $U = U^{(0)} = U^{(1)} \in \Reals^{N \times r}$. In other words, for every $i\in [N]$ and $t \in [T]$, $p_{i,t} = g(U_i, V_t)$, $\theta^{(0)}_{i,t} = \langle U_i, V_t^{(0)} \rangle$, and $\theta^{(1)}_{i,t} = \langle U_i, V_t^{(1)} \rangle$ for some known function $g: \Reals^{r} \times \Reals^{r} \to \Reals$, with $\langle \cdot,\cdot \rangle$ denoting the inner product.
\end{enumerate}
\end{assumption}

For the setting in \cref{example:staggered_adoption}, the function $g$ corresponds to the inner product, the unit-dependent latent factors are 1-dimensional (i.e., $r = 1$) with $U_i = p_i$ for every $i \in [N]$, and the time-dependent latent factors for the assignment probability are such that $V_t = 0$ for every $t \in [T_0]$ and $V_t = 1$ for every $t \in [T] \setminus [T_0]$. 
Consequently, \cref{example:staggered_adoption} is consistent with \cref{defn_stag_adop} if $U^{(a)}_i = p_i$ for every $a \in \normalbraces{0,1}$ and $i \in [N]$. 
Next, we provide a more flexible version of \cref{example:staggered_adoption} that allows different adoption times for different units.

\begin{example}[Different adoption times]\label{example:staggered_adoption_multiple}
    Consider a panel data setting where all units remain in the control group until time $T_0$. At every time $t \in [T] \setminus [T_0]$, each unit $i \in [N]$ receives treatment with probability $p_i$, and remains in treatment until time $T$. Therefore, for $t \in [T] \setminus [T_0]$ and $i \in [N]$, $a_{i,t} = 1$ if the adoption time point $t_i \in \normalbraces{T_0 + 1, \cdots, t}$, which occurs with probability $\sum_{t' \in [t-T_0-1]} (1-p_i)^{t'-1} p_i$. In other words, for each unit $i \in [N]$,
    \begin{align}
        p_{i,t} & = 0 \qtext{for all} t \leq T_0 \qtext{and} p_{i,t} = 1 - (1 - p_i)^{t - T_0} \qtext{for all} T_0 < t \leq T.
    \end{align}
\end{example}
For the setting in \cref{example:staggered_adoption_multiple}, the unit-dependent latent factors are 1-dimensional (i.e., $r = 1$) with $U_i = p_i$ for every $i \in [N]$, and the time-dependent latent factors for the assignment probability are such that $V_t = 0$ for every $t \in [T_0]$ and $V_t = t - T_0$ for every $t \in [T] \setminus [T_0]$. Further the function $g$ is such that $g(U_i, V_t) = 1 - (1 - U_i)^{V_t}$. Consequently, \cref{example:staggered_adoption_multiple} is consistent with \cref{defn_stag_adop} if $U^{(a)}_i = p_i$ for every $a \in \normalbraces{0,1}$ and $i \in [N]$.

We now describe \cfReg, an algorithm that generates estimates of $(\mpob, \mta)$ for the staggered adoption model in \cref{defn_stag_adop} such that \cref{assumption_estimates} holds.

\begin{enumerate}
\item The inputs are $(i)$ $\ta \in \Reals^{N \times T}$, $(ii)$ $\ooa \in \normalbraces{\Reals \cup \{\star\}}^{N \times T}$ for $a \in \normalbraces{0,1}$,  $(iii)$ the rank $r$ of the unit-dependent latent factors, $(iv)$ the time period $T_0$ until which all units remain under control, $(v)$ the time period $t \in [T] \setminus [T_0]$ for which we want to estimate the average treatment effect, and $(vi)$ the function $g$.
\item Let $\oocontrol \in \mathbb R^{N \times T_0}$ be the sub-matrix of $\ooz$ that keeps the first $T_0$ columns only. Run \texttt{SVD} on $\oocontrol$, i.e.,
    \begin{align}
    \texttt{SVD}(\oocontrol)  = (\what{\matU} \in \Reals^{N \times r}, \what{\Sigma} \in \Reals^{r \times r}, \what{\matV} \in \Reals^{\normalabs{T_0}  \times r}). 
    \end{align}
    \item Let $\cR^{(0)}$ and $\cR^{(1)}$ be the set of units receiving control and treatment at time $t$, respectively. In other words, for every $a \in \normalbraces{0,1}$, $\cR^{(a)} \defn \normalbraces{i \in [N]: a_{i,t} = a}$. Next, randomly partition $\cR^{(a)}$ into two nearly equal parts $\cR^{(a)}_0$ and $\cR^{(a)}_1$. For every $s \in \normalbraces{0,1}$, define $\cR_s = \cR^{(0)}_s \cup \cR^{(1)}_s$.
    \item For every $s \in \normalbraces{0,1}$, regress $\normalbraces{a_{i,t}}_{i \in \cR_s}$ on $\normalbraces{\what{U}_i}_{i \in \cR_s}$ using $g$ to obtain $\what{V}_{1-s}$. For every $s \in \normalbraces{0,1}$ and $i \in \cR_s$, return $\what{p}_{i,t} = g(\what{U}_i, \what{V}_{s})$.
    \item For every $a \in \normalbraces{0,1}$ and $s \in \normalbraces{0,1}$, regress $\normalbraces{y_{i,t}}_{i \in \cR^{(a)}_s}$ on $\normalbraces{\what{U}_i}_{i \in \cR^{(a)}_s}$ to obtain $\what{V}_{1-s}^{(a)}$. For every $a \in \normalbraces{0,1}$, $s \in \normalbraces{0,1}$, and $i \in \cR_s$, return  $\what{\theta}^{(a)}_{i,t} = \what{U}_i \what{V}_{s}^{(a)^\top}$.
\end{enumerate}
In summary, $\cfReg$ estimates the shared unit-dependent latent factors using the observed outcomes for all units until time period $T_0$.  Then, for every $s \in \normalbraces{0,1}$, the time-dependent latent factors $\what{V}_{s}$, $\what{V}_{s}^{(0)}$, and $\what{V}_{s}^{(1)}$ are estimated using the treatment assignments and the observed outcomes for units in $\cR_{1-s}$.

To establish guarantees for $\cfReg$, we adopt the 
subsequent assumption on the noise variables.
\begin{assumption}[Independence across units and with respect to pre-adoption noise]
\label{ass_stag_noise}
{\color{white}.}
\begin{enumerate}[itemsep=1mm, topsep=2mm, label=(\alph*)]
    \item\label{item_stag_2aa} 
    $\normalbraces{(\eta_{i,t}, \varepsilon_{i,t}^{(a)}) : i \in [N]}$ are mutually independent (across $i$) given $\normalbraces{\varepsilon_{i,t}^{(0)}}_{i \in [N], t \in [T_0]}$ for every $t \in [T] \setminus [T_0]$ and $a \in \normalbraces{0,1}$.
    \item\label{item_stag_2bb}
    $\normalbraces{\varepsilon_{i,t}^{(0)}}_{i \in [N], t \in [T_0]}   \indep \normalbraces{\eta_{i,t}, \varepsilon_{i,t}^{(a)}}_{i \in [N]}$ for every $t \in [T] \setminus [T_0]$ and $a \in \normalbraces{0,1}$.
\end{enumerate}
\end{assumption}
\cref{ass_stag_noise}\cref{item_stag_2aa} requires the noise $\normalparenth{\noisey, \noisea}$ corresponding to a time period $t \in T \setminus [T_0]$ to be jointly independent across units given the noise $\noiseyz$ corresponding to time periods $[T_0]$, for every $a \in \normalbraces{0,1}$. \cref{ass_stag_noise}\cref{item_stag_2bb} is satisfied if, for instance, the noise variables follow a moving average model of order $t - T_0 - 1$. The following result, proven in \cref{proof_prop_cfreg}, establishes that the estimates generated by $\cfReg$ satisfy \cref{assumption_estimates}.
\begin{proposition}[Guarantees for $\cfReg$]
\label{prop_cfreg}
Consider the staggered adoption model in \cref{defn_stag_adop} and suppose \cref{ass_stag_noise} holds.
Fix any $t \in [T] \setminus [T_0]$, and  $\normalbraces{\what{\theta}^{(0)}_{i,t}, \what{\theta}^{(1)}_{i,t}, \what{p}_{i,t}}_{i \in [N]}$ be the estimates returned by $\cfReg$. 
Then, \cref{assumption_estimates} holds.
\end{proposition}
Deriving error bounds, i.e., $\RP$ and $\RTheta$, for the estimates generated by $\cfReg$ within the staggered adoption model is an interesting direction for future research.

\subsection{Proof of \texorpdfstring{\cref{prop_cfreg}}{}: Guarantees for \texorpdfstring{\cfReg}{}}
\label{proof_prop_cfreg}
Fix any $s \in \normalbraces{0,1}$. Then, \cref{ass_stag_noise}\cref{item_stag_2aa} and \cref{ass_stag_noise}\cref{item_stag_2bb} imply that 
\begin{align}
\normalbraces{\varepsilon_{i,t}^{(0)}}_{i \in [N], t \in [T_0]} \cup \normalbraces{\eta_{i,t}, \varepsilon_{i,t}^{(a)}}_{i \in \wbar{\cR}_{1-s}}   \indep \normalbraces{\eta_{i,t}, \varepsilon_{i,t}^{(a)}}_{i \in \wbar{\cR}_s}, \label{eq_noise_stag_adoption_proof}
\end{align}
for every partition $(\wbar{\cR}_0, \wbar{\cR}_1)$ of the units $[N]$.
    
Now, $\cfReg$ estimates $\normalbraces{\what{p}_{i,t}}_{i \in \cR_s}$ using $\normalbraces{\what{U}_i}_{i \in \cR_s}$ and $\what{V}_{s}$, where $\what{V}_{s}$ is estimated using $\normalbraces{\what{U}_i}_{i \in \cR_{1-s}}$ and $\normalbraces{a_{i,t}}_{i \in \cR_{1-s}}$. In other words, the randomness in $\normalbraces{\what{p}_{i,t}}_{i \in \cR_s}$ stems from the randomness in $\oocontrol$ and $\normalbraces{a_{i,t}}_{i \in \cR_{1-s}}$ which in turn stems from the randomness in $\normalbraces{\varepsilon_{i,t}^{(0)}}_{i \in [N], t \in [T_0]}$ and $\normalbraces{\eta_{i,t}}_{i \in \cR_{1-s}}$. Then, \cref{eq_independence_requirement_estimates_dr_p} follows from \cref{eq_noise_stag_adoption_proof}.

Next, fix any $a \in \normalbraces{0,1}$. Then, $\cfReg$ estimates $\normalbraces{\what{\theta}^{(a)}}_{i \in \cR_s}$ using $\normalbraces{\what{U}_i}_{i \in \cR_s}$ 
    and $\what{V}_{s}^{(a)}$, where $\what{V}_{s}^{(a)}$ is estimated using $\normalbraces{\what{U}_i}_{i \in \cR^{(a)}_{1-s}}$ and $\normalbraces{y_{i,t}}_{i \in \cR^{(a)}_{1-s}}$. In other words, the randomness in $\normalbraces{\what{\theta}^{(a)}}_{i \in \cR_s}$ stems from the randomness in $\oocontrol$ and $\normalbraces{y_{i,t}}_{i \in \cR^{(a)}_{1-s}}$ which in turn stems from the randomness in $\normalbraces{\varepsilon_{i,t}^{(0)}}_{i \in [N], t \in [T_0]}$ and $\normalbraces{\varepsilon^{(a)}_{i,t}}_{i \in \cR^{(a)}_{1-s}}$. Then, \cref{eq_independence_requirement_estimates_dr} follows from \cref{eq_noise_stag_adoption_proof}.

\bibliography{main}

\end{document}